\RequirePackage{tikz}
\usetikzlibrary{automata,arrows,positioning,calc,shapes.geometric,decorations.pathreplacing}
\tikzset{main node/.style={circle,fill=blue!20,draw,minimum size=1cm,inner sep=0pt},}
\documentclass[sn-mathphys]{sn-jnl}

\usepackage{lmodern}
\usepackage{graphicx}
\usepackage{enumitem}
\usepackage[misc]{ifsym}
\usepackage[textwidth=\linewidth]{todonotes}		

\jyear{2023}%

\usepackage{lineno}
\usepackage{subcaption}

\definecolor{darkblue}{RGB}{0, 0, 125} 
\newcommand\reviewedFirstRev[1]{#1}
\newcommand\reviewed[1]{#1}

\newtheorem{theorem}{Theorem}
\newtheorem{corollary}[theorem]{Corollary}

\newtheorem{lemma}[]{Lemma}

\newtheorem{definition}{Definition}%

\graphicspath{{}{figures/}}

\newcommand{\vr}{\boldsymbol{r}}
\newcommand{\vv}{\boldsymbol{v}}
\newcommand{\vu}{\boldsymbol{u}}
\newcommand{\vh}{\boldsymbol{h}}
\newcommand{\vg}{\boldsymbol{g}}
\newcommand{\vd}{\boldsymbol{d}}
\newcommand{\ve}{\boldsymbol{e}}

\newcommand\vp{\boldsymbol{p}}
\newcommand\vq{\boldsymbol{q}}
\newcommand{\vpi}{\boldsymbol{\pi}}
\newcommand{\vmu}{\boldsymbol{\mu}}
\newcommand{\valpha}{\boldsymbol{\alpha}}
\newcommand{\vdelta}{\boldsymbol{\delta}}

\newcommand{\vone}{\boldsymbol{1}}
\newcommand{\vzero}{\boldsymbol{0}}
\newcommand{\widx}{\lambda}

\newcommand\vy{\boldsymbol{y}}
\newcommand\vz{\boldsymbol{z}}

\newcommand{\ie}{\emph{i.e.}}

\def\gA{{\mathcal{A}}}
\def\gS{{\mathcal{S}}}
\def\gR{{\mathcal{R}}}

\newcommand\mX{\boldsymbol{X}}
\newcommand{\mP}{\boldsymbol{P}}
\newcommand{\mDelta}{\boldsymbol{\Delta}}
\newcommand\mA{\boldsymbol{A}}
\newcommand\mB{\boldsymbol{B}}
\newcommand\mC{\boldsymbol{C}}
\newcommand\mD{\boldsymbol{D}}
\newcommand\mE{\boldsymbol{E}}
\newcommand\mF{\boldsymbol{F}}
\newcommand\mG{\boldsymbol{G}}
\newcommand\mI{\boldsymbol{I}}
\newcommand\mU{\boldsymbol{U}}
\newcommand\mV{\boldsymbol{V}}
\newcommand\mW{\boldsymbol{W}}

\newcommand\bff[1]{\boldsymbol{#1}} 

\newcommand\floor[1]{\left\lfloor#1\right\rfloor}
\DeclareMathOperator*{\argmax}{arg\,max}
\DeclareMathOperator*{\argmin}{arg\,min}

\newcommand\real{\mathbb{R}}
\newcommand\N{\mathbb{N}} 

\begin{document}

\title[Whittle Index]{Testing Indexability and Computing Whittle and Gittins Index in Subcubic Time}


%
%
%
%
%

\author{\fnm{Nicolas} \sur{Gast} }\email{nicolas.gast@inria.fr}
\author{\fnm{Bruno} \sur{Gaujal} }\email{bruno.gaujal@inria.fr}
\author{\Letter\ \fnm{Kimang} \sur{Khun}}\email{khun.kimang@gmail.com}
\affil{\orgname{Univ. Grenoble Alpes, Inria, CNRS, Grenoble INP*, LIG}, \postcode{38000}, \state{Grenoble}, \country{France}}
\affil[*]{\orgname{Institute of Engineering Univ. Grenoble Alpes}}



\abstract{
    Whittle index is a generalization of Gittins index that provides very efficient allocation rules for restless multi-armed bandits. In this work, we develop an algorithm to test the indexability and compute the Whittle indices of any finite-state restless bandit arm. This algorithm works in the discounted and non-discounted cases, and can compute Gittins index. Our algorithm builds on three tools: (1) a careful characterization of Whittle index that allows one to compute recursively the $k$th smallest index from the $(k-1)$th smallest and to test indexability, (2) the use of the Sherman-Morrison formula to make this recursive computation efficient, and (3) a sporadic use of the fastest matrix inversion and multiplication methods to obtain a subcubic complexity. We show that an efficient use of the Sherman-Morrison formula leads to an algorithm that computes Whittle index in $(2/3)n^3 + o(n^3)$ arithmetic operations, where $n$ is the number of states of the arm.  The careful use of fast matrix multiplication leads to the first subcubic algorithm to compute Whittle or Gittins index: By using the current fastest matrix multiplication, the theoretical complexity of our algorithm is $O(n^{2.5286})$. We also develop an efficient implementation of our algorithm that can compute indices of Markov chains with several thousands of states in less than a few seconds.
}




\maketitle

\section{Introduction}
\label{sec1}

Markovian bandits form a subclass of multi-armed bandit problems in which each arm has an internal state that evolves over time in a Markovian manner, as a function of the decision maker's actions.  In such a problem, at each time step, the decision maker observes the state of all arms and chooses which one to activate.
When the state of an arm evolves only when this arm is chosen, one falls into the category of \emph{rested} Markovian bandits for which an optimal policy (in the discounted case) was found by Gittins \citep{gittinsBanditProcessesDynamic1979a}.  When the state of an arm can also evolve when the arm is not chosen, the problem is called a \emph{restless bandit} problem, and computing an optimal policy is computationally difficult \citep{papadimitriou1994complexity}.

In his seminal paper \citep{whittle1988restless}, Whittle proposed a very efficient heuristic: For each arm, an index function maps each state of the arm to a real number. The Whittle index policy then consists in activating the arms having the highest index first. This heuristic generalizes Gittins index to restless bandits. Contrary to the rested case, the Whittle index policy is in general not optimal. Yet, this policy has been proven to be very efficient over the years: up to a condition called \emph{indexability}, Whittle index has been shown to be (in the undiscounted case) asymptotically optimal as the number of arms grows to infinity \reviewedFirstRev{under certain technical assumptions} \citep{verloop2016asymptotically, lott2000optimality, weber1990index}. Moreover, the heuristic performs extremely well in practice \citep{glazebrook2006some, ansell2003whittle, glazebrook2002index}. Restless bandits and Whittle index have been applied to many scheduling and resource allocation problems such as wireless communication \citep{aalto2019whittle, liu2010indexability}, web crawling \citep{avrachenkov2022whittle, nino2014dynamic}, congestion control \citep{avrachenkov2013congestion, avrachenkov2018impulsive}, queueing systems \citep{scully2018soap, aalto2011properties,aalto2009gittins,borkar2017whittle, larranaga2015asymptotically, archibald2009indexability, glazebrook2009index}, and clinical trials \citep{villar2015multi}. 

The above examples show that, when a problem is indexable, computing Whittle index is a very efficient way to construct \reviewedFirstRev{a} nearly-optimal heuristic. This raises a few important questions, that we study in this paper: 
\begin{itemize}
    \item Is testing indexability computationally hard? 
    \item Is there an efficient algorithm to compute Whittle index?
    \item Is Whittle index harder to compute than Gittins index? 
\end{itemize}

\textbf{Related work.}
The computation of Gittins index has received a lot of attention in the past, see for instance \citep{chen1986linear, katehakis1987multi, nino20072, sonin2008generalized} and the recent survey \citep{chakravorty2014multi}. For a $n$-state arm, \reviewedFirstRev{the algorithms having the smallest complexity perform $(2/3)n^3+O(n^2)$ arithmetic operations \citep{chakravorty2014multi}. Note that in page $4$ of \citep{nino2020fast} the author claims that it is unlikely that this complexity can be improved}. As we see later, we do improve upon this complexity. 

Concerning Whittle index, to the best of our knowledge, there are very few efficient general-purpose algorithms to test indexability, see \emph{e.g.} \citep{nino2010characterization}, and most papers studying Whittle index either assume that the studied model is indexable or focus on specific classes of restless bandits for which the structure of arms can be used to show indexability, see \emph{e.g.} \citep{aalto2011properties,akbarzadeh2019restless,akbarzadeh2021maintenance,borkar2017whittle}.  Assuming indexability, the computation of Whittle index has been considered by a few papers.

The most efficient numerical algorithm to compute Whittle index is recently presented in \citep{nino2020fast}. This algorithm, called \emph{fast-pivoting} algorithm,  performs $(2/3)n^3+O(n^2)$ arithmetic operations\footnote{multiplications and additions of real numbers, regardless of their values} if the initialization phase is excluded from the count. This is done by using the parametric simplex method and  exploiting the special structure of this linear system to  reduce the complexity of simplex pivoting steps. This fast-pivoting algorithm is an efficient implementation of adaptive-greedy algorithm \citep{nino2007dynamic}. 
Based on a geometric interpretation of Whittle index, the authors in \citep{akbarzadeh2020conditions} propose a refinement of the adaptive-greedy algorithm of \citep{nino2007dynamic} to compute Whittle indices of all indexable restless bandits. For a $n$-state arm, the refined algorithm achieves a  $O(n^3)$ complexity  by using the  Sherman-Morrison formula. The authors also propose a few checkable conditions to test indexability. However, those conditions are not necessary for indexability, which means that if an arm does not verify the conditions, we cannot conclude that the arm is non-indexable and an algorithm to check indexability is still needed. Also, no detailed description is given for adapting those conditions and their algorithm to restless bandit without discount. A thorough comparison between our algorithm and \citep{akbarzadeh2020conditions,nino2020fast} is given in Appendix~\ref{apx:comparison}.  For continuous-time $n$-state restless bandits, the work of \citep{ayesta2021computation} proposes an algorithm to check indexability and compute Whittle index with a complexity exponential in the number of states $n$ of each arm. According to Remark 4.1 of that paper, this complexity can be reduced to $O(n^5)$ if the restless bandit is known to be indexable and threshold-based policies are optimal. It is stated that their approach is not applicable for discounted restless bandits.

While computing the Whittle indices of a known arm's model is still a challenge, there is interesting work in trying to learn Whittle index when only the arm's simulator is given and the arm's model is unknown. For instance, \citep{gibson2021novel, avrachenkov2022whittle, fu2019towards} use Q-learning algorithm to estimate Whittle index as time evolves in finite-state restless bandits. Moreover,
the work of \citep{nakhleh2021neurwin} uses deep reinforcement learning framework to estimate Whittle indices of the arms with large state space or convoluted transition kernel, assuming a notion of strong indexability. 
For learning aspect, \reviewedFirstRev{the work of \citep{gibson2021novel} shows as to learn Whittle index in non-discounted case by maintaining two Q-functions, updating them using Q-learning algorithm, and deducing Whittle index from them when needed. The way that Whittle indices are computed is very close to our work but less efficient than our algorithm since the authors are more interested in learning the index. 
}
\medskip

\textbf{Contributions.} In this paper, we investigate Whittle index computation in restless multi-armed bandit problems and present four main contributions. 

\reviewedFirstRev{Our first contribution is to discuss the ambiguities in the classical definitions of indexability. Classical definitions assume that an arm is indexable if the optimal policy is a non-decreasing function of some penalty term $\lambda$. While this definition works for most practical cases, it is not always precise enough because the optimal policy is in general not unique. In our definition, we specify the notion of increasingness that should be used. Our definition guarantees the uniqueness of Whittle indices. Note that our definition is the same as the one used in some recent papers (\emph{e.g.}, \cite{nino2020fast}), but the ambiguity of the classical definition seems rarely mentioned.}

Our second contribution is to propose a unified algorithm that computes the Whittle indices for both discounted and non-discounted restless bandits. Our algorithm, which can be viewed as a refinement of the algorithm in \citep{akbarzadeh2020conditions}, tests whether the input arm is indexable or not, and computes Whittle index if the arm is. As a byproduct, our algorithm can compute Gittins index in rested bandits which are a subclass of restless bandits.  This algorithm computes the indices in increasing order, and relies on an efficient use of the Sherman-Morrison formula to compute Whittle index in $(2/3)n^3 + O(n^2)$ plus subcubic time \citep{strassen1969gaussian} to solve a linear system of order $n$. This algorithm can detect on the fly if a computed index violates the indexability condition, which adds an extra $(1/3)n^3+O(n^2)$ arithmetic operations. This later test is optional: the complexity of our algorithm is $n^3+o(n^3)$ when testing indexability and $(2/3)n^3 +o(n^3)$ without the test. These two complexities are comparable to the ones excluding the common initialization phase of reduced-pivoting indexability (RPI) and fast-pivoting adaptive greedy (FPAG) algorithms in \citep{nino2010characterization}. \reviewedFirstRev{For discounted problems, our algorithm works for any finite-state arm.  For non-discounted problems, our algorithm takes as an input any arm and can output three results: the arm is indexable, non-indexable, or multichain. 
We show the correctness of the algorithm which proves that for unichain arms, our algorithm have soundness and completeness properties.
The possible outputs of our algorithm are summarized in Figure~\ref{fig:possible_outputs}.}

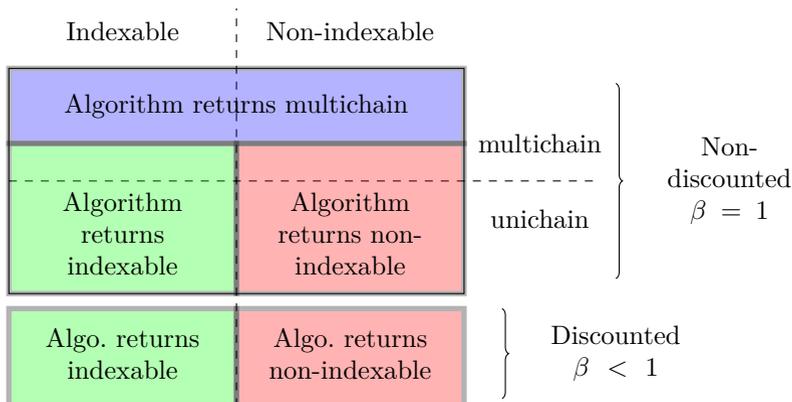
\begin{figure}[ht]
    \centering
    \begin{tikzpicture}
        \draw (0,1) rectangle (6,4);
        \draw[dashed] (3,-.5) -- (3,4.8) (0,2.5) -- (7.75,2.5);
        \node at (1.5,4.5) {Indexable};
        \node at (4.5,4.5) {Non-indexable};
        \node at (7,3) {multichain};
        \node at (7,2) {unichain};
        \draw[line width=2pt, fill=blue,opacity=0.3] (0,3) rectangle (6,4);
        \node at (3,3.5) {Algorithm returns multichain};
        \draw[line width=2pt, fill=green,opacity=0.3] (0,1) rectangle (3,3);
        \node[text width=2cm,align=center] at (1.5,1.8) {Algorithm returns indexable};
        \draw[line width=2pt, fill=red,opacity=0.3] (3,1) rectangle (6,3);
        \node[text width=2cm,align=center] at (4.5,1.8) {Algorithm returns non-indexable};
        \draw[line width=2pt, fill=green,opacity=0.3] (0,-0.5) rectangle (3,0.8);
        \node[text width=2.5cm,align=center] at (1.5,.2) {Algo. returns indexable};
        \draw[line width=2pt, fill=red,opacity=0.3] (3,-0.5) rectangle (6,.8);
        \node[text width=2.5cm,align=center] at (4.5,.2) {Algo. returns non-indexable};

        \draw [decorate,decoration = {brace}] (8,3.8) --  (8,1.2);
        \draw [decorate,decoration = {brace}] (6.5,.8) --  (6.5,-.4);
        \node[text width=2.3cm, align=center] at (9.5,2.5) {Non-discounted\\$\beta=1$};
        \node[text width=3cm, align=center] at (8,0.2) {Discounted\\$\beta<1$};
    \end{tikzpicture}
    \caption{Possible outputs of our algorithm: For unichain or discounted problems, our algorithm tests indexability and returns the index if and only if the problem is indexable. For some multichain problems, the algorithm can test indexability. For the others, it only returns that the problem is multichain.}
    \label{fig:possible_outputs}
\end{figure}


Our third contribution is to show how to reduce the complexity of the above algorithm to obtain the first subcubic algorithm to compute Whittle index. This improvement is made possible by the fact that a linear system can be solved in subcubic time. By carefully reordering the computations, we show that it is possible to reduce the use of the Sherman-Morrison formula at the price of solving more linear systems. The subcubic complexity comes by striking  a good balance between having too many or too few linear systems to solve. By using the current fastest matrix multiplication method, our algorithm can test indexability and compute Whittle index in $O(n^{2.5286})$. Our algorithm is also the first subcubic algorithm to compute Gittins index.

Our fourth and last contribution is to provide an open-source implementation of our algorithm in Python with Numba, and to present an empirical evaluation of the performance of our implementation. Our results show that our algorithm is very efficient in computing Whittle index and testing indexability. Moreover, our simulations indicate that the subcubic version of our algorithm not only has an asymptotically small complexity but is also faster in practice than our original $(2/3)n^3$ algorithm. Testing the indexability and computing indices takes less than one second for $n=1000$ states and less than $10$ minutes for $n=15000$ states. This is $15$ to $20$ times faster than the original computation times reported in \citep{nino2020fast} (for a Matlab implementation), and about $5$ times faster than an optimized implementation of \citep{nino2020fast} (for a Julia implementation).
\medskip

\textbf{Road map.}
The paper is organized as follows. We introduce the problem and the definition of Whittle index in Section~\ref{sec:bandits}. In Section~\ref{sec:widx_compute}, we characterize Whittle index and provide a general idea of how to compute Whittle index. Then, we show, in Section~\ref{sec:formal_widx_algo}, how to use the Sherman-Morrison formula to compute the indices efficiently. We then show how to reduce the complexity of the algorithm by using fast matrix multiplication method in Section~\ref{sec:subcubic algorithm}. We compare the numerical result of different variants of our algorithm in Section~\ref{sec:numerical}. We show how to adapt this approach to the discounted case in Section~\ref{sec:discounted}.  Finally, we conclude in Section~\ref{sec:conclusion}.

\section{Restless bandits and indexability}
\label{sec:bandits}


\subsection{Restless bandit arms and multi-armed bandit}
\label{ssec:def}

In this paper, a restless bandit arm (that we denote later by RB) is a Markov decision process (MDP) with discrete state space $[n]:=\{1,\dots, n\}$ and binary action space $\{0, 1\}$, where $0$ denotes the action ``rest'' and $1$ denotes the action ``activate''. The time is discrete and the evolution is Markovian: If the MDP is in state $i$ and action $a$ is chosen, the decision maker earns an instantaneous reward ${r}^a_i$ and the arm transitions to a new state $j$ with probability ${P}^a_{ij}$. We denote this MDP by the pair $(r,P)$.  The name ``restless'' comes from the fact that an arm put at rest may still transition to a new state. 

A restless multi-armed bandit (RMAB) problem is a finite collection of \reviewedFirstRev{$M\in\N^*$} independent RB arms. At time $t$, the decision maker observes the state of all arms, and can choose \reviewedFirstRev{up to $m$ arms to activate, where $m\le M$ is a fixed constant}. The decision maker then earns a reward that is the sum of the rewards of all arms.
\reviewedFirstRev{The objective of the decision maker is to identify an allocation rule that maximizes the average reward earned over an infinite number of time steps.}
Such a problem is notoriously difficult to solve, as its complexity grows exponentially with the number of arms \cite{papadimitriou1994complexity}.
 \reviewedFirstRev{
In his seminal paper \cite{whittle1988restless}, Whittle proposes the following approach: \textit{if all arms verify a technical condition known as \emph{indexability}, then each state $i$ of each arm is associated with a real number $\widx_{i}$, that is now known as the \emph{Whittle index} of state $i$. At each time, the decision maker activates the $m$ arms whose Whittle index of their current states are the $m$ greatest indices}.  As mentioned earlier, this heuristic performs extremely well in practice, see \emph{e.g.}, \citep{glazebrook2006some, ansell2003whittle, glazebrook2002index}.  This shows that if the $M$ arms are all indexable, then one can derive a very efficient allocation rule for RMAB problems by computing the Whittle indices of all arms. The Whittle indices of an arm do not depend on the other arms. This shows that the computational cost of Whittle index policy is linear in the number of arms multiplied by the time to compute the indices for a single arm. Hence, in the remaining of the paper, we focus on a single arm and present a new algorithm to test indexability and compute the index of a given arm. }

\subsection{Indexability and Whittle index}


\reviewedFirstRev{
    For the remaining of the paper, we consider a single arm $(r, P)$ that has $n$ states. 
    In this section, we introduce the notion of indexability and discuss some ambiguities that we have found when using the definition of indexability defined in previous works.
}

\subsubsection{Policy and arm structure}

\reviewedFirstRev{
    An arm is a two-action MDP. Hence, a policy $\pi$ is a subset of the state space, $\pi\subseteq[n]$, such that the policy chooses to activate the arm in state $i$ if $i\in\pi$. We say that $\pi$ is the set of \emph{active} states, and we say that state $i$ is \emph{passive} if $i\not\in\pi$. By abuse of notation, we will write $\pi_i=1$ if $i\in\pi$ and $\pi_i=0$ if $i\not\in\pi$, and we denote by $\mP^\pi$ the transition matrix corresponding to the policy $\pi$, \emph{i.e.}, $P^\pi_{ij}=P^{\pi_i}_{ij}$.

Following the classical definitions in the literature \cite{putermanMarkovDecisionProcesses1994}, we say that:
\begin{itemize}
    \item A policy $\pi$ is \emph{unichain} if the transition matrix $\mP^\pi$ induced by $\pi$ has a unique recurrent class. A policy that is not unichain is called \emph{multichain}.
    \item An arm is unichain if all policies $\pi\subseteq[n]$ are unichain. An arm is multichain if it is not unichain, \emph{i.e.}, if there exists a policy $\pi\subseteq[n]$ that is multichain.
\end{itemize}
%
}

\subsubsection{Gain optimality and Bellman optimality}
\label{ssec:bellman_optimal}

Following a policy $\pi$, we denote by $g^\pi_i$ the long-run average reward that a decision maker would obtain when starting in state $i$. In the remainder of the paper, we use the term ``gain'' to denote the long-run average reward. Let $g^*_i=\max_\pi g^\pi_i$ be the maximal gain starting from state $i$.
From \citep[Chapter~9]{putermanMarkovDecisionProcesses1994}, $\vg^*$ is uniquely defined. 
We say that a policy $\pi$ is gain optimal if $g^\pi_i=g^*_i$ for all state $i$.

It is shown in \citep[Chapter~9]{putermanMarkovDecisionProcesses1994} that $\vg^*$ is the optimal gain if and only if there exists a vector $\vh^*$, called optimal bias vector that satisfies the Bellman \emph{optimality} equations:  for all $i\in[n]$,
\begin{align}
    g^*_i &= \max_{a\in\{0,1\}} \Bigl(\sum_{j=1}^n P^{a}_{ij}g^*_j\Bigr) \label{eq:gain_opt} \\
    g^*_i + h^*_i &= \max_{a\in\{0,1\}} \Big( r^{a}_i + \sum_{j=1}^n P^{a}_{ij}h^*_j \Big).  \label{eq:bias_opt}
\end{align}
We say that a policy  $\pi$ is \emph{Bellman optimal} if there exists\footnote{If the MDP is unichain, then the bias vector $\vh^*$ is unique up to an additive constant.  This is in general not the case for multichain MDPs. 
} a bias vector $\vh^*$ that is a solution of \eqref{eq:bias_opt} and such that $\pi$ attains the maximum in \eqref{eq:bias_opt}, \emph{i.e.}: for all $i$,
\begin{align}
    \sum_{j=1}^nP^{\pi_i}_{ij}g^*_j=g^*_i \text{ and } \pi_i\in\argmax_{a\in\{0,1\}} \Big(r^{a}_i + \sum_{j=1}^n P^{a}_{ij}h^*_j\Big).
    \label{eq:pi_argmax}
\end{align}
The notion of Bellman optimality is stronger than the notion of gain optimality: A Bellman optimal policy is gain optimal, but the converse is not true in general. Note that the distinction between gain optimal and Bellman optimal policies is only important for the average reward criterion. This distinction disappears for the discounted case that we discuss in Section~\ref{sec:discounted}. The notion of Bellman optimality is equivalent to the notion of canonical optimality, that characterize policies that are optimal for any finite horizon, see \cite{yushkevich1974class}.

\subsubsection{\texorpdfstring{$\lambda$-p}{P}enalized MDP and definition of indexability}
\label{sssec:penal_mdp}

For each $\lambda\in\real$, we define a $\lambda$-penalized MDP\footnote{not to be confused with $\beta$-discounted MDPs, where the discount is on rewards  and not on actions.} whose transition matrices are the same as in the original MDP and whose reward at time $t\ge0$ when taking action $a_t$ in state $s_t$ is $r^{a_t}_{s_t} - \lambda a_t$. The quantity $\lambda$ is a penalty for taking action ``activate''. For $\lambda$-penalized MDPs, we define the gain and bias functions as in Section~\ref{ssec:bellman_optimal}, but these quantities now depend on $\lambda$. Hence, we will write them as functions of $\lambda$: For instance, the optimal gain is $\vg^*(\lambda)$, and we will use the notation $\vh^*(\lambda)$ to denote an optimal bias and $\pi^*(\lambda)$ to denote an optimal policy.


The classical definition of indexability use in the literature \citep{akbarzadeh2020conditions,gibson2021novel,nakhleh2021neurwin} says that an arm is indexable if and only if the optimal policy $\pi^*(\lambda)$ is non-increasing in $\lambda$ (for the inclusion order). If an arm is indexable, these papers define the Whittle index of a state $i$ as a real number $\lambda_i$ such that ${\pi^*(\lambda)=\{i\in[n]: \lambda_i\ > \lambda\}}$.  This definition is ambiguous for two reasons: First, optimal policies are in general not unique. Hence, the notion of $\pi^*(\lambda)$ being non-increasing is unclear: should all optimal policies be non-increasing or at least one? Second, the notion of optimality for a policy  is also unclear: should it mean ``gain optimal'', ``bias optimal'' or another notion of optimality? 

To solve these ambiguities, in this paper, we use the following definition of indexability.  
\begin{definition}
    \label{defn:indexability}
    Given a finite-state arm, let $\Pi^*(\lambda)$ be the set of Bellman optimal policies for a penalty $\lambda$.
    We say that the arm is indexable if for all $\lambda<\lambda'$, and all policies $\pi\in\Pi^*(\lambda)$ and $\pi'\in\Pi^*(\lambda')$, then $\pi\supseteq\pi'$.
\end{definition}
This definition says that the function $\pi^*(\lambda)\supseteq\pi^*(\lambda')$ regardless of the choice of Bellman optimal policies. As we show next, it guarantees that the Whittle indices are uniquely defined when they exist. As we detail in Appendix~\ref{apx:discussion_index}, this is not necessarily the case when we consider other interpretations of the classical definition.

Note that for discounted problems, this definition coincide with the one used in \citep{nino2020fast}. For undiscounted MDPs, we add in addition that the criterion for optimality should be the Bellman optimality.

\subsection{Definition of Whittle index and characterization of indexability}

The proposition below shows that Definition~\ref{defn:indexability} implies that Whittle index is well defined and proposes a characterization of any indexable arm, that we will later use to derive our algorithm. 
\begin{lemma}
    \label{lem:indexable}
    In a $n$-state arm, the following three properties are equivalent: 
    \begin{enumerate}
        \item[(i)] The arm is indexable.
        \item[(ii)] For all state $i\in[n]$, there exists a unique penalty $\lambda_i$ -- called the Whittle index of state $i$ -- such that if $\pi\in\Pi^*(\lambda)$ is any Bellman optimal policy for the penalty $\lambda$, then $\pi_i=1$ if $\lambda<\lambda_i$ and $\pi_i=0$ if $\lambda>\lambda_i$. 
        \item[(iii)] There is a non decreasing sequence of penalties \reviewedFirstRev{$\mu_{\min}^0:=-\infty\le\mu^1_{\min}\le\mu^2_{\min}\le\dots\le\mu^n_{\min}\le\mu^{n+1}_{\min}:=+\infty$} and a sequence of policies $\pi^1:=[n]\supsetneq\pi^2\supsetneq\dots\supsetneq\pi^{n+1}:=\emptyset$ such that:
        \begin{itemize}
            \item If $\lambda\in(\mu^{k-1}_{\min},\mu^{k}_{\min})$, there exists a unique Bellman optimal policy $\pi^{k}$.
            \item If $k$ is such that $\mu^{k-1}_{\min}<\mu^{k}_{\min}$, then all Bellman optimal policies for the penalty $\mu^{k-1}_{\min}$ contain $\pi^{k}$, and $\pi^k$ contains all Bellman optimal policies for the penalty $\mu^k_{\min}$.
        \end{itemize}
    \end{enumerate}
\end{lemma}

In the above lemma, we use a subscript ``min'' in the penalties $\mu^{k}_{\min}$ in order to be consistent with the same quantities used in Algorithm~\ref{algo:whittle_informal} and \ref{algo:whittle_n3}.  The signification of this ``min'' is because it will be a minimum of values of the form $\mu^k_i$.  We should stress that these quantities (as well as the Whittle index $\lambda_i$) can either be finite or infinite. When we say that ``a policy $\pi$ is optimal for the penalty $+\infty$'', this means ``there exists a penalty $\bar{\lambda}$ such that $\pi$ is optimal for all $\lambda\ge\bar{\lambda}$''.  Also, the last part of the lemma implies that $\pi^k$ is the unique Bellman optimal policy for all penalty $\lambda\in(\mu^{k-1}_{\min}, \mu^{k}_{\min})$.

\begin{proof}
    The lemma is a direct consequence of the definition of indexability.

    \reviewedFirstRev{
        $(i)\Rightarrow(ii)$ -- Assume first that the arm is indexable. Let $i\in[n]$ be a state and let $\lambda_i=\sup\{\lambda : \exists\pi\in\Pi^*(\lambda)\text{ such that }\pi_i=0\}$. By Definition~\ref{defn:indexability}, if $\pi'$ is a Bellman optimal policy for a penalty $\lambda>\lambda_i$, then $\pi'\subseteq\pi$, which in turn implies that $\pi'_i=0$. Similarly, if $\lambda<\lambda_i$, then $\pi'_i=1$. This implies (ii).
    
        $(ii)\Rightarrow(iii)$ -- Assume (ii) and let $\sigma^k$ be the state with the $k$th smallest index (where ties are broken arbitrarily). Let $\mu^k_{\min}:=\lambda_{\sigma^k}$ be the index of the state $\sigma^k$ and let $\lambda\in(\mu^{k-1}_{\min},\mu^{k}_{\min}$). By (ii), any Bellman optimal policy for the penalty $\lambda<\lambda_{\sigma^{k-1}}$ contains $\pi^{k}:=[n]\setminus\{\sigma^1,\dots, \sigma^{k-1}\}$. Similarly, $\pi^{k}$ contains any Bellman optimal policy for the penalty $\lambda>\lambda_{\sigma^{k-1}}$.
        This implies that the policy $\pi^{k}$ is the unique optimal policy for all $\lambda\in(\mu^{k-1}_{\min},\mu^{k}_{\min})$. 

        $(iii)\Rightarrow(i)$ -- The property (iii) implies that implies that $\pi^k$ is the unique Bellman optimal policy for all $\lambda\in(\mu^{k-1}_{\min},\mu^k_{\min})$. 
    }
\end{proof}

\section{Condition for indexability and basic algorithm}
\label{sec:widx_compute}

This section aims at providing a basic algorithm to detect whether an arm is indexable or not and if it is the case, to compute the Whittle index of all states.  This algorithm tries to construct a sequence of \emph{unichain} policies $\pi^1\supsetneq \pi^2\supsetneq\dots$ that satisfy the conditions of Lemma~\ref{lem:indexable}.  We prove the correctness of our algorithm: if it can construct such a sequence, then the problem is indexable and the computed indices are correct. If the algorithm cannot compute such a sequence of policies, this is either because the problem is not indexable, or because the arm is multichain.


\subsection{Condition for optimality}

In this section, we provide two technical lemmas that we will use in our algorithm. 
They provide conditions to verify when a given unichain policy is Bellman optimal and if yes, when it is the unique Bellman optimal policy. 

Let $\pi\subseteq[n]$ be a unichain policy and fix a penalty $\lambda$. By \cite[Chapter~8]{putermanMarkovDecisionProcesses1994}, there exists a gain $g^\pi$ and a bias vector $\vh^\pi$ such that $\pi$ satisfies Bellman \emph{evaluation} equations: for all $i\in[n]$,
\begin{equation}
    g^\pi(\lambda) + h^\pi_i(\lambda) = r^{\pi_i}_i -\lambda \pi_i + \sum_{j=1}^n P^{\pi_i}_{ij}h^\pi_j(\lambda).  \label{eq:bias_eval}
\end{equation}
We denote by  $\alpha^{\pi}_i$  the \emph{active advantage} in state $i$ under policy $\pi$, which is the difference between the value in state $i$ of action activate and the one of action rest.
It is defined by:
\begin{equation}
    \label{eq:advantage}
    \alpha^\pi_i(\lambda):=r^1_i -r^0_i -\lambda +\sum_{j=1}^n (P^1_{ij} -P^0_{ij})h^\pi_j(\lambda).
\end{equation}
For a unichain policy, Equation \eqref{eq:bias_eval} uniquely determines the vector $\vh^\pi(\lambda)$ up to an additive constant $c\mathbf{1}$ (see \cite[Chapter 8]{putermanMarkovDecisionProcesses1994}). Hence the active advantage vector $\valpha^\pi(\lambda)$ is uniquely determined for a unichain policy $\pi$. As we will see later, the function $\valpha^\pi(\lambda)$ is affine in $\lambda$.
Note that despite the name ``advantage'', $\valpha^{\pi}$ can be negative.

Our algorithm computes the Whittle index in increasing order, by trying to eliminate states one by one.
The following lemma shows that to compute the next Bellman optimal policy, one should look at when the active advantage of a state is equal to $0$.
In this lemma, $\pi\ominus\{i\}$ denotes the symmetric difference between $\pi$ and $\{i\}$, \emph{i.e.}, $\pi\ominus\{i\}=\pi\setminus\{i\}$ if $i\in\pi$ and $\pi\ominus\{i\}=\pi\cup\{i\}$ if $i\not\in\pi$.
Also, the active advantage provides necessary and sufficient condition for a unichain policy to be Bellman optimal, and/or to be the unique Bellman optimal policy as shown in the following lemma.
\begin{lemma}
    \label{lem:unicity_binary}
    In a finite-state arm, let $\pi$ be a unichain policy. Then, for any penalty $\lambda$:
    \begin{enumerate}[label=(\roman*)]
        \item \label{it:binary_opt1} $\pi$ is Bellman optimal if and only if $\alpha^\pi_i(\lambda)\ge0$ for all $i\in\pi$ and $\alpha^\pi_i(\lambda)\le0$ for all $i\notin\pi$.
        \item \label{it:binary_optX} Suppose that $\pi$ is Bellman optimal and $\alpha^\pi_i(\lambda)=0$. Then, $\pi\ominus\{i\}$ is also Bellman optimal. If, in addition, $\pi\ominus\{i\}$ is unichain, then $\valpha^\pi(\lambda)=\valpha^{\pi\ominus\{i\}}(\lambda)$.
        \item \label{it:binary_opt2} $\pi$ is the unique Bellman optimal policy if and only if $\alpha^\pi_i(\lambda)>0$ for all $i\in\pi$ and $\alpha^\pi_i(\lambda)<0$ for all $i\notin\pi$.
    \end{enumerate}
\end{lemma}
\begin{proof}
    For the first point \ref{it:binary_opt1}, one direction of the equivalence is direct: If policy $\pi$ is Bellman optimal, then $\alpha^\pi_i(\lambda)\ge0$ for all $i\in\pi$ and $\alpha^\pi_i(\lambda)\le0$ for all $i\notin\pi$.
    This is because a bias vector $\vh^\pi$ that is a solution of Bellman evaluation equations \eqref{eq:bias_eval} satisfies Bellman optimality equations \eqref{eq:bias_opt}.

    We now prove the other direction of Point \ref{it:binary_opt1}: If $\alpha^\pi_i(\lambda)\ge0$ for all $i\in\pi$ and $\alpha^\pi_i(\lambda)\le0$ for all $i\notin\pi$, then policy $\pi$ is Bellman optimal.
    Since $\pi$ is unichain, its gain $g^\pi$ is state independent and satisfies the optimality equations \eqref{eq:gain_opt}.
    If $\alpha^\pi_i(\lambda)\ge0$ for all $i\in\pi$ and $\alpha^\pi_i(\lambda)\le0$ for all $i\notin\pi$, then any bias vector $\vh^\pi$ that is a solution of \eqref{eq:bias_eval} also satisfies \eqref{eq:bias_opt}.
    In consequence, $g^\pi\vone$ and $\vh^\pi$ form a solution of the optimality equations \eqref{eq:gain_opt} and \eqref{eq:bias_opt}.
    From \citep[Chapter~9]{putermanMarkovDecisionProcesses1994}, $\vg^*$ is uniquely defined by \eqref{eq:gain_opt} and \eqref{eq:bias_opt}.
    Thus, $g^*_i=g^\pi$ for all $i\in[n]$ and the first condition of Bellman policy characterization equation \eqref{eq:pi_argmax} is satisfied.
    Finally, the fact that bias vector $\vh^\pi$ satisfies \eqref{eq:bias_opt} fulfills the second condition of \eqref{eq:pi_argmax}.
    That concludes the proof.

    For the second point \ref{it:binary_optX}, since $\pi$ is unichain, the optimal gain $g^*_i=g^\pi$ for all $i\in[n]$.
    So, $\pi\ominus\{i\}$ satisfies the first condition of \eqref{eq:pi_argmax}.
    Moreover, $\alpha^\pi_i(\lambda)=0$ implies that policy $\pi\ominus\{i\}$ satisfies evaluation equations \eqref{eq:bias_eval} for some $\vh^\pi$.
    Since $\pi$ is Bellman optimal, $\vh^\pi$ is a solution of $\eqref{eq:bias_opt}$.
    So, $\pi\ominus\{i\}$ satisfies the second condition of \eqref{eq:pi_argmax}.
    We conclude that $\pi\ominus\{i\}$ is Bellman optimal.
    Last but not least, if, in addition, $\pi\ominus\{i\}$ is unichain, then $\vh^\pi$ is a solution of the evaluation equations \eqref{eq:bias_eval} for policy $\pi\ominus\{i\}$. Consequently, $\valpha^\pi(\lambda)=\valpha^{\pi\ominus\{i\}}(\lambda)$.
    
    Points~\ref{it:binary_opt1} and \ref{it:binary_optX} also show one direction of the equivalence of \ref{it:binary_opt2}: If policy $\pi$ is the unique Bellman optimal policy, then for all state $i$, $\alpha^\pi_i(\lambda)\ne0$. The non-trivial property is the other direction of the equivalence. This is a consequence of Lemma~\ref{lem:unicity_BO} that we prove in Appendix~\ref{apx:unicity_BO}.
\end{proof}

Note that the main difficulty in proving Lemma~\ref{lem:unicity_BO} is that we do not assume the arm to be unichain: for a unichain arm, the bias of the optimal policy is unique up to a constant vector (see \cite{schweitzer1978functional} and \cite[Section 8.4]{putermanMarkovDecisionProcesses1994}).
This implies that if $\pi$ is an optimal policy and $\alpha^\pi_i(\lambda)\ne0$ for all $i$, then $\pi$ is the unique optimal policy.
The proof of Lemma~\ref{lem:unicity_BO} that we do in Appendix~\ref{apx:unicity_BO} does not require the MDP to be unichain, but only the Bellman optimal policy $\pi$ to be unichain. \reviewed{Lemma~\ref{lem:unicity_BO} shows that  $\alpha^\pi_i(\lambda)\ne0$ implies that no other policy can be Bellman optimal (not even multichain policies)}. In Appendix~\ref{apx:unicity_BO}, we prove that this holds not only for two-action MDPs but also for any MDP with finite state and action spaces.

\begin{figure}[ht]
    \begin{tabular}{cc}
        \begin{subfigure}[t]{0.48\linewidth}
            \includegraphics[width=\linewidth]{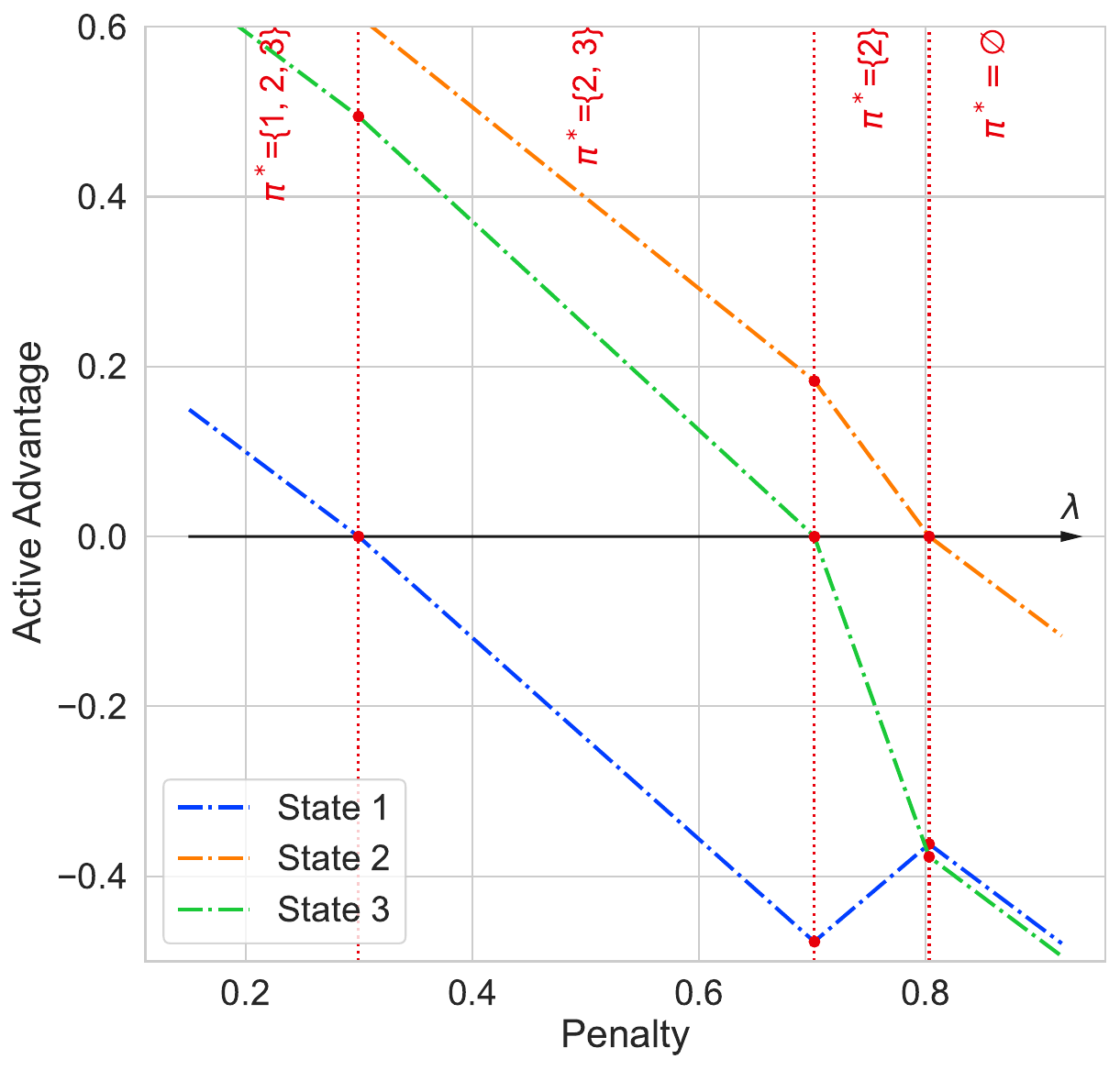}
            \caption{Indexable arm with $3$ states}
            \label{fig:illustrate_vf}
        \end{subfigure}
        &\begin{subfigure}[t]{0.48\linewidth}
            \includegraphics[width=\linewidth]{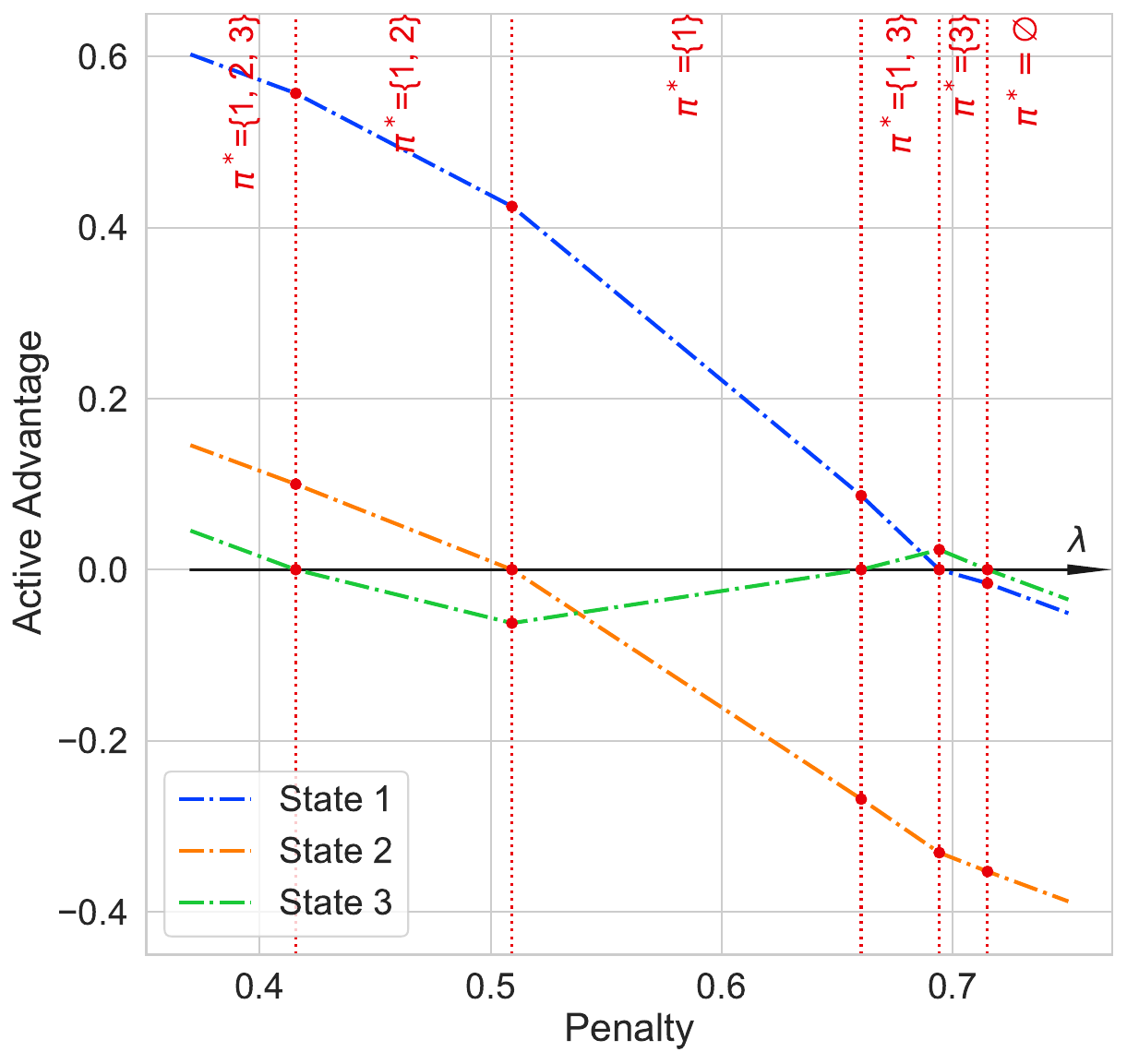}
            \caption{Non-indexable arm with $3$ states}
            \label{fig:illustrate_non_indexable}
        \end{subfigure}            
    \end{tabular}
    \caption{
        The active advantage as a function of penalty for two \reviewedFirstRev{unichain} examples, one is indexable (Figure~\ref{fig:illustrate_vf}) and the other is not (Figure~\ref{fig:illustrate_non_indexable}).
        \reviewedFirstRev{The red dots mark where the lines change their slope. The parameters of both examples are provided in Appendix \ref{apx:non_indexable_example}.}
    }
    \label{fig:illustrate_indexability}
\end{figure}

To illustrate Lemma~\ref{lem:indexable} and Lemma~\ref{lem:unicity_binary}, we consider two three-state arms. For each model, we plot in Figure~\ref{fig:illustrate_indexability} the active advantage $\alpha^{\pi^*(\lambda)}_i(\lambda)$ as a function of the penalty $\lambda$ for all states $i\in\{1,2,3\}$.  By Lemma~\ref{lem:unicity_binary}, we know that an optimal policy should activate all states having a positive advantage and rest all states having a negative advantage. Combined with the characterization of Lemma~\ref{lem:indexable}, this shows that:
\begin{itemize}
    \item The model presented in Figure~\ref{fig:illustrate_vf} is indexable: the optimal policy is a non-increasing function of $\lambda$ and the indices are $\lambda_1\approx0.3$, $\lambda_2\approx0.8$ and $\lambda_3\approx0.7$.
    \item The model presented in Figure~\ref{fig:illustrate_non_indexable} is not indexable: the optimal policy $\pi^*(0.7)=\{3\}$ is not included in $\pi^*(0.6)=\{1\}$. 
\end{itemize}



\subsection{Overview of the algorithm}
\label{ssec:informal_widx_algo}

Our algorithm computes Whittle index in increasing order by navigating through \reviewedFirstRev{unichain Bellman optimal policies} and using the characterization provided by Lemma~\ref{lem:indexable}. It follows the graphical construction given in Figure \ref{fig:illustrate_vf}. It uses the following facts:
\begin{itemize}
    \item If policy $\pi^1:=[n]$ is unichain, then $\valpha^{\pi^1}$ is decreasing in $\lambda$
    \item Similarly, if policy $\pi^{n+1}:=\emptyset$ is unichain, then $\valpha^{\pi^{n+1}}$ is decreasing in $\lambda$
    \item For an indexable arm, computing the index can be done by a greedy algorithm that constructs a sequence of penalties $\mu^1_{\min}\le\mu^2_{\min}\le\dots\le\mu^n_{\min}$ and a sequence of unichain policies $\pi^1\supsetneq\dots\supsetneq\pi^{n}$ by looking at where \reviewedFirstRev{$\alpha^{\pi^k}_i(\lambda)$ intersects horizontal axis for all $i\in\pi^k$}.
    \item The arm is indexable \reviewedFirstRev{if and only if} for all $k$ such that $\mu^{k-1}_{\min}<\mu^{k}_{\min}$, the constructed $\pi^k$ is the largest Bellman optimal policy for the penalty $\mu^k_{\min}$.
\end{itemize}
In order to compute Whittle index and test indexability, our algorithm needs that all policies $\pi^k$ constructed by the algorithm to be unichain. It does not require the arm to be unichain.

This leads to Algorithm~\ref{algo:whittle_informal}, that we write in pseudo-code. 
This algorithm relies on two subroutines: on Line~\ref{algo1:next_mu}, to compute the next index and on Line~\ref{algo1:test1} to test if a policy is Bellman optimal. We will describe later in the paper how to implement these functions in an efficient manner.  Note that in all the paper, we use the superscript $k$ (\emph{e.g.}, $\pi^k, \mu^k, \sigma^k$) to refer to the quantities computed at iteration $k$. We use the subscripts $i$ or $j$ (\emph{e.g.}, $\pi_i, \widx_i, \mu_i, \pi_j$) to refer to the quantities related to states $i$ or $j$.

\begin{algorithm}[ht]
    \caption{Given a $n$-state arm, test indexability and compute Whittle index (if indexable).}\label{algo:whittle_informal}
    \begin{algorithmic}[1]
        \State Set $\pi^1:=[n]$, $\mu^0_{\min}:=-\infty$
        \If{$\pi^1$ is multichain}
            \State \Return{the arm is multichain}
        \EndIf
        \For{$k=1$ to $n$}
            \State Compute $\valpha^{\pi^k}(\lambda)$
            \State Let $\mu_{\min}^k := \inf\{\lambda\ge\mu^{k-1}_{\min} : \exists i\in\pi^k, \alpha^{\pi^k}_i(\lambda)=0\}$ \label{algo1:next_mu}
            \If{$\mu^{k-1}_{\min}{<}\mu^k_{\min}$ and for some $i\notin\pi^k$, $\alpha^{\pi^k}_i(\mu^k_{\min})\ge0$ \label{algo1:test1}}
                \State \Return{the arm is not indexable}
            \EndIf
            \If{$\mu^k_{\min}=+\infty$}
                \State Set $\lambda_i:=+\infty$ for all $i\in\pi^k$
                \State \Return{the arm is indexable and the indices are $\{\lambda_i\}_{i\in[n]}$.}
            \EndIf
            \State Let $\sigma^k\in\pi^k$ be such that $\alpha^{\pi^k}_{\sigma^k}(\mu_{\min}^k)=0$ and $\widx_{\sigma^k}=\mu^k_{\min}$ \label{algo1:sigma^k}
            \State Set $\pi^{k+1} := \pi^{k} \setminus\{\sigma^k\}$ \label{algo1:pik}
            \If{$\pi^{k+1}$ is multichain}
                \State \Return{the arm is multichain}
            \EndIf
        \EndFor
        \State \Return{the arm is indexable and the indices are $\{\lambda_i\}_{i\in[n]}$.}
    \end{algorithmic}
\end{algorithm}

\reviewedFirstRev{
    Note that when $\mu^k_{\min}=+\infty$, the quantity $\alpha^{\pi^k}_i(\mu^k_{\min})$ defined in Line~\ref{algo1:next_mu} of Algorithm~\ref{algo:whittle_informal} should be understood as $\lim_{\lambda\to\infty}\alpha^{\pi^k}_i(\lambda)\in\real\cup\{-\infty,+\infty\}$. These limits are well defined because the functions $\alpha$s are affine in $\lambda$.
}

\begin{figure}[ht]
    \begin{tabular}{cc}
        \begin{subfigure}[t]{0.48\linewidth}
            \includegraphics[width=\linewidth]{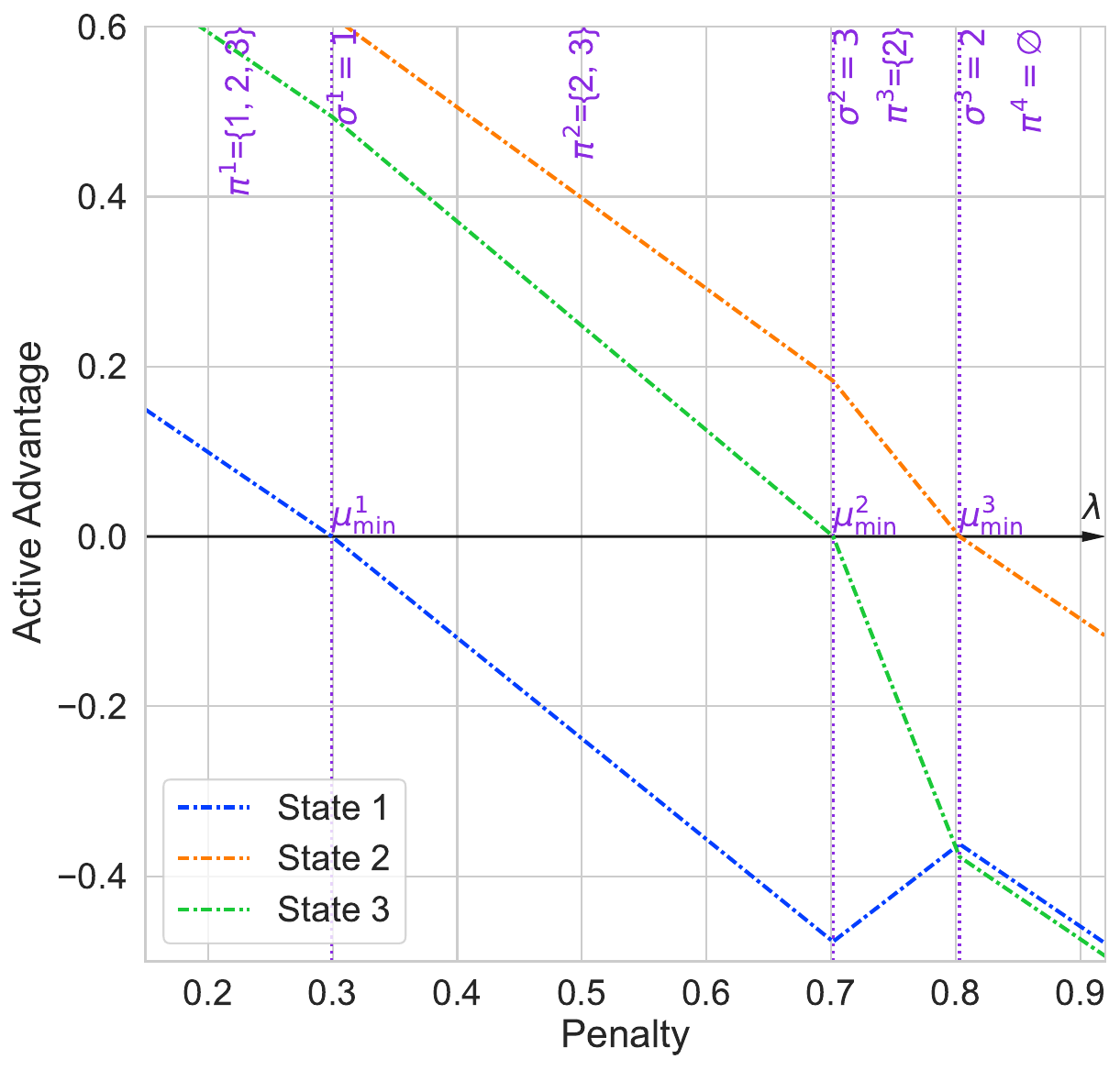}
            \caption{Indexable arm with $3$ states. Note that the algorithm does not compute $\valpha^{\pi^4}$. It checks if policy $\pi^4{=}\emptyset$ is unichain or not. If it is, then $\alpha^{\pi^4}_i$ is decreasing in $\lambda$ for each $i$.}
            \label{fig:illustrate_algo_ind}
        \end{subfigure}
        &\begin{subfigure}[t]{0.48\linewidth}
            \includegraphics[width=\linewidth]{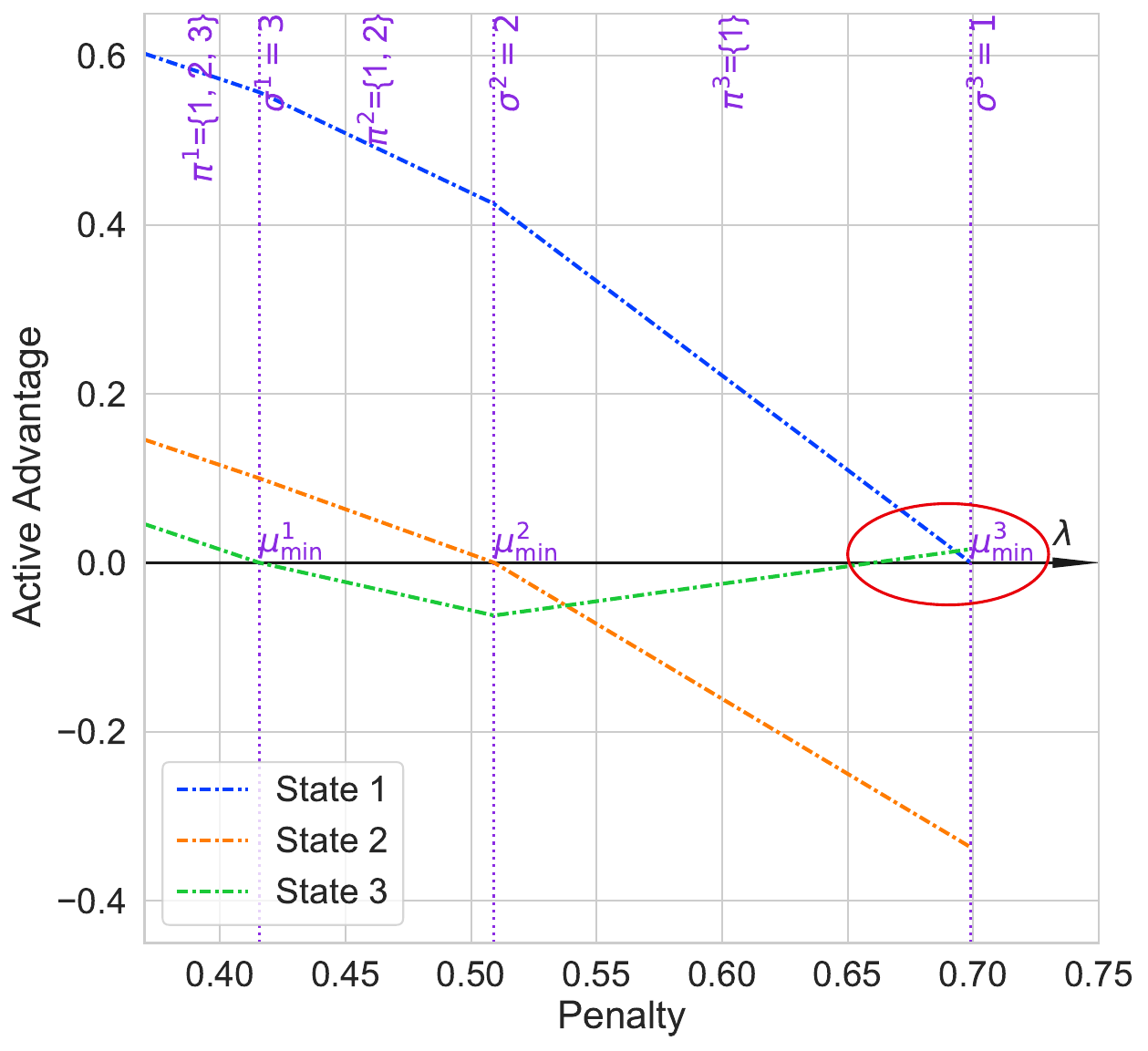}
            \caption{Non-indexable arm with $3$ states. The algorithm stops at iteration $3$ because $\alpha^{\pi^3}_3(\mu^3_{\min})>0$ (the green line at the zone circled with red ellipse).}
            \label{fig:illustrate_algo_nind}
        \end{subfigure}            
    \end{tabular}
    \caption{
        The active advantage $\alpha^{\pi^k}_i(\lambda)$ computed by the algorithm, for the two examples of Figure~\ref{fig:illustrate_indexability}.
    }
    \label{fig:illustrate_algo}
\end{figure}

\medskip

To illustrate how the algorithm works, we plot in Figure~\ref{fig:illustrate_algo} the values computed by the algorithm for the two arms represented in Figure~\ref{fig:illustrate_indexability}. For both models (indexable and non-indexable), the algorithm starts with the policy $[n]$ for which the derivative of the active advantage with respect to $\lambda$ is $-1$ for all states. It then computes $\mu^1_{\min}$ which is the potential index of State~$\sigma^1=1$ for \ref{fig:illustrate_algo_ind} and of State~$\sigma^1=3$ for \ref{fig:illustrate_algo_nind}. The algorithm then moves to iteration 2 and computes $\mu^2_{\min}>\mu^1_{\min}$ for both models and observes that for both models $\alpha^{\pi^2}_{\sigma^1}(\mu^2_{\min})<0$ for both models. Then the algorithm moves to iteration $3$ and computes $\mu^3_{\min}>\mu^2_{\min}$. There are now two cases: 
\begin{itemize}
    \item For \ref{fig:illustrate_algo_ind}, the algorithm verifies that $\pi^4:=\emptyset$ is unichain ($\alpha^{\emptyset}_i$ is decreasing in $\lambda$ for all $i$), terminates, and returns that the arm is indexable.
    \item For \ref{fig:illustrate_algo_nind}, the algorithm realizes that $\alpha^{\pi^3}_{\sigma^1}(\mu^3_{\min})>0$ which shows that this model is not indexable. 
\end{itemize}

\noindent Note that for the indexable example of Figure~\ref{fig:illustrate_algo_ind}, the active advantage function is not a decreasing function of $\lambda$. Hence, this example is neither PCL-indexable (defined in \citep[Definition~3]{nino2020fast}) nor strongly-indexable (defined in \cite{nakhleh2021neurwin}). However, this does not prevent our algorithm from working.

\subsection{Correctness of Algorithm~\ref{algo:whittle_informal}}


The following result shows that Algorithm~\ref{algo:whittle_informal} is correct.

\begin{theorem}
    \label{thm:correct}
    Given a $n$-state arm:
    \begin{enumerate}[label=(\roman*)]
        \item \label{it:idx_proof} if Algorithm~\ref{algo:whittle_informal} outputs ``the arm is indexable and the indices are $\{\lambda_i\}_{i\in[n]}$'', then the arm is indexable and each $\lambda_i$ is the Whittle index of state $i$;
        \item \label{it:non_idx_proof} if Algorithm~\ref{algo:whittle_informal} outputs ``non-indexable'', then the arm is non-indexable;
        \item \label{it:multi_chain} if Algorithm~\ref{algo:whittle_informal} outputs ``multichain'', then the arm is multichain.
    \end{enumerate}
\end{theorem}
A direct consequence of Theorem~\ref{thm:correct} is that for unichain arms, Algorithm~\ref{algo:whittle_informal} provides a full characterization of indexability. 
\begin{corollary}
    \label{coro:correct_unichain}
    Given a unichain arm with finite states, Algorithm~\ref{algo:whittle_informal} outputs ``the arm is indexable'' if and only if it is indexable. 
\end{corollary}

Note that Algorithm~\ref{algo:whittle_informal} does not require the arm to be unichain to work. The required condition is that the Bellman optimal policies $\{\pi^k\}_{k\ge1}$ that Algorithm~\ref{algo:whittle_informal} uses are all unichain. In particular, there exist examples of arms that are multichain and indexable and for which the algorithm returns ``indexable''. Similarly, there exist examples of arms that are multichain and non-indexable and for which the algorithm returns ``non-indexable''. \reviewed{We provide such examples in Appendix~\ref{apx:multichain2}.}

Our algorithm works by exploring solely unichain policies. It does so because (in general) the bias of multichain policy is not unique. The characterization of Bellman optimal policies is much more difficult for multichain models, and the notion of indexability becomes more elusive (see Example \ref{fig:ambiguous_example} in Appendix \ref{apx:discussion_index}). When Algorithm~\ref{algo:whittle_informal} returns ``multichain'', it means that the algorithm is unable to decide whether the arm is indexable or not (but the algorithm knows that the arm is multichain because it has just found a multichain policy).


\begin{proof}[Proof of Theorem~\ref{thm:correct}]
    \textbf{Proof of \ref{it:idx_proof}} -- We first prove by induction on $k$ that:\\
    \fbox{\parbox{.98\linewidth}{If Algorithm~\ref{algo:whittle_informal} completes iteration $k\ge1$, then $\pi^k$ and $\pi^{k+1}$ are unichain and
    \begin{enumerate}[label=(\Alph*)]
        \item \label{eq:induc1} 
        $\pi^{k}$ is the unique Bellman optimal policy for all $\lambda\in(\mu_{\min}^{k-1}, \mu_{\min}^{k})$;
        \item \label{eq:induc2} $\pi^{k+1}$ is Bellman optimal for $\mu^k_{\min}$ and $\valpha^{\pi^{k+1}}(\mu^{k}_{\min})=\valpha^{\pi^{k}}(\mu^{k}_{\min})$.
    \end{enumerate}
    }}
    Base case $k=1$: As we prove later in \eqref{eq:b_pi_one}, $\frac{\partial \alpha^{\pi^1}_i}{\partial \lambda} =-1$. So, for each $i\in\pi^1$, $\alpha^{\pi^1}_i(\lambda)$ is decreasing in $\lambda$. By definition, $\mu^1_{\min}$ is the smallest $\lambda$ such that one of the $\alpha^{\pi^{1}}_i(\lambda) = 0$. Hence, for all $\lambda<\mu^1_{\min}$: $\alpha^{\pi^{1}}_i(\lambda)>0$. By Lemma~\ref{lem:unicity_binary}, this shows that $\pi^1$ is the unique Bellman optimal policy for all $\lambda\in (\mu^0_{\min}, \mu^1_{\min})$, so \ref{eq:induc1} is true. Moreover, since $\pi^1$ is Bellman optimal for the penalty $\mu^1_{\min}$ and $\pi^2$ is unichain, Lemma~\ref{lem:unicity_binary} implies that $\pi^{2}$ is Bellman optimal for the penalty $\mu^{1}_{\min}$ and $\valpha^{\pi^2}(\mu^1_{\min}) =\valpha^{\pi^1}(\mu^1_{\min})$.
    This shows \ref{eq:induc2}.
    
    \begin{figure}[ht]
        \centering
            \includegraphics[width=.8\linewidth]{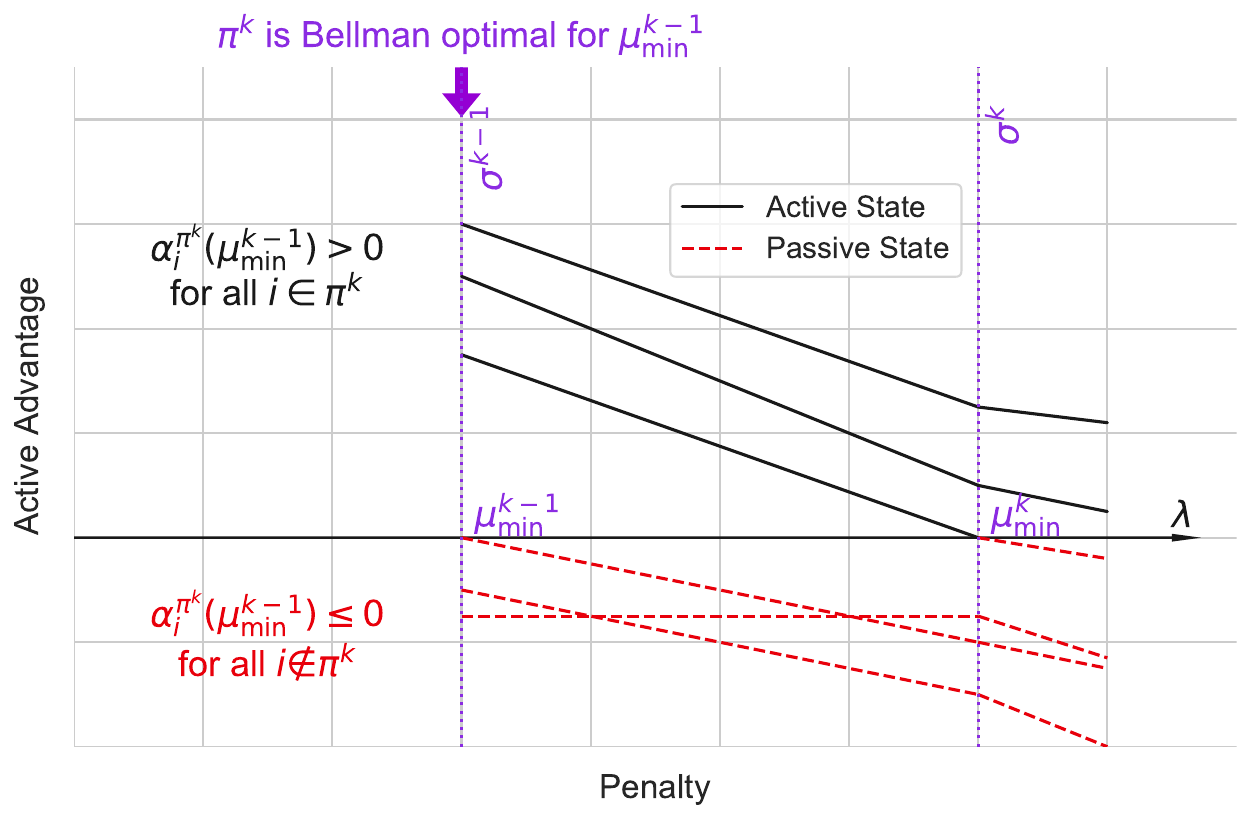}
        \caption{Illustration of what happens when $\mu^{k-1}_{\min}<\mu^k_{\min}$ and when the test of Line~\ref{algo1:test1} is successful (recall that the function $\alpha^{\pi^{k}}_i(\lambda)$ is affine in $\lambda$). The black lines are the advantage functions $\alpha^{\pi^k}_i(\lambda)$ of active state $i\in\pi^k$. The dashed red lines are the advantage functions $\alpha^{\pi^k}_i(\lambda)$ of passive state $i\not\in\pi^k$.}
        \label{fig:proof_algo}
    \end{figure}

    Suppose that the induction is true until iteration $k-1$ and that the algorithm completes iteration $k$.  If $\mu^{k-1}_{\min}=\mu^k_{\min}$, \ref{eq:induc1} is trivial and \ref{eq:induc2} is a direct consequence of the definition of $\mu^k_{\min}$ and of Lemma~\ref{lem:unicity_binary}. Consider now that $\mu^{k-1}_{\min}<\mu^k_{\min}$ and observe what happens in Figure~\ref{fig:proof_algo}. By the induction hypothesis, $\pi^k$ is Bellman optimal for the penalty $\mu^{k-1}_{\min}$. Moreover, by definition of $\mu^k_{\min}$, together with $\mu^{k-1}_{\min}<\mu^k_{\min}$, $\alpha^{\pi^k}_i(\mu^{k-1}_{\min})\ne0$ for all $i\in\pi^k$. Hence:
    \begin{align*}
        \alpha^{\pi^{k}}_i(\mu^{k-1}_{\min})>0 \text{ for $i\in\pi^k$ \qquad and \qquad } \alpha^{\pi^{k}}_i(\mu^{k-1}_{\min})\le0\text{ for $i\not\in\pi^k$}.
    \end{align*}
    Finally, thanks to the test on Line~\ref{algo1:test1} of the algorithm, we have:
    \begin{align}
        \label{eq:test2_algo}
        \alpha^{\pi^{k}}_i(\mu^{k}_{\min})\ge0 \text{ for $i\in\pi^k$ \qquad and \qquad } \alpha^{\pi^{k}}_i(\mu^{k}_{\min})<0\text{ for $i\not\in\pi^k$}.
    \end{align}
    


    In consequence, for each $\lambda\in(\mu^{k-1}_{\min}, \mu^{k}_{\min})$, $\alpha^{\pi^{k}}_i(\lambda) >0$ for $i\in\pi^{k}$ and $\alpha^{\pi^{k}}_i(\lambda) <0$ for $i\notin\pi^{k}$. Lemma~\ref{lem:unicity_binary} implies that $\pi^{k}$ is the unique Bellman optimal policy for each $\lambda\in(\mu^{k-1}_{\min}, \mu^{k}_{\min})$. 
    This shows \ref{eq:induc1}.
    Also, \eqref{eq:test2_algo} implies that $\pi^{k}$ is Bellman optimal for the penalty $\mu^{k}_{\min}$. Combine this with the fact that $\pi^{k+1}$ is unichain, Lemma~\ref{lem:unicity_binary} implies that $\pi^{k+1}$ is Bellman optimal for $\mu^{k}_{\min}$ and $\valpha^{\pi^{k+1}}(\mu^{k}_{\min}) =\valpha^{\pi^{k}}(\mu^{k}_{\min})$.
    This shows \ref{eq:induc2}.
    So, the induction is also true for iteration $k$.

    This shows that the induction property is true for all $k\in[n]$.
    In particular, when $\pi^{n+1}:=\emptyset$ is unichain, $\alpha^{\pi^{n+1}}_i$ is decreasing in $\lambda$ as we prove later in \eqref{eq:b_pi_n1} that $\frac{\partial \alpha^{\pi^{n+1}}_i}{\partial \lambda} =-1$.
    Combine this with the fact that $\alpha^{\pi^{n+1}}_i(\mu^{n}_{\min}) =\alpha^{\pi^{n}}_i(\mu^{n}_{\min})\le0$ for all $i$, Lemma~\ref{lem:unicity_binary} implies that $\pi^{n+1}$ is the unique Bellman optimal policy for $\lambda>\mu^n_{\min}$.
    We simply set $\mu^{n+1}_{\min}:=+\infty$.
    To sum up, there are two cases for which the algorithm outputs that the arm is indexable:
    \begin{enumerate}
        \item if the algorithm goes until the end of iteration $n$, then the sequence of values $\{\mu^k_{\min}\}_{k\in[n+1]}$ and of policies $\{\pi^k\}_{k\in[n+1]}$ satisfies the conditions of Lemma~\ref{lem:indexable}(iii) and the arm is indexable.
        \item if the algorithm stops at iteration $k$ because $\mu^k_{\min}=+\infty$, then one can set $\mu^{k+1}_{\min}:=\dots:=\mu^{n+1}_{\min}:=+\infty$ and define a sequence of policies $\pi^{k+1}\supsetneq\dots\supsetneq\pi^{n+1}:=\emptyset$ by eliminating all states of $\pi^k$ in an arbitrary order. These sequences satisfies the conditions of Lemma~\ref{lem:indexable}(iii) and the arm is indexable.
    \end{enumerate}
    
    \medskip

    \textbf{Proof of \ref{it:non_idx_proof}} -- 
    Our algorithm outputs non-indexable if there exists an iteration $k$ and a state $j\notin\pi^k$, such that $\alpha^{\pi^k}_j(\mu^k_{\min})\ge0$ when $\mu^{k-1}_{\min}<\mu^k_{\min}$.  We know that $\pi^k$ is Bellman optimal for $\mu^{k-1}_{\min}$, otherwise the algorithm would have stopped before. Assume that:
    \begin{align}
        \label{eq:contradiction}
        \text{All Bellman optimal policies for any $\lambda\in(\mu^{k-1}_{\min},\mu^k_{\min}]$ are included in $\pi^k$.}
    \end{align}
    We will see that this assumption leads to a contradiction. We distinguish two possibilities: 
    \begin{enumerate}
        \item $\pi^k$ is Bellman optimal for the penalty $\mu^k_{\min}$ -- This implies that $\alpha^{\pi^k}_i(\mu^k_{\min})\le0$ for all $i\not\in\pi^k$ which, together with $\alpha^{\pi^k}_j(\mu^k_{\min})\ge0$, implies that $\alpha^{\pi^k}_j(\mu^k_{\min})=0$. By Lemma~\ref{lem:unicity_binary}, this would imply that $\pi^k\cup\{j\}$ is Bellman optimal for $\mu^k_{\min}$. This is in contradiction with \eqref{eq:contradiction}.
        \item $\pi^k$ is not Bellman optimal for $\mu^k_{\min}$ -- In this case, we denote by $\tilde{\lambda}$ the smallest penalty $\lambda\in[\mu^{k-1}_{\min},\mu^{k}_{\min}]$ such that there exists $\pi\subsetneq\pi^k$ that is Bellman optimal for $\tilde{\lambda}$ (it exists because $\pi^k$ is not Bellman optimal for $\mu^k_{\min}$ and we assumed \eqref{eq:contradiction}).
            By definition of $\tilde{\lambda}$, $\pi$ and $\pi^k$ are both Bellman optimal for the penalty $\tilde{\lambda}$.
            Let $i\in\pi^k\setminus\pi$. By Lemma~\ref{lem:unicity_binary}, this implies that $\alpha^{\pi^k}_i(\tilde{\lambda})=0$. The problem is that by definition, $\mu^k_{\min}$ is the smallest penalty $\lambda$ 
            for which there exists $i\in\pi^k$ such that $\alpha^{\pi^k}_i(\lambda)=0$. This implies that $\tilde{\lambda}=\mu^k_{\min}$ which in turn implies that $\pi^k$ is optimal for $\mu^k_{\min}$. This leads to a contradiction.
    \end{enumerate}
    This shows that neither case 1 nor 2 are possible. So, \eqref{eq:contradiction} cannot be true.
    In consequence, the negation of \eqref{eq:contradiction} is true: there exists $\lambda>\mu^{k-1}_{\min}$ and $\pi\not\subseteq \pi^k$ such that $\pi$ is Bellman optimal for $\lambda$.
    This contradicts Definition~\ref{defn:indexability} and therefore implies that the arm is not indexable.

    \medskip
    \textbf{Proof of \ref{it:multi_chain}} -- if our algorithm outputs multichain, then the arm is multichain.
    This is straightforward based on the definition of multichain MDP.
\end{proof}

We should note that \reviewedFirstRev{by Line~\ref{algo1:sigma^k}, it is possible to have $\mu^k_{\min}=\mu^{k-1}_{\min}$. This happens when several states have the same value of Whittle index. This is not problematic because we are sure that $\sigma^k\neq\sigma^{k-1}$ by Line~\ref{algo1:pik}.}

\reviewedFirstRev{
In the proof of Theorem~\ref{thm:correct}\ref{it:idx_proof}, we showed that when policy $[n]$ is unichain, the function $\alpha^{\pi^1}_{i}(\lambda)$ is decreasing in $\lambda$ which implies that it crosses the line $0$ at some finite value $\mu^1_{i}$. This implies that for an indexable arm, if $[n]$ is unichain then the Whittle index are all strictly larger than $-\infty$. A symmetric argument shows that if policy $\emptyset$ is unichain, then all Whittle index are strictly smaller than $+\infty$. This implies the following result. 
\begin{corollary}
    Given a unichain arm with $n$ states, if the arm is indexable, then the indices of the $n$ states are finite: $\lambda_i\not\in\{-\infty,+\infty\}$ for all $i\in[n]$.
\end{corollary}
This is not necessarily true for multichain arms (see the discussion in Appendix~\ref{apx:multichain}.)
}

\subsection{Naive implementation of Algorithm~\ref{algo:whittle_informal} \texorpdfstring{(in $O(n^4)$)}{}}
\label{ssec:implem}

For a given penalty $\lambda$, \reviewedFirstRev{we consider a policy $\pi$ that is Bellman optimal and unichain. Recall that $g^{*}(\lambda)$ is the maximal gain, and $\vh^{\pi}(\lambda)\in\real^n$ is a solution of \eqref{eq:bias_eval}.
We consider $\vh^{\pi}(\lambda)$ such that $h^\pi_1(\lambda)=0$.
Recall from \eqref{eq:bias_eval} that for all $i\in[n]:$}
\begin{align}
    \label{eq:bellman}
    g^*(\widx) + h^{\pi}_i(\widx) = r_i^{\pi_i} -\widx\pi_i + \sum_{j=1}^nP^{\pi_i}_{ij}h^{\pi}_j(\widx).
\end{align}
The above system is a system of $n+1$ linear equations with $n+1$ variables (the additional equation begins with $h_1^\pi(\lambda)=0$). As $\pi$ is unichain, the maximal gain $g^*(\lambda)$ and $\vh^\pi(\lambda)$ are uniquely determined by the system of linear equations \eqref{eq:bellman}, together with the condition that $h_1^\pi(\lambda)=0$.  Note that in \eqref{eq:bellman} the sum is for $j=1$ to $n$. Since $h_1^\pi(\lambda)=0$, it can be transformed into a sum from $j=2$ to $n$. 

Let us define the vector $\vv^\pi(\lambda):=[ g^*(\lambda)\ h^\pi_2(\lambda)\ \dots\ h^\pi_n(\lambda)]^\top$ which is similar to the vector $\vh^\pi(\lambda)$ in which we replaced $h^\pi_1(\lambda)$ by $g^*(\lambda)$. We can write Equation~\eqref{eq:bellman} under a matrix form as:
\begin{align}
    \mA^{\pi} \vv^\pi(\widx)
    = \vr^{\pi} - \lambda \vpi, \label{eq:pre_lin_eq}
\end{align}
where $\vr^{\pi}$ is the reward vector under $\pi$: $\vr^{\pi}{:=}[r^{\pi_1}_1\ \dots\ r^{\pi_n}_n]^\top$, $\vpi{:=}[\pi_1\ \dots\ \pi_n]^\top$, and $\mA^{\pi}$ is the following square matrix:
\begin{align}
    {\mA^{\pi} := \left[\begin{array}{cccccc}
        1 & & & &\\
        1 & 1 &  &\\
        1 & & \ddots & \\
        1 & &  &1\\
    \end{array}
\right] {-}
        \left[\begin{array}{ccccc}
                0 & P_{12}^{\pi_1} &\dots & P^{\pi_1}_{1n}\\
                0 & P_{22}^{\pi_2} &\dots & P^{\pi_2}_{2n}\\
                  &\vdots&&\vdots\\
                0 & P_{n2}^{\pi_n} &\dots & P^{\pi_n}_{nn}\\
        \end{array}
        \right]}
        {=} {\left[\begin{array}{ccccc}
                1 & - P_{12}^{\pi_1} & \dots & -P_{1n}^{\pi_1}\\
                1 & 1-P_{22}^{\pi_2} & \dots & -P_{2n}^{\pi_2}\\
                \vdots\\
                1 &  -P_{n2}^{\pi_n} & \dots &1-P_{nn}^{\pi_n}
        \end{array}\right]} \label{eq:def_A}
\end{align}
As we show in Lemma~\ref{lem:invertible}, \reviewedFirstRev{the matrix $\mA^{\pi}$ is invertible if and only if policy $\pi$ is unichain}. In consequence, $\vv^{\pi}$ is an affine function of $\lambda$:
\begin{align}
    \label{eq:h is linear}
    \vv^{\pi}(\lambda) = (\mA^{\pi})^{-1}(\vr^{\pi} - \lambda \vpi) =(\mA^{\pi})^{-1}\vr^\pi - \lambda (\mA^{\pi})^{-1}\vpi
\end{align}



For a state $i$, let $\delta_i:=r^1_i - r^0_i$, $\Delta_{i1}:=0$ and $\Delta_{ij}:=P^1_{ij}-P^0_{ij}$ for $j\in\{2,\dots, n\}$. By definition of the advantage function of \eqref{eq:advantage}, we have: 
\begin{align*}
    \valpha^\pi(\lambda)=\vdelta -\lambda\vone +\mDelta\vv^\pi(\lambda).
\end{align*}

For each active state $i$, we want to find the smallest penalty $\mu^k_i\ge\mu^{k-1}_{\min}$ such that $\alpha^{\pi^k}_i(\mu^k_i)=0$.
Suppose that $\pi^{k-1}$ and $\pi^k$ are unichain and $\pi^{k-1}$ is Bellman optimal for $\mu^{k-1}_{\min}$.
By Lemma~\ref{lem:unicity_binary}, $\valpha^{\pi^k}(\mu^{k-1}_{\min})= \valpha^{\pi^{k-1}}(\mu^{k-1}_{\min})$.
Let $\vd^{\pi^k}:=-(\mA^{\pi^k})^{-1}{\vpi^k}$.
By \eqref{eq:h is linear}, $\valpha^{\pi^k}(\lambda)$ is a linear function of $\lambda$ whose derivative is $-(\vone-\mDelta\vd^{\pi^k})$.
In particular, $\valpha^{\pi^k}(\lambda)=\valpha^{\pi^{k-1}}(\mu^{k-1}_{\min}) -(\lambda-\mu^{k-1}_{\min})(\vone -\mDelta\vd^{\pi^k})$.
Thus, $\alpha^{\pi^k}_i(\lambda)=0$ if and only if
\begin{align}
    \label{eq:alpha_null}
    \alpha^{\pi^{k-1}}_i(\mu^{k-1}_{\min}) = (\lambda -\mu^{k-1}_{\min})(1 -\sum_{j=2}^n\Delta_{ij} d_j^{\pi^k}).
\end{align}
Recall that $\alpha^{\pi^{k-1}}_i(\mu^{k-1}_{\min})$ is non-negative for active state $i\in\pi^k$. The value $\mu^k_i$ is the smallest $\lambda\ge\mu^{k-1}_{\min}$ that satisfies Equation~\eqref{eq:alpha_null}. There are three cases: 
\begin{enumerate}
    \item if $\alpha^{\pi^{k-1}}_i(\mu^{k-1}_{\min})=0$, then $\mu^k_i:=\mu^{k-1}_{\min}$;
    \item if $\alpha^{\pi^{k-1}}_i(\mu^{k-1}_{\min})>0$ and
        \begin{enumerate}
            \item if $1 -\sum_{j=2}^n\Delta_{ij} d_j^{\pi^k}>0$, then
            \begin{align}
                \label{eq:mu_i_k_from_d}
                \mu^k_i:=\mu^{k-1}_{\min} +\frac{ \alpha^{\pi^{k-1}}_i(\mu^{k-1}_{\min})}{1 -\sum_{j=2}^n\Delta_{ij} d_j^{\pi^k}};
            \end{align}
            \item if $1 -\sum_{j=2}^n\Delta_{ij} d_j^{\pi^k}\le0$, then $\mu^k_i:=+\infty$.
        \end{enumerate}
\end{enumerate}

This shows that, for a given $k$, computing \reviewedFirstRev{$\mu^k_{\min}$} of Line~\ref{algo1:next_mu} can be done in $O(n^3)$: A first part in  $O(n^3)$ to compute the inverse of matrix $\mA^{\pi}$ and to compute $\vd^\pi$, plus some smaller order terms to compute the solutions of \eqref{eq:alpha_null}.
Similarly, the test in Line~\ref{algo1:test1} of Algorithm~\ref{algo:whittle_informal} can also be implemented in $O(n^3)$
\reviewedFirstRev{
    by using $\valpha^{\pi^k}(\mu^k_{\min}) =\valpha^{\pi^{k-1}}(\mu^{k-1}_{\min}) -(\mu^k_{\min} -\mu^{k-1}_{\min})(\vone -\mDelta\vd^{\pi^k})$ with the convention that when $\mu^k_{\min} = +\infty$, $+\infty\times0=0$ and $+\infty\times x=\mathrm{sign}(x)\infty$ for any $x\ne0$. 
}
This leads to an overall complexity of $O(n^4)$ for Algorithm~\ref{algo:whittle_informal} that contains $n$ loops each having a $O(n^3)$ complexity.
\reviewedFirstRev{
    If at some iteration $k$ the matrix $\mA^{\pi^k}$ is not invertible, then Lemma~\ref{lem:invertible} implies that $\pi^k$ is multichain.
    In consequence, the algorithm outputs multichain and stops.
    We integrate this in the newer version of our algorithm below.
}


\section{The \texorpdfstring{$(2/3)n^3 + o(n^3)$}{2n3/3+o(n3)} algorithm}
\label{sec:formal_widx_algo}

This section describes a way to implement Algorithm~\ref{algo:whittle_informal}  efficiently using  $O(n^3)$ operations.
\reviewedFirstRev{
    The main idea is to use the Sherman-Morrison formula to compute in $O(n^2)$ the active advantage vector $\valpha^{\pi^k}(\lambda)$ associated to $\pi^k$ from the one associated to $\pi^{k-1}$.
}
This leads to a $O(n^3)$ algorithm. Once this main idea is in place, we show how to avoid unnecessary computations to obtain an algorithm that performs $(2/3)n^3+o(n^3)$ arithmetic operations.

\subsection{Additional notations}

In order to obtain a more efficient and compact algorithm, for an iteration $k$ and a state $i$, we define \reviewedFirstRev{
    $y^k_i:=\sum_{j=2}^n\Delta_{ij}d_j^{\pi^k}$ and $z^k_i:=\alpha^{\pi^k}_i(\mu^k_{\min})$, where $d_j^{\pi^k}$, $\Delta_{ij}$ and $\alpha^{\pi^k}_i$ are as in \eqref{eq:mu_i_k_from_d}.
}
Equation~\eqref{eq:mu_i_k_from_d} can be rewritten as
\begin{align}
    \mu_i^k &= \mu^{k-1}_{\min} +\displaystyle\frac{z_i^{k-1}}{1-y_i^k}.
    \label{eq:mu^k_i_from_y}
\end{align}
The above equation can be used to compute $\mu_i^k$ and $\mu^k_{\min}$ easily from $y^k_i$ and $z^{k-1}_i$. 
Indeed, from the previous section, we have $\alpha^{\pi^k}_i(\mu^k_{\min}) = \alpha^{\pi^{k-1}}_i(\mu^{k-1}_{\min}) -(\mu^k_{\min} -\mu^{k-1}_{\min})(1 -\sum_{j=2}^n\Delta_{ij}d_j^{\pi^k})$ which translates into
\begin{align}
    z^k_i &= z^{k-1}_i -(\mu^k_{\min} -\mu^{k-1}_{\min})(1 -y^k_i). \label{eq:z^k-i}
\end{align}
This shows that the critical values to compute are the variables $y^k_i$. In the remainder of this section, we show that the quantity $y^k_i$ can be computed efficiently by a recursive formula.


\subsection{Application of the Sherman-Morrison formula}
\label{ssec:sm_form}

To compute \reviewedFirstRev{$y^{k+1}_i := \sum_{j=2}^n\Delta_{ij}d_j^{\pi^{k+1}}$}, we need to compute the quantities $\vd^{\pi^{k+1}} := -(\mA^{\pi^{k+1}})^{-1}\vpi^{k+1}$.
This requires the inverse of ${\mA^{\pi^{k+1}}}$.
By definition of $\pi^k$, two policies $\pi^k$ and $\pi^{k+1}$ differ by exactly one state: $\vpi^{k+1}=\vpi^k - \ve_{\sigma^k}$ where $\ve_j$ denotes the column vector with a $1$ in $j$th coordinate and $0$'s elsewhere.
Also by definition of $\mA^\pi$ in \eqref{eq:def_A}, the two matrices $\mA^{\pi^k}$ and $\mA^{\pi^{k+1}}$ differ only at the row $\sigma^k$:
\begin{align}
    \mA^{\pi^{k+1}} = \mA^{\pi^k} + \ve_{\sigma^k} \mDelta_{\sigma^k}
    \label{eq:Ak_from_Ak-1}
\end{align}
where $\mDelta_{\sigma^k}$ is a row vector defined as in the previous section.

\reviewedFirstRev{
    One can efficiently compute the inverse of matrix $\mA^{\pi^{k+1}}$ from the one of $\mA^{\pi^k}$ by using the Sherman-Morrison formula, which says that if $\mA\in\real^{n\times n}$ is an invertible square matrix and $\vp, \vq\in\real^n$ are two column vectors, then the matrix $\mA+\vp\vq^\top$ is invertible if and only if $1+\vq^\top \mA^{-1}\vp\neq0$ and if $\mA+\vp\vq^\top$ is invertible, then:
}
\begin{align*}
    (\mA + \vp\vq^\top)^{-1}
    &=\mA^{-1} - \frac{\mA^{-1}\vp\vq^\top \mA^{-1}}{1+\vq^\top \mA^{-1}\vp}.
\end{align*}
Let $X^k_{ij} := \mDelta_i (\mA^{\pi^k})^{-1}\ve_j$. Following \eqref{eq:Ak_from_Ak-1}, we can apply the Sherman-Morrison formula with matrix $\mA^{\pi^k}$, and vectors $\vp=\ve_{\sigma^k}$ and $\vq^\top=\mDelta_{\sigma^k}$. After some simplification, we get:
\begin{align}
    X^{k+1}_{ij} &:= \mDelta_i (\mA^{\pi^{k+1}})^{-1}\ve_j =\mDelta_i (\mA^{\pi^k}+\ve_{\sigma^k}\mDelta_{\sigma^k})^{-1}\ve_j\nonumber\\
                 &= X^{k}_{ij} - \frac{X^{k}_{i\sigma^k} }{1+X^{k}_{\sigma^k \sigma^k}}X^{k}_{\sigma^k j} \label{eq:X^k+1-i} \\
    \text{In particular,}\ X^{k+1}_{i\sigma^k} &=\frac{X^{k}_{i\sigma^k}}{1+X^{k}_{\sigma^k \sigma^k}}.\nonumber
\end{align}
\reviewedFirstRev{
    Before computing $X^{k+1}_{ij}$, we need to verify that $\pi^{k+1}$ is unichain.
    With the help of Lemma~\ref{lem:invertible} and the Sherman-Morrison formula, this can be done easily: $\pi^{k+1}$ is unichain if and only if $1+X^k_{\sigma^k\sigma^k}\neq0$. 
}

For $y^{k+1}_i := \mDelta_{i}\vd^{\pi^{k+1}} = -\mDelta_i(\mA^{\pi^{k+1}})^{-1} \vpi^{k+1}$, we use $\vpi^{k+1}=\vpi^k - \ve_{\sigma^k}$ and apply the Sherman-Morrison formula to get:
\begin{align}
    y^{k+1}_i &= -\mDelta_i(\mA^{\pi^{k}}+\ve_{\sigma^k}\mDelta_{\sigma^k})^{-1} (\vpi^k-\ve_{\sigma^k})\nonumber\\
              &= -\mDelta_i (\mA^{\pi^{k}})^{-1}(\vpi^k-\ve_{\sigma^k})  + \frac{\mDelta_i (\mA^{\pi^{k}})^{-1} \ve_{\sigma^k}\mDelta_{\sigma^k}(\mA^{\pi^{k}})^{-1}}{1 + \mDelta_{\sigma^k} (\mA^{\pi^{k}})^{-1} \ve_{\sigma^k}} (\vpi_k-\ve_{\sigma^k})\nonumber\\
    &= y^k_i + X^k_{i\sigma^k} + \frac{X^k_{i\sigma^k}(-y^k_{\sigma^k} - X^{k}_{\sigma^k\sigma^k})}{1+X^{k}_{\sigma^k\sigma^k}}\nonumber\\
    &= y^k_i + \frac{X^k_{i\sigma^k}(1-y^k_{\sigma^k})}{1+X^{k}_{\sigma^k\sigma^k}} =y^k_i +(1-y^k_{\sigma^k})X^{k+1}_{i\sigma^k}\label{eq:y^k+1-i}
\end{align}

The above formula indicate how to compute $\vy^{k+1}$ from $\vy^k$. To complete this analysis, let us show that $\vy^1=\mathbf{0}$. For a given policy $\pi$, the vector $\vd^{\pi}$ satisfies the same equation as Equation~\eqref{eq:pre_lin_eq} but replacing $r^{\pi_i}_i-\lambda\pi_i$ by $-\pi_i$. This implies that for $\vpi=\vpi^1=[1\ \dots\ 1]^\top$, one has $\vd^\pi=[-1\ 0\ \dots\ 0]^\top$ as $d^\pi_1$ is the long-run average reward of \reviewedFirstRev{a Markov reward process} whose reward is negative one in all states and $d^\pi_2,\dots, d^\pi_n$ is the bias of this process.  This shows that for all $i$, one has
\reviewedFirstRev{
    $y^{1}_i:=\sum_{j=2}^n\Delta_{ij}d_j^{\pi^{1}}=0$.
    Moreover, by \eqref{eq:h is linear}, one has
    \begin{align}
        \valpha^{\pi^1}(\lambda)
        &=\vdelta -\lambda\vone +\mDelta(\mA^{\pi^1})^{-1}\vr^{\pi^1} -\lambda\underbrace{\mDelta(\mA^{\pi^1})^{-1}\vpi^1}_{=:\vy^1} \nonumber\\
        &=\vdelta -\lambda\vone +\mX^1\vr^{\pi^1}. \label{eq:b_pi_one}
    \end{align}
    Finally, for $\vpi^{n+1}=[0\ \dots\ 0]^\top$, one has
    \begin{equation}
        \label{eq:b_pi_n1}
        \valpha^{\pi^{n+1}}(\lambda)
        =\vdelta -\lambda\vone +\mDelta(\mA^{\pi^{n+1}})^{-1}\vr^{\pi^{n+1}}.
    \end{equation}

}

\subsection{Detailed algorithm}
\label{ssec:two_third_algo}

Equation~\eqref{eq:mu^k_i_from_y} shows how to compute $\mu^k_i$ from the values of $y^{k}_i$ and $z^{k-1}_i$ while \eqref{eq:y^k+1-i}, \eqref{eq:X^k+1-i} and \eqref{eq:z^k-i} show how to compute the values of $\vy$, $\vz$ and $\mX$ recursively in $k$. In order to compute $\mu^k_{\min}$ and $\sigma^k$, one only needs to compute the values $\mu^k_i$ for $i\in\pi^{k}$.  Once \reviewedFirstRev{
    $\mu^k_{\min}=\min_{i\in\pi^k}\mu^k_i$ is determined, if $\mu^{k-1}_{\min} <\mu^k_{\min}$, then Line~\ref{algo1:test1} of Algorithm~\ref{algo:whittle_informal} can be performed, based on \eqref{eq:z^k-i}, by checking if $z^{k}_i \ge 0$ for some $i\in[n]\setminus\pi^{k}$.
}

\begin{algorithm}[hbtp]
    \begin{algorithmic}[1]
        \State Set $\pi^1:=[n]$, $k_0=1$, $\Delta_{i1}:=0, \Delta_{ij}:=P^1_{ij}-P^0_{ij}, \forall i\in[n], j\in\{2,\dots, n\}$, $\vy^1=\pmb{0}$ and $\mA^{\pi^1}$ is defined by \eqref{eq:def_A}. \label{algo2:init_beginning}
        \If{$\pi^1$ is multichain}
            \State \Return{the arm is multichain}
        \EndIf
        \State Set $\mX^1:=\mDelta(\mA^{\pi^1})^{-1}$  \label{algo2:line_inverse}
    \State Set $\vmu^1:=\vr^1-\vr^0 +\mX^1\vr^{1}$
    \State Let $\sigma^1:=\argmin_{i\in\pi^1}\mu_i^1$, $\lambda_{\sigma^1} = \mu^1_{\min} := \mu^1_{\sigma^1}$, and $\pi^{2} := \pi^{1}\setminus\{\sigma^{1}\}$ \label{line:def_mu1}
    \State Set $\vz^1=\vmu^1 -\mu^1_{\min}\vone$ \label{algo2:init_end}
    \For{$k=2$ {\bfseries to} $n$}
        \State Update\_X($k-1$) \label{line:update_X} \Comment{Here we call Subroutine~\ref{algo:update_X} or Subroutine~\ref{algo:FMM}}
        \State Set $\vy^{k}=\vy^{k-1} +(1 -y^{k-1}_{\sigma^{k-1}}) \mX^{k}_{:\sigma^{k-1}}$ \label{algo2:y^k}
        \For{$i\in\pi^k$}
        \State Set $\displaystyle\mu_i^k {:=}\begin{cases}
                \mu^{k-1}_{\min}, & \text{if } z^{k-1}_i=0 \\
                \mu^{k-1}_{\min} {+}\displaystyle\frac{z^{k-1}_i}{1{-}y^k_i}, & \text{if } z^{k-1}_i{>}0 \text{ and } 1{-}y^k_i{>}0 \\
                +\infty, & \text{otherwise}
\end{cases}$ \label{algo2:muki}
        \EndFor
        \State Let $\sigma^k:=\argmin_{i\in\pi^k}\mu_i^k$ and $\mu^k_{\min} := \mu^k_{\sigma^k}$ \label{line:def_muk}
        \State Set $\vz^k = \vz^{k-1} -(\mu^k_{\min}-\mu^{k-1}_{\min}) (\vone-\vy^k)$
        \If{$\mu^{k-1}_{\min}<\mu^k_{\min}$ and $z^{k}_i \ge 0$ for some $i\in[n]\setminus\pi^k$}
            \State \Return{the arm is not indexable}
        \EndIf
        \If{$\mu^k_{\min}=+\infty$}
            \State Set $\lambda_i:=+\infty$ for all $i\in\pi^k$
            \State \Return{the arm is indexable and the indices are $\{\lambda_i\}_{i\in[n]}$.}
        \EndIf
        \State Set $\lambda_{\sigma^k} = \mu^k_{\min}$ and $\pi^{k+1} := \pi^{k}\setminus\{\sigma^{k}\}$ 
        \label{algo2:end_for_loop}
    \EndFor
    \If{$\pi^{n+1}$ is multichain}
        \State \Return{the arm is multichain}
    \EndIf
    \State \Return{the arm is indexable and the indices are $\{\lambda_{i}\}_{i\in[n]}$}.
    \end{algorithmic}
    \caption{Given a $n$-state arm, test indexability and compute Whittle indices (if indexable).}
    \label{algo:whittle_n3}
\end{algorithm}


\begin{algorithm}[ht]
    \floatname{algorithm}{Subroutine}
    \begin{algorithmic}[1]
    \For{$\ell=1$ to $k-1$}
        \For{$i\in[n]$} \Comment{or $i\in\pi^{\ell+1}$ if we do not test indexability.}
            \State $X_{i\sigma^{k}}^{\ell+1} = X_{i\sigma^{k}}^{\ell} -X^{\ell+1}_{i\sigma^{\ell}}X^{\ell}_{\sigma^{\ell}\sigma^{k}}$
        \EndFor
    \EndFor
    \If{\reviewedFirstRev{$1+X^{k}_{\sigma^{k}\sigma^{k}}=0$}}
        \State \Return{\reviewedFirstRev{the arm is multichain}}
    \EndIf
    \For{$i\in[n]$} \Comment{or $i\in\pi^{k+1}$ if we do not test indexability.}
    \State $X_{i\sigma^{k}}^{k+1} = \displaystyle\frac{X^{k}_{i\sigma^{k}}}{1+X^{k}_{\sigma^{k}\sigma^{k}}}$
    \EndFor
    \end{algorithmic}
    \caption{Update\_X(k)}
    \label{algo:update_X}
\end{algorithm}

This leads to Algorithm~\ref{algo:whittle_n3} that can be decomposed as follows:
\begin{enumerate}
    \item In Lines~\ref{algo2:init_beginning} to \ref{algo2:init_end}, we initialize the various variables. The main complexity of this part is to compute the matrix $\mX^1$, which is equivalent to solving the linear system $\mX\mA^{\pi^1}=\mDelta$. It can be done by inverting the matrix $\mA^{\pi^1}$ and multiplying this by the matrix $\mDelta$. This can be done in a subcubic complexity by using for instance Strassen's algorithm \citep{strassen1969gaussian}.
    \item We then enter the main loop:
        \begin{itemize}
            \item We update the vectors $\vmu$, $\vz$ by using Equations \eqref{eq:mu^k_i_from_y} and \eqref{eq:z^k-i}, and test indexability. This costs $O(n)$ operations per iteration, thus $O(n^2)$ in total.
            \item We update the vector $\mX$ according to \eqref{eq:X^k+1-i}. The ``naive'' way to do so is to use Subroutine~\ref{algo:update_X}. At iteration $k$ this costs $2kn$ arithmetic operations if we test indexability, and $2\sum_{l=1}^{k}(n-l)$ if we do not test indexability. The total complexity of computing $\mX$ is $n^3 +O(n^2)$ arithmetic operations if we test indexability and $(2/3)n^3 +O(n^2)$ if we do not.
              A detailed study of the arithmetic complexity is provided in Appendix~\ref{apx:implementation}, where we also provide details on how to efficiently implement the algorithm, including how to optimize the cost of memory access.
            \item In Line \ref{algo2:y^k}, we update the vector $\vy$ by using Equation~\eqref{eq:y^k+1-i}, which costs $O(n)$ per iteration. 
        \end{itemize}
    \item Testing if $\pi^{n+1}$ is unichain can be done in $O(n^2)$ by Tarjan's strongly connected component algorithm.
\end{enumerate}
Hence, the total complexity of this algorithm is $n^3+o(n^3)$ if we test indexability and $(2/3)n^3+o(n^3)$ if we do not test indexability. Without testing the indexability, our algorithm has the same main complexity as \citep{nino2020fast}. However, the algorithm of \citep{nino2020fast} computes Whittle index only for an arm that is PCL-indexable. Hence we can claim that our algorithm is the first algorithm that computes Whittle index with cubic complexity for general indexable restless bandits. It is also the first algorithm \reviewedFirstRev{with cubic complexity} that tests non-indexability of restless bandits \reviewed{composed of unichain arms}.
\reviewedFirstRev{Note that the ties are broken arbitrarily for Lines~\ref{line:def_mu1} and \ref{line:def_muk} of Algorithm~\ref{algo:whittle_n3}. 
}

\textbf{Remark:} Equivalently, instead of calling Subroutine~\ref{algo:update_X} at Line~\ref{line:update_X}, one could do the following update for all $j\in\pi^{k+1}$ and $i\in[n]$ (or $i\in\pi^{k+1}$ if we do not test indexability):
\begin{align}
    X_{ij}^{k+1} = X_{ij}^{k} -\displaystyle\frac{X^{k}_{i\sigma^{k}}}{1+X^{k}_{\sigma^{k}\sigma^{k}}}X^{k}_{\sigma^{k}j}. \label{eq:update_X_naif}
\end{align}
This iterative update is very close to the one used in \citep{akbarzadeh2020conditions,nino2020fast} for discounted restless bandit. This results in an algorithm that has the same total complexity as Subroutine~\ref{algo:update_X} (of $n^3+O(n^2)$ or $(2/3)n^3+O(n^2)$ with or without the indexability test) because both algorithms will have computed the same values of $X^{k}_{ij}$. The reason to use Subroutine~\ref{algo:update_X} is that, as we will see in the next section, not all values of $X^{k}_{ij}$ are needed at iteration $k$: in particular, for $\ell<k$, the computation of the value $X^\ell_{i\sigma^{k-1}}$ has no interest per say and is only useful because it allows to recursively compute $X^{k}_{i\sigma^{k-1}}$. In the section below, we show how to reduce the cost by avoiding the computation of $X^\ell_{i\sigma^{k-1}}$ when $\ell$ is much smaller than $k$. We comment more on the differences with \citep{akbarzadeh2020conditions,nino2020fast} in Appendix~\ref{apx:comparison} and in particular we explain why our approach can be tuned into a subcubic algorithm while \eqref{eq:update_X_naif} cannot.

\section{The subcubic algorithm}
\label{sec:subcubic algorithm}

\subsection{Main idea: recomputing \texorpdfstring{$\mX^{k}$ from $\mA^{\pi^k}$}{Xk from Apik} periodically}

The main computational burden of Algorithm~\ref{algo:whittle_n3} is concentrated on two lines: on Line~\ref{algo2:line_inverse} where we compute $\mX^1$ by solving a linear system, and on Line~\ref{line:update_X} where we compute the column vector $\mX^{k}_{:\sigma^{k-1}}$ from $\mX^1_{: \sigma^{k-1}}$. The remainder of the code runs in $O(n^2)$ operations and is therefore negligible for large matrices. In fact, the computation of $\mX^1$ is done by solving a linear system, which can be computed by using a subcubic algorithm (like Strassen~\citep{strassen1969gaussian}). In this section, we show how to optimize our algorithm by reducing the complexity of the update\_X() function, at the price of recomputing the full matrix $\mX^{k_0}$ from $\mA^{\pi^{k_0}}$ periodically.

At iteration $k$ of Algorithm~\ref{algo:whittle_n3}, the quantities $X^{k}_{i\sigma^{k-1}}$ are used at Line~\ref{algo2:y^k} to obtain the values $y^{k}_i$. The matrix $\mX^k$ is defined as $\mX^k:=\mDelta(\mA^{\pi^k})^{-1}$. It also satisfies Equation~\eqref{eq:X^k+1-i}, that is, for an iteration $\ell$ and states $i$ and $j$, we have: 
\reviewedFirstRev{
    \begin{align}
        \label{eq:X^ell+1-i}
        X^{\ell+1}_{ij} &= X^{\ell}_{ij} - X^{\ell+1}_{i\sigma^\ell} X^{\ell}_{{\sigma^\ell} j}.
    \end{align}
}

The way Subroutine~\ref{algo:update_X} is implemented is to initialize $\mX^1:=\mDelta(\mA^{\pi^1})^{-1}$ and then use \eqref{eq:X^ell+1-i} recursively to compute the column vector $\mX^{k}_{: \sigma^{k-1}}$ from $\mX^1_{: \sigma^{k-1}}$ at each iteration.

Here, we propose an alternative formulation which consists in recomputing the whole matrix $\mX^{k}:=\mDelta(\mA^{\pi^k})^{-1}$ every $K$ iterations. In the meantime, we use \eqref{eq:X^ell+1-i} to compute the values $X^{\ell}_{i\sigma^{k-1}}$ for ${\ell\in\{k_0+1,\dots, k\}}$ where $k_0$ is the iteration at which we recomputed the whole matrix ${\mX^{k_0}:=\mDelta(\mA^{\pi^{k_0}})^{-1}}$. This can be implemented by replacing the call to Subroutine~\ref{algo:update_X} at Line~\ref{line:update_X} with a call to Subroutine~\ref{algo:FMM}.
\begin{algorithm}[ht]
    \floatname{algorithm}{Subroutine}
    \begin{algorithmic}[1]
    \If{$k$ is an iteration at which we recompute the whole $\mX^{k+1}$}
        \State Set $k_0 := k+1$\;
        \If{\reviewedFirstRev{$\mA^{\pi^{k_0}}$ is not invertible}}
        \State \Return{\reviewedFirstRev{the arm is multichain}}
        \EndIf
        \State Set $\mX^{k_0}:=\mDelta(\mA^{\pi^{k_0}})^{-1}$ \label{algoFMM:recompute}\;
    \Else
        \For{$\ell=k_0$ to $k-1$}
            \label{algoFMM:internal loop}
            \For{$i\in[n]$} \Comment{or $i\in\pi^{\ell+1}$ if we do not test indexability.}
                \State $X_{i\sigma^{k}}^{\ell+1} = X_{i\sigma^{k}}^{\ell} -X^{\ell+1}_{i\sigma^{\ell}}X^{\ell}_{\sigma^{\ell}\sigma^{k}}$ \;
            \EndFor
        \EndFor
        \If{\reviewedFirstRev{$1+X^{k}_{\sigma^{k}\sigma^{k}}=0$}}
        \State \Return{\reviewedFirstRev{the arm is multichain}}
        \EndIf
        \For{$i\in[n]$} \Comment{or $i\in\pi^{k+1}$ if we do not test indexability.}
        \State $X_{i\sigma^{k}}^{k+1} = \displaystyle\frac{X^{k}_{i\sigma^{k}}}{1+X^{k}_{\sigma^{k}\sigma^{k}}}$
        \EndFor
    \EndIf
    \end{algorithmic}
    \caption{Update\_X\_FMM(k)}
    \label{algo:FMM}
\end{algorithm}

To see why this can be more efficient, we illustrate in Figure~\ref{fig:fast_mm} the pairs $(\ell,\sigma^{k-1})$ for which we compute the vector $\mX^{\ell}_{: \sigma^{k-1}}$, either for Subroutine~\ref{algo:update_X} or Subroutine~\ref{algo:FMM}.  In each case, a vertical blue line indicates that we recompute the whole matrix $\mX^k$ by solving a linear system.  The gray zone corresponds to the values $(\ell,\sigma^{k-1})$ for which we compute $\mX^\ell_{: \sigma^{k-1}}$ using Equation \eqref{eq:X^ell+1-i} and the red squares represent the vector $\mX^{k}_{: \sigma^{k-1}}$ used at Line~\ref{algo2:y^k} of Algorithm~\ref{algo:whittle_n3}. For Subroutine~\ref{algo:update_X}, we do one matrix inversion at the beginning and then compute for all $(\ell,\sigma^{k-1})$ with $\ell\le k$ because the red square at value $(k,\sigma^{k-1})$ is computed starting from the vertical blue line at value $(1,\sigma^{k-1})$. For Subroutine~\ref{algo:FMM}, we do \reviewedFirstRev{$ \floor{n/K} $ } (here \reviewedFirstRev{$ \floor{n/K} =4$}) full recomputation of $\mX^k$, which correspond to the vertical blue lines. We gain in terms of operations because the surface of the gray zone to compute is divided by \reviewedFirstRev{$ \floor{n/K} $}.

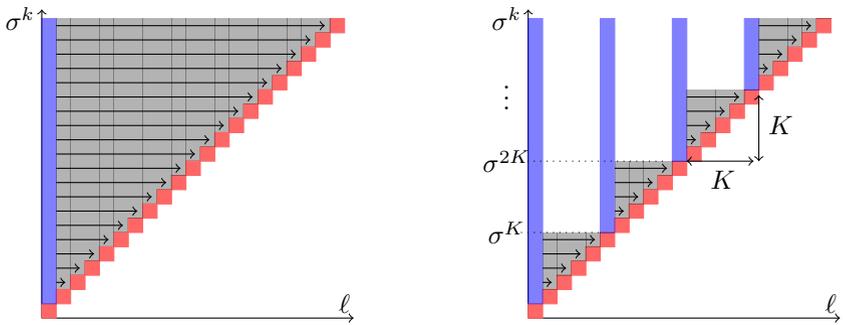
\begin{figure}[ht]
    \centering
    \begin{tabular}{cc}
        \begin{tikzpicture}[scale=0.19, shorten >=2pt]
            \draw (0,0) edge[->] (22,0);
            \draw (0,0) edge[->] (0,22);
            \node at (21,1) {$\ell$};
            \node at (-1.5,21) {$\sigma^k$};
            \foreach \i in {0,...,20} \draw[fill,red!60] (\i,\i) rectangle (\i+1,\i+1);
            \foreach \i in {0}{
               \draw[fill,blue,opacity=0.5] (\i,\i+1) rectangle (\i+1,21);
               \foreach \j in {1,...,19}{
                   \fill[black,opacity=0.3] (\i+\j,\i+\j+1) rectangle (\i+\j+1,\i+21);
               }
            }
            \foreach \i in {2,...,20}{
                \draw (1,\i+0.5) edge[->] (\i,\i+0.5);
            }
        \end{tikzpicture}
        &
        \begin{tikzpicture}[scale=0.19, shorten >=2pt]
            \draw (0,0) edge[->] (22,0);
            \draw (0,0) edge[->] (0,22);
            \node at (21,1) {$\ell$};
            \node at (-1.5,21) {$\sigma^k$};
            \foreach \i in {0,...,20} \draw[fill,red!60] (\i,\i) rectangle (\i+1,\i+1);
            \foreach \i in {0,5,10,15}{
               \draw[fill,blue,opacity=0.5] (\i,\i+1) rectangle (\i+1,21);
               \foreach \j in {1,...,5}{
                   \fill[black,opacity=0.3] (\i+\j,\i+\j+1) rectangle (\i+\j+1,\i+6);
                }
                \foreach \j in {2,...,5}{
                    \draw (\i+1,\i+\j+0.5) edge[->] (\i+\j,\i+\j+0.5);
                }
            }
            \draw (11,11) edge[<->] node[below] {$K$} (16,11);
            \draw (16,11) edge[<->] node[right] {$K$} (16,16);
            \node at (-1.5,6) {$\sigma^K$}; \draw[dotted] (-.5,6) -- (5.5,6);
            \node at (-1.5,11) {$\sigma^{2K}$}; \draw[dotted] (-.5,11) -- (10,11);
            \node at (-1.5,16) {$\vdots$};
        \end{tikzpicture} \\
        (a) Computation load of Subroutine~\ref{algo:update_X}.
        &(b) Computation load of Subroutine~\ref{algo:FMM}
    \end{tabular}
    \caption{Illustration of the improvement proposed by replacing Subroutine~\ref{algo:update_X} with Subroutine~\ref{algo:FMM} \reviewedFirstRev{when we do not test the indexability.} For Subroutine~\ref{algo:FMM}, we solve more linear systems (each vertical blue line corresponds to solving a linear system) but we reduce the gray zone to compute. A linear arrow corresponds to the internal loop of Line~\ref{algoFMM:internal loop} of Subroutine~\ref{algo:FMM}.}
    \label{fig:fast_mm}
\end{figure}

Note that the $y$-axis of Figure~\ref{fig:fast_mm} is ordered by increasing value of $\sigma^k$ (and not by increasing value of $k$). The value of $\sigma^k$ is computed at iteration $k$ but unknown before iteration $k$. This explains why in Subroutine~\ref{algo:FMM}, when we recompute the matrix $(X^{k}_{ij})$ at an iteration $k$ (vertical blue lines in Figure~\ref{fig:fast_mm}(b)), we recompute it for all $i,j$ and not just $i,j\in\{\sigma^{k}, \dots, \sigma^K\}$ (which are the only values that we will use): Indeed, $\sigma^{k}, \dots, \sigma^K$ are unknown at iteration $k$.

\subsection{A subcubic algorithm for Whittle index}

We now assume to have access to a subcubic matrix multiplication algorithm that satisfies the following property:
\begin{enumerate}[label=(FMM)]    
    \item \label{hypo:FMM} There exists an algorithm to multiply a matrix of size $n\times n$ by a matrix of size\footnote{We write $n^\gamma$ which is possibly non-integer. For the sake of simplicity, we write $n^\gamma$ and $n^{1-\gamma}$ but they should be understood as $\floor{n^\gamma}$ and $\floor{n^{1-\gamma}}$ respectively.} $n\times  n^\gamma $ that runs in $O(n^{\omega(\gamma)})$, where $\omega:[0,1]\to[2,3]$ is a non-decreasing function.
\end{enumerate}

Going back to Algorithm~\ref{algo:whittle_n3} where Line~\ref{line:update_X} is Subroutine~\ref{algo:FMM}, we now assume that we recompute the whole matrix $\mX^k$ every $O(n^\gamma)$ iterations. The new algorithm has a subcubic complexity:

\begin{theorem}
    \label{th:subcubic}
    \reviewedFirstRev{Given a $n$-state arm}, Algorithm~\ref{algo:whittle_n3} with Subroutine~\ref{algo:FMM} checks indexability and computes Whittle (and Gittins) index in time at least $\Omega(n^{2.5})$ and at most $O(n^{2.5286})$ when choosing $\gamma=0.5286$.
\end{theorem}

We believe that Theorem~\ref{th:subcubic} is \reviewedFirstRev{the first theoretical result} that shows that Whittle index can be computed in subcubic time. As we show in Section~\ref{sec:discounted}, this algorithm can be directly extended to discounted index. As a byproduct, we also obtain the first subcubic algorithm to compute Gittins index.

\begin{proof}
    The algorithm starts by computing $\mX^1$ which can be done in $O( n^{\omega(1)})$. Then, there are $ n^{1-\gamma} $ times that we do:
    \begin{enumerate}
        \item We fill the ``gray'' mini matrices by using \eqref{eq:X^ell+1-i}. This amounts to three for loops of size $n$ (for $i$), $ n^\gamma $ (for $k$) and $ n^\gamma $ (for $\ell$). Hence, each small gray matrix costs $O( n^{1+2\gamma} )$.
        \item At the end of a cycle, we recompute the full inverse by updating $(\mA^{\pi^k})^{-1}$ from $(\mA^{\pi^{k- n^\gamma }})^{-1}$. As we show in Lemma~\ref{lem:woodbury} (stated below), this can be done in $O(  n^{\omega(\gamma)} )$.
\end{enumerate}
This implies that the algorithm has a complexity:
\begin{align*}
    O( n^{\omega(\gamma)} ) +  n^{1-\gamma}  \left(O( n^{1+2\gamma} )+O( n^{\omega(\gamma)} )\right)  = O( n^{\max\{2+\gamma, 1-\gamma + \omega(\gamma)\}} ).
\end{align*}

To compute the optimal value of $\gamma$ minimizing this expression requires the knowledge of the function 
$\omega(\gamma)$ which is  not known. The current state of the art  only gives  a lower bound ($\omega(\gamma) \geq 2$ ) and an upper bound described in  \citep{gall2018improved}.

It is shown in \citep[Table~3]{gall2018improved} that $\gamma=0.5286$ is the smallest currently known value of $\gamma$ for which  $\omega(\gamma)<1+2\gamma$. This implies that the complexity is  at most $O(n^{2.5286})$.

As for the lower bound,  $\omega(\gamma) \geq 2$ implies that the complexity of the
algorithm is at least  $\Omega(n^{2.5})$.
    \end{proof}
  

    In the next lemma, $\mB$ plays the role of $\mA^{\pi^k}$ and $\mA$ the role of $\mA^{\pi^{k- n^{\gamma} }}$. Note that as required in the lemma, exactly $ n^\gamma $ rows and columns are changed between the two.
\begin{lemma}
    \label{lem:woodbury}
    Assume~\ref{hypo:FMM}. Let $\mA$ be a square matrix whose inverse $\mA^{-1}$ has already been computed, and let $\mB$ be an invertible square matrix such that $\mA-\mB$ is of rank smaller than $ n^{\gamma} $. Then, it is possible to compute the inverse of $\mB$ in $O(n^{\omega(\gamma)})$. 
\end{lemma}
\begin{proof}
    The matrix $\mB$ can be written as $\mB=\mA + \mU\mC\mV$ where $\mU$ is a $n\times n^\gamma $ matrix, $\mC$ is $ n^\gamma \times n^\gamma $ and $\mV$ is $ n^\gamma \times n$. The Sherman–Morrison–Woodbury formula~\citep{woodbury1950inverting} states that
    \begin{align*}
        \mB^{-1} = \left(\mA + \mU\mC\mV \right)^{-1} = \mA^{-1} - \mA^{-1}\mU \left(\mC^{-1} + \mV\mA^{-1}\mU \right)^{-1} \mV\mA^{-1}.
    \end{align*}
    This shows that $\mB^{-1}$ can be computed by:
    \begin{itemize}
        \item Computing $\mD:=\mA^{-1}\mU$ and $\mE:=\mV\mA^{-1}$: this takes $O(n^{\omega(\gamma)})$. 
        \item Computing $\mF:=\left(\mC^{-1} + \mV\mA^{-1}\mU \right)^{-1}$: as this is the inversion of a $ n^\gamma \times n^\gamma $ matrix, it can be done in $O(n^{\gamma\omega(1)})$ where $\gamma\omega(1)\le \omega(\gamma)$.
        \item Computing $\mG:=\mD\mF$ and then $\mG\mE$: this again takes $O(n^{\omega(\gamma)})$. 
    \end{itemize}
    Hence, computing $\mB^{-1}$ can be done in $O(n^{\omega(\gamma)})$ operations for the inversion and all   multiplications plus an additional $O(n^2)$ term for the subtraction and the addition. As $\omega(\gamma)\ge2$, this concludes the proof of the lemma.
\end{proof}

\subsection{The subcubic algorithm in practice}

The complexity of $O(n^{2.5286})$ given in Theorem \ref{th:subcubic} is mainly of theoretical interest. The value $\gamma = 0.5286$ is obtained by using the best upper bound  on $\omega(\gamma)$ known today which is based on the Coppersmith-Winograd algorithm and its variants. The Coppersmith-Winograd algorithm (or its variants) are, however, known as a \emph{galactic algorithm}: the hidden constant in the $O()$ is so large that their runtime is prohibitive for any reasonable value of $n$. Hence, the existence of these algorithms is of theoretical interest but has limited applicability.

This does not discard the practical improvement provided by Subroutine~\ref{algo:FMM} which is based on the mere fact that multiplying  two matrices (or inverting a matrix) is faster than three nested loops even for matrices of moderate size. To verify this, we launched  a detailed profiling of the code of Algorithm \ref{algo:whittle_n3} with the non-optimized Subroutine~\ref{algo:update_X}. It shows that for a problem of dimensions $5000$, the update of Line~\ref{line:update_X} takes more than 90\% of the computation time, the initialization of $\mX^1$ on Line~\ref{algo2:line_inverse} takes about $5\%$ of the time and the rest of the code takes less than $1\%$ of the running time.

Now, if inverting the full matrix takes about $5\%$ of the execution time, and updating the gray zone takes $95\%$, then by doing $5$ updates, one can hope to obtain an algorithm whose running time is roughly $5\times5 + 95/5\approx43\%$ the one of the original implementation. As we observe in Section~\ref{sec:numerical}, this is close to the gain that we obtain in practice. A general way to choose the best number of updates is used in the numerical section. It is based on the following reasoning. For large matrices (say $n\ge 10^3$), the fastest implementations of matrix multiplication and inversion are based on Strassen's algorithm~\citep{huang2016strassen,huang2018practical}. As we report in Appendix~\ref{apx:inversion}, the time to solve a linear system of size $n$ by using the default installation of \texttt{scipy} seems to run in $O(n^{2.8})$. By replacing the function $\omega(\gamma)$ used in Theorem \ref{th:subcubic} by a more practical bound ($\omega(\gamma)=2.8$), the best value for $\gamma$ becomes $\gamma = 0.9$. This indicates that our algorithm can be implemented in $O(n^{2.9})$ by doing $O(n^{0.1})$ recomputation of $\mX^k$ from $\mA^{\pi^k}$. Note that even for very large values of $n$ (like $n=15000$), $n^{0.1}$ remains quite small, \emph{e.g.}, $15000^{0.1}\approx 2.6$. In practice, we observe that updating $\mathrm{int}(2n^{0.1})$ times (the notation $\mathrm{int}(x)$ indicates that it is rounded to the closest integer) gives the best performance among all algorithms, as reported in the next section.

\section{Numerical experiments}
\label{sec:numerical}


In complement to our theoretical analysis, we developed a python package that implements Algorithm~\ref{algo:whittle_n3} and gives the choice of using the variant of Subroutine~\ref{algo:update_X} or of Subroutine~\ref{algo:FMM} to do the ``update\_X()'' function. This package relies on three python libraries: \texttt{scipy} and \texttt{numpy} for matrix operations, and \texttt{numba} to compile the python code. To facilitate its usage, this package can be installed by using \texttt{pip install markovianbandit-pkg}.

All experiments were conducted on a laptop (Macbook Pro 2020) with an Intel Core i9 CPU at 2.3 GHz with 16GB of Memory using Python 3.6.9 :: Anaconda custom (64-bit) under macOS Big Sur version 11.6.2. The version of the packages are \texttt{scipy} version 1.5.4, \texttt{numpy} version 1.19.5 and \texttt{numba} version 0.53.1. The code of all experiments is available at \url{https://gitlab.inria.fr/markovianbandit/efficient-whittle-index-computation}.

In all of our experiments, \reviewed{the way we generated arms guarantees that they are almost surely unichain because all the elements on their diagonal as well as on the upper and lower diagonals are positive.}

\subsection{Time to compute Whittle indices}

To test the implementation of our algorithm, we randomly generate restless bandit arms with $n$ states where $n\in\{100,1000,\dots,15000\}$. In each case, both transition matrices are uniform probabilistic matrices: for each row of each matrix, we generate $n$ i.i.d. entries following the exponential distribution and divide the row by its sum. This means that all matrices are dense.
We use dense matrix since it is the worst case for computational interest. By running our algorithms on sparse matrix, we would expect to have faster running time.
Note that all tested matrices are indexable. This is coherent with \citep{gast2020exponential,nino2007dynamic} that report that for uniform matrices, the probability of finding a non-indexable example decreases very rapidly with the dimension $n$. Finally, reward vectors were generated from random Uniform[0,1) entries.


\begin{table}[ht]
    \caption{Running time (in seconds) of the variants of Algorithm~\ref{algo:whittle_n3}}
    \label{table:running_time}
    \begin{tabular}{|c|c|c||c|c|}
        \hline
        &\multicolumn{2}{c||}{$O(n^3)$ algorithm (Subroutine~\ref{algo:update_X})} & \multicolumn{2}{c|}{Subcubic algorithm (Subroutine~\ref{algo:FMM})}\\
        $n$    &  (a) With index. test &  (b) Without test  & (c) With index. test &  (d) Without test  \\
        \hline
        100 & 0.006 & 0.004 & 0.007 & 0.005 \\
        1000 &   0.2 &   0.2 &   0.2 &   0.2 \\
        2000 &   1.9 &   1.2 &   1.1 &   1.0 \\
        3000 &   7.2 &   5.0 &   3.2 &   2.6 \\
        4000 &    17 &    12 &     8 &     6 \\
        5000 &    34 &    25 &    16 &    12 \\
        6000 &    60 &    43 &    27 &    20 \\
        7000 &    95 &    68 &    42 &    33 \\
        8000 &   142 &    99 &    63 &    49 \\
        9000 &   199 &   141 &    92 &    70 \\
        10000 &   275 &   191 &   122 &    95 \\
        11000 &   361 &   257 &   164 &   133 \\
        12000 &   471 &   339 &   225 &   190 \\
        13000 &   620 &   428 &   286 &   243 \\
        14000 &   790 &   553 &   408 &   317 \\
        15000 &   965 &   685 &   501 &   403 \\
        \hline
    \end{tabular}
\end{table}
We record the runtime of the different variants of our algorithm and report the results in Table~\ref{table:running_time}. Note that these results present the whole execution time of the algorithm, including the initialization phase in which $\mX^1$ is computed. For each value of $n$, we run Algorithm~\ref{algo:whittle_n3} with four variants: 
\begin{itemize}
    \item The first two columns correspond to the $O(n^3)$ algorithm (that uses Subroutine~\ref{algo:update_X} for Line~\ref{line:update_X}), either (a) with the indexability test, or (b) without the indexability test. 
    \item The last two columns correspond to the subcubic algorithm (that uses Subroutine~\ref{algo:FMM} for Line~\ref{line:update_X} with int($2n^{0.1}$) updates), either (c) with the indexability test, or (d) without the indexability test.
\end{itemize}
Our numbers show that our algorithm can compute the Whittle index in less than one second for $n=1000$ states and slightly less than $7$ minutes for $n=15000$ states with variant (d). As expected, not doing the indexability test does improve the performance compared to doing the indexability test (here by a factor approximately $1/3$ for Subroutine~\ref{algo:update_X} and $1/5$ for Subroutine~\ref{algo:FMM}). More importantly, this table shows that the time when using the subcubic variant, Subroutine~\ref{algo:FMM}, diminishes the computation time by about 40\% to 50\% compared to when using Subroutine~\ref{algo:update_X}. Note that for $n=100$, using the Subroutine~\ref{algo:update_X} is slightly faster than using the Subroutine~\ref{algo:FMM} (while both takes only a few milliseconds). This indicates that the subcubic algorithm becomes interesting when $n$ is large enough (say $n\ge2000$).

To give a visual idea of how the various variants of the algorithms compare, we plot in Figure~\ref{fig:benchmark_algo_whittle_generic} the runtime of the four variants along with two variants of the algorithms of \cite{nino2020fast}: FPA-Matlab (the original matlab implementation), and FPA-Julia: a Julia's implementation of the algorithm provided by the authors. We choose to compare to this algorithm as it was the one with the smallest complexity up to now. The numbers for FPA-Matlab are the ones reported in \cite{nino2020fast} and they are comparable to the ones that we obtained on our machine with the same algorithm. FPA-Julia is significantly faster. Hence, we plot in Figure~\ref{fig:speedup} the runtime of each variant divided by the runtime of FPA-Julia. For large $n$, our best implementation is about $4$ to $6$ times faster than the best one of \citep{nino2020fast}. For instance, for $n=15000$, our implementation takes about $7$ minutes to compute the index (or $9$ minutes when checking indexability) whereas FPA-Julia takes $33$ minutes (and does not check indexability on the fly). In our implementation, not testing indexability reduces the computation time of $15-20\%$ for Subroutine~\ref{algo:FMM} or $25-30\%$ for Subroutine~\ref{algo:update_X} compared to the version that tests indexability.

It should be clear that the comparison of the computation times of our implementation versus the ones of FPA-Matlab or FPA-Julia has its limits, because we do not use the same programming language. The influence of the choice of programming language is clear when comparing FPA-Matlab and FPA-Julia: while both codes are similar, the compiled Julia code is about 5 times faster.  To obtain a fairer comparison with our algorithm, we tried to rewrite the algorithm of \cite{nino2020fast} in Python with Numba but our implementation was significantly slower than the one of FPA-Julia. We do not know if this is by lack of optimization of our code or because Julia is indeed faster. Regardless of the choice programming langages, our implementation has nevertheless a few advantages that can explain why it is faster: our algorithm has some technical advantage (a simpler internal loop, and the use of Subroutine~\ref{algo:FMM}), and our implementation is optimized to improve data access pattern (which consists mostly in sorting on the fly the array according to the permutation $\sigma$, see Appendix~\ref{apx:implementation}).


\begin{figure}[ht]
    \begin{tabular}{cc}
        \begin{subfigure}[t]{0.47\linewidth}
            \includegraphics[width=\linewidth]{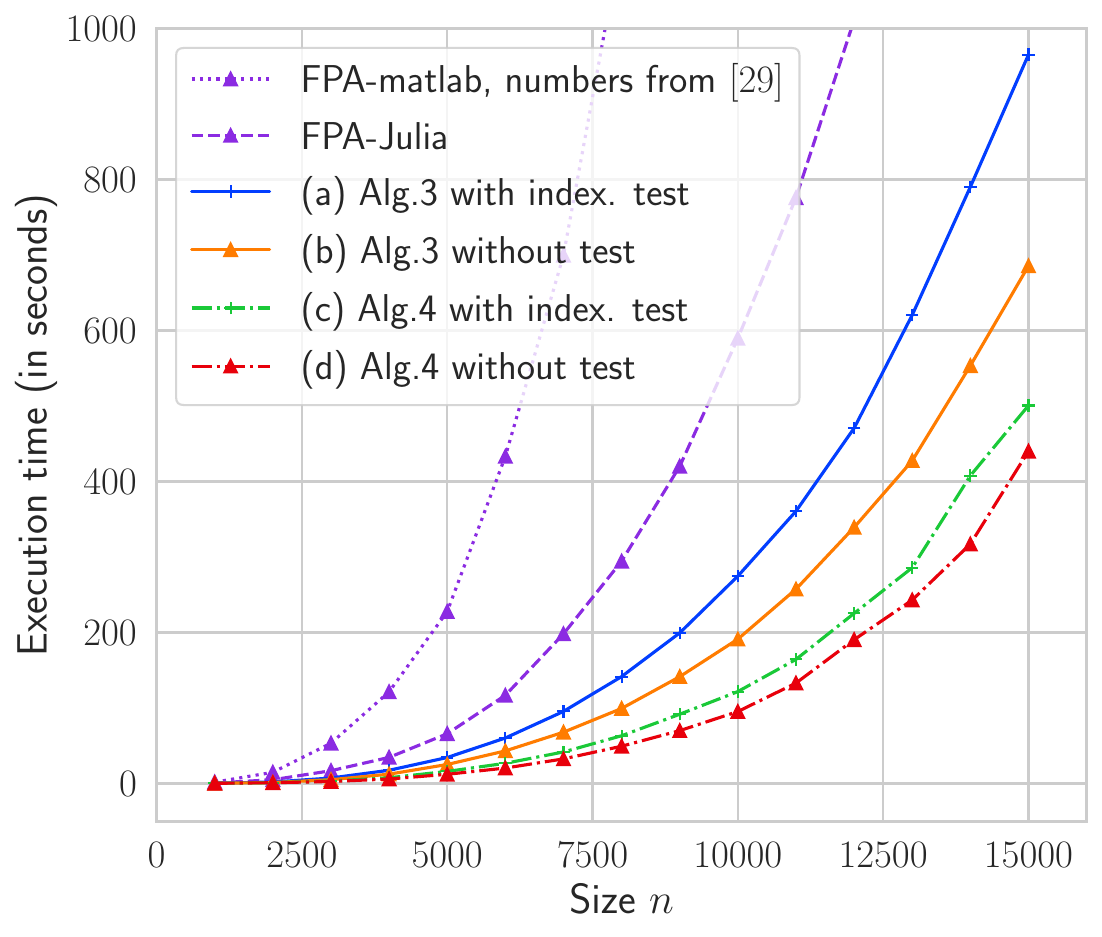}
            \caption{Runtime of our implementations as a function of the state size $n$.}
            \label{fig:benchmark_algo_whittle_generic}
        \end{subfigure}            
        &\begin{subfigure}[t]{0.47\linewidth}
            \includegraphics[width=\linewidth]{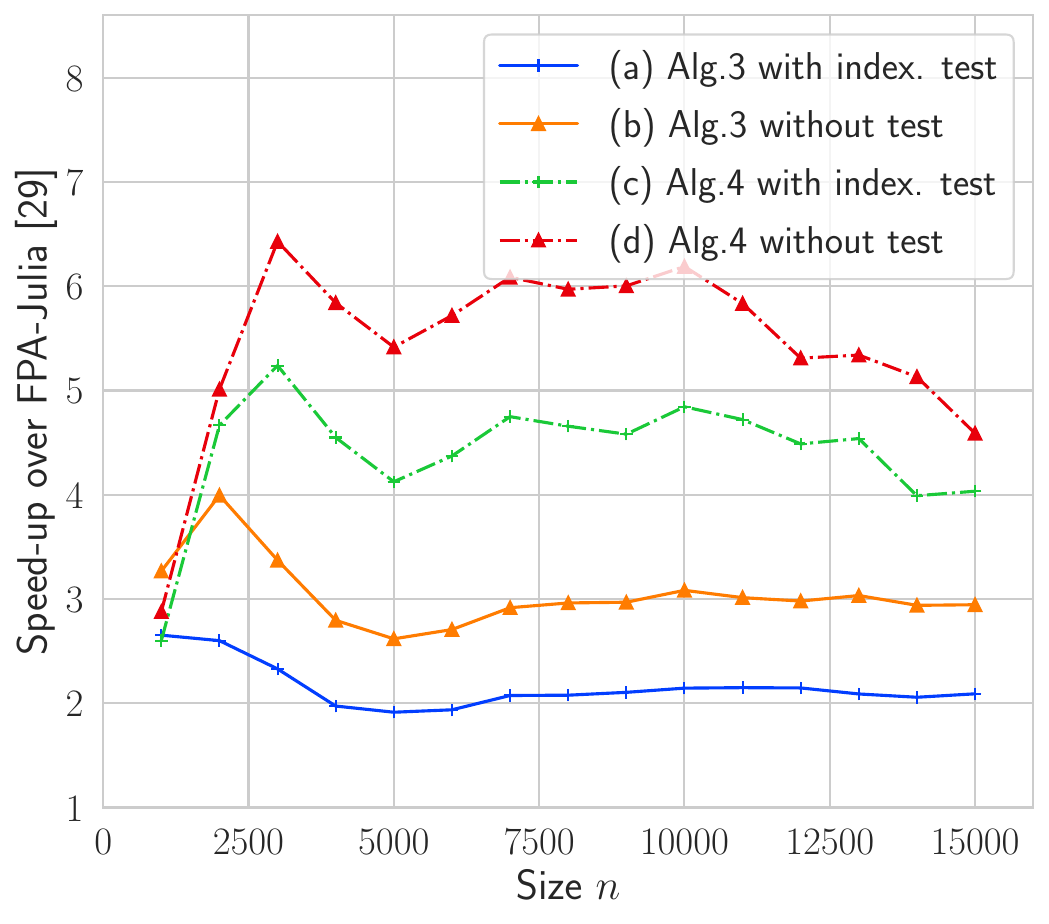}
            \caption{Speedup of each implementation over the data from \citep{nino2020fast} as a function of the state size $n$.}
            \label{fig:speedup}
        \end{subfigure}
    \end{tabular}
    \caption{Numerical result over 7 simulations: in each simulation, we run the algorithm over randomly generated RBs with the state size ranging over $\{1000,\dots,15000\}$. We plot the average runtime over 7 simulations. The solid lines represent the result of Subroutine~\ref{algo:update_X} and the dashed-dot lines represent the one of Subroutine~\ref{algo:FMM}. The marker ``$+$'' indicates that algorithms test indexability and the triangles indicates that algorithms do not test indexability.}
    \label{fig:numerical_result}
\end{figure}

\subsection{Statistics of indexable problems}

To the best of our knowledge, our algorithm provides the first indexability test that scales well with the dimension $n$. We used this to answer a very natural question: given a randomly generated arm, how likely is it to be indexable? This question was partially answered in \citep{nino2007dynamic} that shows that when generating dense arms, the probability of generating a non-indexable arm is close to $10^{-n}$ for $n\in\{3,\dots,7\}$. This suggests that most arms are indexable. Below, we answer two questions: what happens for larger values of $n$, and more importantly, what happens when the state transition matrices are not dense? 

To answer these questions, we consider randomly generated arms where the matrices $\mP^0$ and $\mP^1$ are $b$-diagonal matrices with $b$  non-null diagonals. In particular, $b=3$ corresponds to tridiagonal matrices, $b=5$ corresponds to pentadiagonal matrices and $b=7$ corresponds to septadiagonal matrices. We also compare with the classical case of dense matrices (which corresponds to $b=2n-1$).  For each model, the entries are generated from the exponential distribution for each row and we divide the row by the sum of generated entries for this row.   We vary $n$ from $3$ to $50$ and for each case, we generate $100000$ arms. We report in Table~\ref{tb:number_indexable} the number of indexable arms for each case. Note that pentadiagonal matrices are dense matrices for $n=3$ and septadiagonal matrices are dense matrices for $n=4$ and do not make sense for $n=3$, which is why no numbers are reported.

\begin{table}[ht]
    \caption{Number of indexable problems among $100\,000$ randomly generated problems.}
    \centering
    \label{tb:number_indexable}
    \begin{tabular}{|c|r|r|r|r|}
        \hline
        Problem size $n$ &  Tridiagonal &  5-diagonal &  7-diagonal &  Dense \\
        \hline
        3  &     98\,731 &        -- &        -- &   99\,883 \\
        4  &     95\,067 &     99\,655 &        -- &   99\,931 \\
        5  &     89\,198 &     99\,309 &     99\,902 &   99\,969 \\
        10 &     54\,129 &     90\,377 &     98\,914 &  100\,000 \\
        30 &      7\,094 &     29\,699 &     66\,143 &  $100\,000$ \\
        50 &      1\,823 &      9\,332 &     32\,069 &  $100\,000$ \\
        \hline
    \end{tabular}
\end{table}

Based on these results, we can assert that dense models are essentially always indexable which conforms with the data reported in \citep{nino2007dynamic}. The situation is, however, radically different for sparse models: the number of indexable problems decreases quickly with the number of states. For instance, there are only $1\,823$ indexable $50$-state problems among $100\,000$ generated tridiagonal models (\emph{i.e.} around $1.8\%$ are indexable). Note that a tridiagonal model is a birth-death Markov chain which is frequently used for queueing systems. Hence, it is very important to check the indexability of the problem because it is not a prevalent property for sparse models. This also calls for new efficient policies in restless multi-arm bandit problems that are not based on Whittle indices.

\section{Extension to the discounted case}
\label{sec:discounted}

The model described in Section~\ref{sec:bandits} corresponds to the definition of Whittle index for a time-average criterion, for which Whittle index is known to be asymptotically optimal \citep{weber1990index} for restless bandits. Yet, Whittle index can also be defined for the discounted case \citep{nino2020fast, akbarzadeh2020conditions}. Notably, the discounted Whittle index simplifies into Gittins index when the bandit is rested (\emph{i.e.}, when $\mP^0=\mI$ and $\vr^0=\boldsymbol{0}$). In this section, we show how to adapt our algorithm to the discounted case. As a by product, we obtain the first subcubic algorithm to compute Gittins index rested bandit.

\subsection{Discounted Whittle index}

We now consider a $\widx$-penalized MDP in which the instantaneous reward received at time $t\ge0$ is discounted by a factor $\beta^t$, where $\beta\in(0,1)$ is called the discount factor: when executing action $a$ in state $i$ at time $t\ge0$, the decision maker earns a reward $\beta^t (r^a_i - \lambda a)$. For a given policy $\pi$, we denote by $u^\pi_i(\lambda)$ the expected sum of discounted rewards earned by the decision maker when the MDP starts in state $i$ at time $0$. The vector $\vu^\pi(\lambda)=[u^\pi_1(\lambda)\ \dots\ u^\pi_n(\lambda)]^\top$ is called the value function of the policy $\pi$. From \citep{putermanMarkovDecisionProcesses1994}, it satisfies Bellman's equation, that is, for all state $i$ we have: 
\begin{align}
    \label{eq:u^pi_discounted}
    u^\pi_i(\lambda) = r_i^{\pi_i} - \lambda \pi_i + \beta \sum_{j=1}^n P^{\pi_i}_{ij} u^\pi_j(\lambda).
\end{align}
The above equation is a linear equation, whose solution is unique because $\beta<1$. It is given by:
\begin{align}
    \label{eq:h is linear discounted}
    \vu^\pi(\lambda) = (\mI - \beta \mP^\pi)^{-1}(\vr^{\pi} - \lambda \vpi).
\end{align}
For a given penalty $\lambda$ and a state $i$, we denote by $u^*_i(\lambda):=\max_\pi u^\pi_i(\lambda)$ be the optimal value of state $i$.  A policy $\pi$ is optimal for the penalty $\lambda$, \ie, $\pi\in\Pi^*(\lambda)$, if for all state $i\in[n]$, $u^*_i(\lambda)=u^\pi_i(\lambda)$. By~\citep{putermanMarkovDecisionProcesses1994} such a policy exists, $\lvert\Pi^*(\lambda)\rvert >0$.
\reviewedFirstRev{
    As mentioned in Section~\ref{sssec:penal_mdp}, the distinction between Bellman optimal and gain optimal disappears in discounted MDP in which we are concerned with maximizing the value function.
    Similarly to the time-average criterion studied before, a $\beta$-discounted RB is called indexable if for all penalty $\lambda<\lambda'$, all $\pi\in\Pi^*(\lambda)$ and $\pi'\in\Pi^*(\lambda')$, one has $\pi\supseteq\pi'$.
}

\subsection{Analogy between the time-average and the discounted versions}

Let $\pi$ be a policy and $i\in\pi$ be an active state.
\reviewedFirstRev{
    Similarly to average reward model studied before, the advantage of action activate over action rest in state $i$ right before following policy $\pi$ is given by,
    $\alpha^\pi_i(\lambda):=r^1_i -r^0_i -\lambda +\beta\sum_{j=1}^n (P^1_{ij} -P^1_{ij})u^\pi_j(\lambda)$.
    Then $\alpha^\pi_i(\lambda)=0$ if and only if $\lambda = \delta_i + \sum_{j=1}^n\tilde{\Delta}_{ij} u_j^\pi(\lambda)$. 
}
where $\delta_i$ is defined as for the average reward model and $\tilde{\mDelta}$ is such that for all\footnote{Note that the definition of $\tilde{\mDelta}$ is identical to $\mDelta$ except when $j=1$, for which $\Delta_{i1}:=0$ but $\tilde{\Delta}_{i1}:=\beta(P^1_{i1}-P^0_{i1})$.} states $i,j\in[n]$: $\tilde{\Delta}_{ij}:=\beta(P^1_{ij}-P^0_{ij})$.

To finish the derivation of the algorithm, one should note that the value function $\vu(\lambda)$ plays the same role as the vector $\vv(\lambda)$ defined for the average reward model.
In particular, the definition of $\vu$ in Equation~\eqref{eq:h is linear discounted} is the analogue of the definition of $\vv$ in \eqref{eq:h is linear} up to the replacement of the matrix $\mA^\pi$ in \eqref{eq:h is linear} by the matrix $\mI-\beta \mP^\pi$.
This means that similarly to $\vv$, the value function $\vu(\widx)$ is affine in $\lambda$.

Hence, following the same development in Section~\ref{sec:formal_widx_algo}, we can modify Algorithm~\ref{algo:whittle_n3} to compute the discounted Whittle index by modifying only the initialization phase:
\begin{align*}
    \tilde{\mDelta}:=\beta(\mP^1-\mP^0) \text{ and } \mX^1:=\tilde{\mDelta}(\mI-\beta \mP^{\pi^1})^{-1}.
\end{align*}

Note that we still have $\vy^1=\pmb{0}$ because $(\mI-\beta \mP^{\pi^1})^{-1}\vpi^1=\frac1{1-\beta}\pmb{1}$ (value function in a $\beta$-discounted Markov reward process with reward equals to $1$ in all states) and $\tilde{\mDelta}\pmb{1}=\pmb{0}$. Also, if Subroutine~\ref{algo:FMM} is used, Line~\ref{algoFMM:recompute} should be changed to $\mX^{k_0}:=\tilde{\mDelta}(\mI-\beta \mP^{\pi^{k_0}})^{-1}$. Last but not least, in the discounted case, \reviewedFirstRev{we no longer need to check if the optimal policies are unichain} because the matrix $(\mI-\beta \mP^\pi)$ is invertible for any policy $\pi$ as long as $\beta<1$ (from Perron-Frobenius' Theorem).

\subsection{Gittins index}

The notion of ``restless'' bandit comes from the fact that even when the action ``rest'' is taken, the Markov chain can still change state and generate rewards. When this is not the case (\emph{i.e.}, when $\mP^0=\mI$ and $\vr^0=\boldsymbol{0}$), an arm is no longer restless and is simply called a Markovian bandit (or a \emph{rested} Markovian bandit if one wants to emphasize that it is not restless).

In a discounted rested bandit, the notion of Whittle index coincides with the notion of Gittins index (In fact, Whittle index was first introduced as a generalization of Gittins index to restless bandit in \citep{whittle1988restless}). In such a case, there is no notion of indexability: a discounted rested bandit is always indexable. Its index can be computed by Algorithm~\ref{algo:whittle_n3} without testing indexability. The best known algorithms to compute Gittins index runs in $(2/3) n^3 +O(n^2)$ \citep{chakravorty2014multi}. When using fast multiplication, our algorithm computes Gittins index in $O(n^{2.5286})$ which makes it the first algorithm to compute Gittins index in subcubic time.  

Note that when $\mP^1=\mI$, it is possible to compute $\mX^1=\tilde{\mDelta}/(1-\beta)$ without having to solve a linear system which means that, when $\mP^1=\mI$, our algorithm with the variant Subroutine~\ref{algo:update_X} has complexity $(2/3) n^3 +O(n^2)$.
This shows that our cubic algorithm can compute the Gittins index also in $(2/3) n^3 +O(n^2)$ instead of $(2/3)n^3 +o(n^3)$.

\section{Conclusion}
\label{sec:conclusion}
In this paper, we \reviewedFirstRev{propose a univocal definition of indexability} and present an algorithm that is efficient for detecting the non-indexability and computing the Whittle index of \reviewedFirstRev{ all indexable finite-state restless bandits whose arms are all unichain. With no assumptions on the structure of arms, this algorithm can still test the indexability and compute the Whittle index of some arms that are multichain and remains efficient if it is able to do so.}
Our algorithm is based on the efficient application of the Sherman-Morrison formula. This is a unified algorithm that works for both discounted and non-discounted restless bandits, and can be used for Gittins index computation.  We present a first version of our algorithm that runs in $n^3+o(n^3)$ arithmetic operations (or in $(2/3)n^3+o(n^3)$ if we do not test the indexability). So, we conclude that Whittle index is not harder to compute than Gittins index. The second version of our algorithm uses the fastest matrix multiplication method and has a complexity of $O(n^{2.5286})$. This makes it the first subcubic algorithm to compute Whittle index or Gittins index. We provide numerical simulations that show that our algorithm is very efficient in practice: it can test indexability and compute the index of a $n$-state restless bandit arm in less than one second for $n=1000$, and in a few minutes for $n=15000$. These numbers are provided for dense matrices. One might expect to have more efficient algorithms if the arm has a sparse structure. We leave this question for future work.

\backmatter


%
%

\bmhead{Acknowledgments} This work is supported by the French National Research Agency (ANR) through REFINO Project under Grant ANR-19-CE23-0015.



\section*{Declarations}


\begin{itemize}
    \item Funding: the French National Research Agency (ANR) through REFINO project under Grant ANR-19-CE23-0015
    \item Conflict of interest/Competing interests: Not applicable
    \item Ethics approval: No ethical concerns
    \item Availability of data and materials: Not applicable
    \item Code availability: \url{https://gitlab.inria.fr/markovianbandit/efficient-whittle-index-computation}
\end{itemize}







\begin{appendices}




\section{Examples and counterexamples}

In this section, we provide a few examples to illustrate the ambiguities in the classical definition of indexability, and to illustrate what can happen for some multichain arms. We also provide the parameters of arms presented in Figure~\ref{fig:illustrate_indexability}.

\subsection{Discussion on the definition of indexability}
\label{apx:discussion_index}

The classical notion of indexability used in the literature is to say that the optimal policy $\pi^*(\lambda)$ should be non-increasing in $\lambda$. Yet, we argue that this definition has two problems:
\begin{enumerate}
    \item What does ``increasing'' mean when $\pi^*(\lambda)$ is not unique? Two possibilities are: for all penalties $\lambda<\lambda'$:
    \begin{enumerate}
        \item[($\exists$)] there exist policies $\pi,\pi'$ with $\pi$ optimal for $\lambda$ and $\pi'$ optimal for $\lambda'$ such that $\pi\supseteq\pi'$;
        \item[($\forall$)] for all policies $\pi,\pi'$ such that  $\pi$ is optimal for $\lambda$ and $\pi'$ is optimal for $\lambda'$, we have $\pi\supseteq\pi'$.
    \end{enumerate}
    \item What notion of ``optimality'' should be used? Two possibilities are: 
    \begin{enumerate}
       \item[(GO)] ``optimal'' means gain optimal.
       \item[(BO)] ``optimal'' means Bellman optimal.
   \end{enumerate}
\end{enumerate}
The most problematic choice is the notion of increasingness: Interpretation $(\exists)$ is more permissive: For instance, consider an arm with two states and assume that the optimal policy is $\{1,2\}$ for $\lambda<0$ and is either $\{1\}$ or $\emptyset$ for $\lambda>0$. Interpretation $(\exists)$ says that the arm is indexable while interpretation ($\forall$) says that this arm is not indexable. If the arm is indexable, what should the index of state $1$ be? Any choice $\lambda_1\in[0,+\infty]$ seems reasonable. Saying that the arm is not indexable clarifies the situation. This is why we choose interpretation ($\forall$) in our paper.

In our paper, we choose the combination ($\forall$-BO) because we believe that, for a problem that has transient state, the notion of Bellman optimality is more meaningful than the notion of gain optimality. Also, our combination ($\forall$-BO) allows for more problems to be indexable compared to ($\forall$-GO) and is easier to characterize.

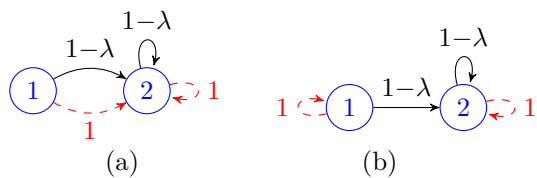
\begin{figure}[ht]
    \centering
    \begin{tabular}{ccc}
        \begin{minipage}{.25\linewidth}
            \begin{tikzpicture}[on grid, state/.style={circle,draw}, >= stealth', auto, prob/.style = {inner sep=1pt,font=\scriptsize}]
                \node[state,color=blue]  (A) {$2$};
                \node[state,color=blue]  (B) [left =1.5cm of A]   {$1$};
                \path[->]
                    (A) edge[loop above,color=black]  node{$1{-}\lambda$} (A)
                    (A) edge[loop right, color=red, dashed]     node{$1$} (A)
                    (B) edge[bend left, color=black]     node{$1{-}\lambda$} (A)
                    (B) edge[bend right, color=red, dashed]     node[below]{$1$} (A);
            \end{tikzpicture}
        \end{minipage}
        &
        \begin{minipage}{.25\linewidth}
            \begin{tikzpicture}[on grid, state/.style={circle,draw}, >= stealth', auto, prob/.style = {inner sep=1pt,font=\scriptsize}]
            \node[state,color=blue]  (A) {$2$};
            \node[state,color=blue]  (B) [left =1.5cm of A]   {$1$};
            \path[->]
                (A) edge[loop above,color=black]  node{$1{-}\lambda$} (A)
                (A) edge[loop right, color=red, dashed]     node{$1$} (A)
                (B) edge[color=black]     node{$1{-}\lambda$} (A)
	            (B) edge[loop left, color=red, dashed]     node[left]{$1$} (B);
            \end{tikzpicture}
        \end{minipage}\\
        (a) & (b) 
    \end{tabular}

    \caption{Ambiguous examples: All transitions are deterministic and labels on transitions indicate rewards. Solid black arrows correspond to the action ``activate'' and dashed red arrows to the action ``rest''.
}
    \label{fig:ambiguous_example}
\end{figure}

We illustrate these different definitions in Figure~\ref{fig:ambiguous_example}. For example (a):
\begin{itemize}
    \item The gain optimal policies are $\{1,2\}$ and $\{2\}$ for $\lambda<0$, and $\{1\}$ and $\emptyset$ for $\lambda>0$: According to the interpretation~$(\exists)$, the problem should be indexable but the index for state $1$ is unclear. According to the interpretation~$(\forall)$, the problem should not be indexable.
    \item The Bellman optimal policy is $\{1,2\}$ for $\lambda<0$, and $\emptyset$ for $\lambda>0$. According to our definition, ($\forall$-BO), the problem is indexable and the indices are $\lambda_1=\lambda_2=0$.
\end{itemize}
For example (b), the Bellman optimal and gain optimal policies are identical and equal to the gain optimal policies of example (a). Hence, example (b) is not indexable according to our definition. The output of our algorithm for this problem is "multichain". 

Note that if the distinction between (BO) and (GO) disappears for discounted problems, the distinction between ($\forall$) and ($\exists$) remains.

{
\subsection{Possible outputs for multichain arms}
\label{apx:multichain2}

When running our algorithm on a multichain arm, it outputs "multichain" if one of the policies $\pi^k$ is multichain. This suggests that our algorithm will not necessarily output ``multichain'' for all multichain arms because it explores only a small subset of the policies. In Figure~\ref{fig:example_multichain2}, we provide two examples that illustrate this case. The two are multichain arms for which our algorithm is able to identify the indexability.

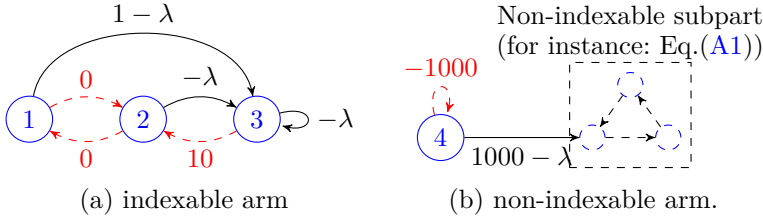
\begin{figure}[ht]
    \begin{tabular}{cc}
        \begin{tikzpicture}[on grid, state/.style={circle,draw}, >= stealth', auto, prob/.style = {inner sep=1pt,font=\scriptsize}]
            \node[state,color=blue]  (A) {$1$};
            \node[state,color=blue]  (B) [right =1.5cm of A]   {$2$};
            \node[state,color=blue]  (C) [right =1.5cm of B]   {$3$};
            \path[->] 
            (A) edge[bend left,color=red, dashed]  node{$0$} (B)
            (B) edge[bend left,color=red, dashed]  node{$0$} (A)
            (C) edge[bend left,color=red, dashed]  node{$10$} (B)
            (A) edge[bend left, color=black,out=80,in=100]     node{$1-\lambda$} (C)
            (B) edge[bend left, color=black]     node{$-\lambda$} (C)
            (C) edge[loop right, color=black]     node{$-\lambda$} (C);
        \end{tikzpicture}
        &
        \begin{tikzpicture}[on grid, state/.style={circle,draw}, >= stealth', auto, prob/.style = {inner sep=1pt,font=\scriptsize}]
            \node[state,color=blue] at (-0.5,-0.3) (A) {$4$};
            \node[state,dashed,color=blue] at (1.5,-0.3) (B) {};
            \node[state,dashed,color=blue] at (2.5,-0.3) (C) {};
            \node[state,dashed,color=blue] at (2,0.4) (D) {};
            \draw[dashed] (1.2,-.7) rectangle (2.8,.7);
            \node at (2,1.3) {Non-indexable subpart};
            \node at (2,0.9) {(for instance: Eq.\eqref{eq:exemple_non_idexable})};

            \path[->] 
            (A) edge[loop above,color=red, dashed]  node{$-1000$} (B)
            (B) edge[dashed] (C) (C) edge[dashed] (D) (D) edge[dashed] (B)
            (A) edge[color=black]  node[below]{$1000-\lambda$} (B);
        \end{tikzpicture}\\
        (a) indexable arm& (b) non-indexable arm.
    \end{tabular}
    \caption{Two examples of multichain arms for which our algorithm does not return ``multichain'' but returns ``indexable'' (a) or ``non-indexable'' (b).}
    \label{fig:example_multichain2}
\end{figure}

In the first example shown in Figure~\ref{fig:example_multichain2}(a), the arm is multichain because the policy $\{3\}$ has two recurrent classes: $\{1,2\}$ and $\{3\}$. Yet, this arm is indexable and the indices are $\{11,8,-10\}$. Our algorithm will output that this arm is indexable because it will explore the sequence of policies $\pi^1,\pi^2,\pi^3,\pi^4$, where
\begin{itemize}
    \item $\pi^1=\{1,2,3\}$ is the unique Bellman optimal policy for $\lambda<-10$;
    \item $\pi^2=\{1,2\}$ is the unique Bellman optimal policy for $\lambda\in(-10,8)$;
    \item $\pi^3=\{1\}$ is the unique Bellman optimal policy for $\lambda\in(8,11)$;
    \item $\pi^4=\emptyset$ is the unique Bellman optimal policy for $\lambda>11$. 
\end{itemize}
All these policies are unichain, and the policy $\{3\}$ will never be explored.  Hence, our algorithm will output "indexable" for this case and will compute the indices.

In the second example, shown in Figure~\ref{fig:example_multichain2}(b), we construct a non-indexable arm by taking the non-indexable $3$-state example shown in Figure~\ref{fig:illustrate_algo_nind} (parameters are given in \eqref{eq:exemple_non_idexable}) to which we add an extra state ``4''. For this state, the active action has a very high reward ($1000$) and leads to the non-indexable recurrent class. The passive action has a very low reward and stays in state $4$. Any policy that does not activate $4$ is multichain. The algorithm will start by exploring policies that activate the state $4$. As for the original example presented in Figure~\ref{fig:illustrate_algo_nind}, our algorithm will realize that the arm is non-indexable when exploring values around $\lambda\approx0.70$. The algorithm will stop and answer ``not indexable'' before  trying  the passive action for state $4$  because the active avantage for state $4$ is larger than $2000-\lambda$. The output of the algorithm is thus "non-indexable".
}

\subsection{Multichain arms and infinite indices}
\label{apx:multichain}

\begin{figure}[ht]
    \centering
    \begin{minipage}{.35\linewidth}
        \centering
        \begin{tikzpicture}[on grid, state/.style={circle,draw}, >= stealth', auto, prob/.style = {inner sep=1pt,font=\scriptsize}]
            \node[state,color=blue]  (A) {$2$};
            \node[state,color=blue]  (B) [left =1.5cm of A]   {$1$};
            \path[->]
            (A) edge[loop right,color=black]  node{$1-\lambda$} (A)
            (A) edge[loop above, color=red, dashed]     node{$1$} (A)
            (B) edge[color=black]     node{$1-\lambda$} (A)
            (B) edge[loop above, color=red, dashed]     node{$0$} (B);
        \end{tikzpicture}\\
        (a) Our algorithm returns ``indexable''.
    \end{minipage}\qquad 
    \begin{minipage}{.35\linewidth}
        \centering
        \begin{tikzpicture}[on grid, state/.style={circle,draw}, >= stealth', auto, prob/.style = {inner sep=1pt,font=\scriptsize}]
            \node[state,color=blue]  (A) {$2$};
            \node[state,color=blue]  (B) [left =1.5cm of A]   {$1$};
            \path[->]
            (A) edge[loop right,color=red, dashed]  node{$0$} (A)
            (A) edge[loop above, color=black]     node{$1-\lambda$} (A)
            (B) edge[color=red, dashed]     node{$0$} (A)
            (B) edge[loop above, color=black]     node{$-\lambda$} (B);
        \end{tikzpicture}\\
        (b) Our algorithm returns ``multichain''.
    \end{minipage}
    
    \caption{Example of an indexable multichain problem with infinite Whittle index. Transitions are deterministic and labels on edges indicate rewards (for the $\lambda$-penalized arm). Solid black transitions correspond to action ``activate''  and dashed red to the action ``rest''.}
    \label{fig:example_multichain}
\end{figure}
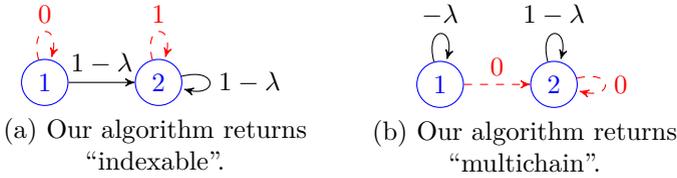 

Consider the two examples of Figure~\ref{fig:example_multichain}. The examples are multichain because policy $\emptyset$ has two irreducible classes for example (a) and policy $\{1,2\}$ has two irreducible classes for example (b). These two problems are indexable:
\begin{itemize}
    \item For (a), the Bellman optimal policy for $\lambda<0$ is $\{1,2\}$ and $\{1\}$ for $\lambda>0$. The indices are $\lambda_2=0$ and $\lambda_1=+\infty$.
    \item For (b), the  Bellman optimal policy is $\{2\}$ for $\lambda<0$  and $\{1,2\}$ for $\lambda>0$. The indices are $\lambda_1=0$ and $\lambda_2=-\infty$.
\end{itemize}
For the first example, our algorithm returns the correct indices because the constructed policies are $\pi^1:=\{1,2\}\supsetneq\pi_2:=\{1\}$ and they are both unichain.  For the second example, our algorithm will start with the policy $\{1,2\}$ and will stop by saying that this example is multichain.

\subsection{Parameters for the example of Figure~\ref{fig:illustrate_indexability}}
\label{apx:non_indexable_example}

The numerical data of the indexable arm presented in Figure~\ref{fig:illustrate_vf} is
\begin{equation*}
    \mP^0{=}\begin{bmatrix}
        0.363 & 0.503 & 0.134 \\
        0.082 & 0.754 & 0.164 \\
        0.246 & 0.029 & 0.724
    \end{bmatrix}
    \mP^1{=}\begin{bmatrix}
        0.172 & 0.175 & 0.653 \\
        0.055 & 0.931 & 0.014 \\
        0.155 & 0.627 & 0.218
    \end{bmatrix}
    \vr^1{=}\begin{bmatrix}
        0.441 \\
        0.803 \\
        0.426
    \end{bmatrix}
    \vr^0{=}\vzero
\end{equation*}

\noindent
The numerical data of the non-indexable arm presented in Figure~\ref{fig:illustrate_non_indexable} is

\begin{equation}
    \label{eq:exemple_non_idexable}
    \mP^0{=}\begin{bmatrix}
        0.005 & 0.793 & 0.202 \\
        0.027 & 0.558 & 0.415 \\
        0.736 & 0.249 & 0.015
    \end{bmatrix}
    \mP^1{=}\begin{bmatrix}
        0.718 & 0.254 & 0.028 \\
        0.347 & 0.097 & 0.556 \\
        0.015 & 0.956 & 0.029
    \end{bmatrix}
    \vr^1{=}\begin{bmatrix}
        0.699 \\
        0.362 \\
        0.715
    \end{bmatrix} \vr^0{=}\vzero
\end{equation}

\section{Technical lemmas}

\subsection{Unicity of Bellman optimal policy}
\label{apx:unicity_BO}

\subsubsection{Definition and notation}

We are given a MDP $\langle\gS,\gA,r,P\rangle$ with finite state and action spaces.
As shown in \cite[Chapter~9]{putermanMarkovDecisionProcesses1994}, for such a MDP, the optimal gain $\vg^*$ is a vector that satisfies the \emph{multichain optimality equations}: for each $i\in\gS$:
\begin{align}
    \max_{a\in\gA}{\Big(\sum_{j\in\gS}P^a_{ij}g^*_j-g^*_i\Big)} =0 \label{eq:gain_max}\\
    \max_{a\in\gA}{\Big(r^a_i -g^*_i +\sum_{j\in\gS}P^a_{ij}h_j-h_i\Big)} =0. \label{eq:bias_max}
\end{align}
This system uniquely determines the optimal gain that we denote by $\vg^*$.
However, vector $\vh$ is not uniquely determined by the system.
In the following, we denote by $H$ the set of bias vector $\vh$ such that $(\vg^*,\vh)$ is a solution of the optimality equations \eqref{eq:gain_max}--\eqref{eq:bias_max}. A policy $\pi$ is Bellman optimal if there exists a bias $\vh\in H$ such that policy $\pi$ attains the maximum \eqref{eq:bias_max}, \emph{i.e.}, for all $i\in\gS$:
\begin{align*}
    \pi_i \in \argmax_{a\in\gA}{\Big(r^a_i -g_i +\sum_{j\in\gS}P^a_{ij}h_j-h_i\Big)} = \argmax_{a\in\gA}{\Big(r^a_i +\sum_{j\in\gS}P^a_{ij}h_j\Big)}.
\end{align*}

For a given policy $\pi:\gS\mapsto\gA$, we denote by $\vr^\pi$ and $\mP^\pi$ the reward vector and state transition matrix under policy $\pi$: $r^\pi_i=r^{\pi_i}_{i}$ and $P^\pi_{ij} = P^{\pi_i}_{ij}$. Let $\vg^\pi,\vh^\pi\in\real^{\vert \gS\vert}$ be a solution of the following system:
\begin{align}
    \vg^\pi - \mP^{\pi}\vg^\pi &=\vzero \label{eq:gain_eval}\\
    \vr^{\pi}  - \vg^\pi+\mP^{\pi}\vh^\pi  -\vh^\pi &= \vzero \label{eq:bias_eval1}.
\end{align}
The vector $\vg^\pi$ is uniquely determined by this system of equations and is called the long-run average reward or gain of policy $\pi$. The vector $\vh^\pi$ is unique up to an element of the null space of $(\mI-\mP^\pi)$. Such a vector $\vh^\pi$ is called a bias of policy $\pi$.

If $(\vg^\pi, \vh^\pi)$ is a solution of \eqref{eq:gain_eval} and \eqref{eq:bias_eval1}, then the advantage of action $a$ over the action $\pi_i$ when the MDP is in state $i$ is given by:
\begin{align*}
    B^a_i(\vh^\pi) &:=r^a_i +\sum_{j\in\gS}P^a_{ij}h^\pi_j -g^{\pi}_i -h^\pi_i\\
    &= r^a_i - r^{\pi_i}_i +\sum_{j\in\gS}(P^a_{ij}-P^{\pi_i}_{ij})h^\pi_j
\end{align*}

We recall the two notions of optimality:
\begin{itemize}
    \item A policy $\pi$ is gain optimal if its gain $\vg^\pi$ equals $\vg^*$.
    \item A policy $\pi$ is Bellman optimal if $\vg^*=\mP^\pi \vg^*$ and if there exists a bias $\vh^\pi$ that is a solution of \eqref{eq:bias_eval1} and also \eqref{eq:bias_max}.
\end{itemize}
Recall that a Bellman optimal policy is also gain optimal.

Note that by the definition the advantage $B(\cdot)$, a policy is Bellman optimal if $\vg^*=\mP^\pi \vg^*$ and if there exists $\vh\in H$ such that $B^{\pi_i}_i(\vh)=0$.


\paragraph*{Useful notations}

Let $\pi:\gS\to\gA$ be a policy. We say that a state is recurrent for $\pi$ if it is recurrent for the Markov chain whose transition matrix is $\mP^\pi$. In other words, a state $i\in\gS$ is recurrent if when the chain starts in $i$ at time $0$, it almost surely visits the state $i$ at some time $t\ge1$.  We denote by $\gR^\pi$ the set of recurrent states of policy $\pi$. 

We also define $\bar{\mP}^\pi$ as the Cesaro limit of the sequence $\{(\mP^\pi)^t\}_{t=1}$:
\begin{align*}
    \bar{\mP}^\pi:=\lim_{T\to\infty}\frac1T\sum_{t=1}^T(\mP^\pi)^{t-1}.
\end{align*}
From \cite[Section~A.4 of Appendix~A]{putermanMarkovDecisionProcesses1994}, the matrix  $\bar{\mP}^\pi$ exists and has the following properties:
\begin{itemize}
    \item $\bar{\mP}^\pi$ is a stochastic matrix, and satisfies $\mP^\pi \bar{\mP}^\pi = \bar{\mP}^\pi\mP^\pi =\bar{\mP}^\pi$.
    \item For all state $i,j$, if $j\not\in\gR^\pi$, then $\bar{P}^\pi_{ij}=0$.
    \item If $\pi$ is unichain, then the rows of $\bar{\mP}^\pi$ are identical. 
\end{itemize}

\subsubsection{Characterization of gain optimal policies}

The following lemma characterizes gain optimal policy by showing that the policy must satisfies \eqref{eq:bias_max} on their recurrent states.
\begin{lemma}
    \label{lem:opt_pol}
    Let $\pi:\gS\mapsto\gA$ be a policy and recall that $\gR^\pi$ is the set of recurrent states of policy $\pi$.
    The three properties below are equivalent.
    \begin{enumerate}[label=(\roman*)]
        \item \label{it:opt_pol1} $\mP^\pi\vg^*=\vg^*$ and for all $\vh\in H$, $B^{\pi_i}_i(\vh)=0$ for all $i\in\gR^\pi$
        \item \label{it:opt_pol2} $\mP^\pi\vg^*=\vg^*$ and for some $\vh\in H$, $B^{\pi_i}_i(\vh)=0$ for all $i\in\gR^\pi$
        \item \label{it:opt_pol3} $\pi$ is gain optimal.
    \end{enumerate}
\end{lemma}
\begin{proof}
    \ref{it:opt_pol1} $\Rightarrow$ \ref{it:opt_pol2} is trivial.

    \ref{it:opt_pol2} $\Rightarrow$ \ref{it:opt_pol3}: By definition of $B_i^a(\vh)$, we have $r^{\pi}_i-g^*_i = h_i-\sum_{j\in\gS} P^{\pi}_{ij}h_j$ for any recurrent state $i$ of $\pi$.  Multiply this with $\bar{P}^\pi_{ki}$ and sum over $i\in\gS$ (if $i$ is not recurrent, then $\bar{P}^\pi_{ki}=0$) gives
    \begin{align*}
        \sum_{i\in\gS} \bar{P}^\pi_{ki}(r^{\pi}_i-g^*_i) = \sum_{i\in\gS} \bar{P}^\pi_{ki}h_i - \underbrace{\sum_{i\in\gS} \bar{P}^\pi_{ki}\sum_{j\in\gS} P^{\pi}_{ij}h_j}_{=\sum_{j\in\gS} \bar{P}^\pi_{kj}h_j \text{ since $\bar{\mP}^\pi\mP^\pi=\bar{\mP}^\pi$.}}
        &= \vzero.
    \end{align*}
    By Theorem~8.2.6 of \cite{putermanMarkovDecisionProcesses1994}, the average reward of $\pi$ is $\bar{\mP}^\pi r^\pi$. The above equation shows that $\bar{\mP}^\pi\vr^\pi = \bar{\mP}^\pi\vg^*$. Moreover, the assumption  $\mP^\pi\vg^*=\vg^*$ implies that $\bar{\mP}^\pi\vg^*=\vg^*$ which in turn implies that $\bar{\mP}^\pi\vr^\pi=\vg^*$. This shows that the average reward of $\pi$ is $\vg^*$ and therefore $\pi$ is gain optimal.

    \ref{it:opt_pol3} $\Rightarrow$ \ref{it:opt_pol1}: If $\pi$ is gain optimal, then $\mP^\pi \vg^*=\vg^*$ and $\bar{\mP}^\pi(\vr^\pi-\vg^*)=\vzero$. The latter rewrites as $\sum_{i\in\gS}\bar{P}^\pi_{ki}(r^{\pi}_i-g^*_i) =0$ for all state $k$. Let $\vh\in H$ be an optimal bias.  For all state $k$, we have    
    \begin{align}
        \label{eq:apx_proof_H}
        \sum_{i\in\gS}\bar{P}^\pi_{ki}B^{\pi_i}_i(\vh) = \sum_{i\in\gS}\bar{P}^\pi_{ki}(r^{\pi}_i-g^*_i +\sum_{j\in\gS}P^{\pi}_{ij}h_j -h_i) &=0.
    \end{align}
    As $\vh$ satisfied \eqref{eq:bias_max}, for all action $a$, we have $B^{a}_i(\vh) \le0$ for all states $i\in\gS$ and in particular $B^{\pi_i}_i(\vh) \le0$. This shows that for any state $i$ such that $\bar{P}^\pi_{ki}>0$, one must have $B^{\pi_i}_i(\vh) =0$. Such state $i$ are the recurrent states of $\pi$. This shows that $B^{\pi_i}_i(\vh) =0$ for all $i\in\gR^\pi$.
\end{proof}

\subsubsection{Characterization of Bellman optimal policies}

The previous lemma shows that a policy is gain optimal if and only if the actions for the recurrent states of the policy satisfy \eqref{eq:bias_max}.
The following lemma shows the relationship between two policies that are unichain and satisfy \eqref{eq:bias_max} on all states.

\begin{lemma}
    \label{lem:equi_bias}
    Suppose that two policies $\pi$ and $\theta$ are Bellman optimal, unichain and have at least one common recurrent state: $\gR^\pi\cap\gR^\theta\neq\emptyset$.
    
    Then for any $\vh^\pi$ and $\vh^\theta$ solutions of \eqref{eq:bias_eval1} for $\pi$ and $\theta$ respectively, there exists a constant $c$ such that for all state $i$: $h^\pi_i -h^\theta_i =c$. Moreover, in this case, ${B_i^{\theta_i}(\vh^\pi)=B_i^{\pi_i}(\vh^\theta)=0}$ for all $i$.
\end{lemma}
\begin{proof}
    Since $\pi$ and $\theta$ are Bellman optimal, $\vh^\pi,\vh^\theta\in H$.
    In consequence, we have
    \begin{align*}
        \vh^\pi \ge \vr^\theta -\vg^* +\mP^\theta\vh^\pi.
    \end{align*}
    By Lemma~\ref{lem:opt_pol}~\ref{it:opt_pol1}, the above inequality is an equality for all $i\in\gR^\theta$ because $\theta$ is gain optimal.

    As $\vh^\theta$ satisfies \eqref{eq:bias_eval1}, we have
    \begin{align*}
        \vh^\theta -\vh^\pi &\le \vr^\theta -\vg^* +\mP^\theta\vh^\theta -(\vr^\theta -\vg^* +\mP^\theta\vh^\pi) = \mP^\theta(\vh^\theta -\vh^\pi),
    \end{align*}
    with equality for all state $i\in\gR^\theta$. This shows that for all $t$, $\vh^\theta -\vh^\pi\le (\mP^\theta)^t(\vh^\theta -\vh^\pi)$ which implies that $\vh^\theta -\vh^\pi\le \bar{\mP}^\theta(\vh^\theta -\vh^\pi)$ with equality for all states $i\in\gR^\theta$. Similarly, $\vh^\pi -\vh^\theta \le \bar{\mP}^\pi(\vh^\pi -\vh^\theta)$ with equality for any state $i\in\gR^\pi$.

    Let $c^\pi_i=\sum_{j\in\gS}\bar{P}^\pi_{ij}(h_j^\pi-h_j^\theta)$ and $c^\theta_i=\sum_{j\in\gS}\bar{P}^\theta_{ij}(h_j^\pi-h_j^\theta)$. By what we have just shown, for all state $i$, we have
    \begin{align*}
        c^\theta_i \underbrace{\le}_{\text{equality if $i\in\gR^\theta$}} h_i^\pi-h_i^\theta \underbrace{\le}_{\text{equality if $i\in\gR^\pi$}} c^\pi
    \end{align*}
    As both policies are unichain, $c^\pi_i$ and $c^\theta_i$ do not depend on $i$.  Moreover, if there exists $i\in\gR^\theta\cap\gR^\pi$, we have $c^\pi_i=c^\theta_i =: c$. In consequence,  $h_i^\pi-h_i^\theta=c$ for all state $i$. 
\end{proof}

\subsubsection{Unicity of Bellman optimal policy}
\label{ssec:unicity}


\begin{lemma}
    \label{lem:unicity_BO}
    Let $\pi$ be a Bellman optimal policy that is unichain. If $\pi$ is not the unique Bellman optimal policy, then there exists a state $i$ and an action $a\neq\pi_i$ such that $B_i^a(\vh^\pi)=0$.
\end{lemma}
\begin{proof}

    Let $\theta\neq\pi$ be another Bellman optimal policy. Since $\theta$ is gain optimal and $\vh^\pi\in H$, Lemma~\ref{lem:opt_pol} implies that $B_i^{\theta_i}(\vh^\pi) =0$ for all $i\in\gR^\theta$. If there exists $i\in\gR^\theta$ such that $\theta_i\neq\pi_i$, then the proof is concluded.  Otherwise, $\theta_i=\pi_i$ for all $i\in\gR^\theta$. This show that $\pi$ and $\theta$ coincide for all recurrent states of $\theta$ and that $\gR^\theta=\gR^\pi$. Moreover, as $\pi$ is unichain, $\theta$ is also unichain. Hence, Lemma~\ref{lem:equi_bias} implies that $B_i^{\theta_i}(\vh^\pi)=0$ for all $i$. Since $\theta\neq\pi$, there exists at least one state $i\in\gS$ such that $\theta_i\neq\pi_i$.
\end{proof}

\subsection{Unichain property}

\begin{lemma}
\label{lem:invertible}
Given a two-action MDP $\langle [n], \{0,1\}, r, P\rangle$, let $\mP^\pi$ be the transition matrix under a Bellman optimal policy $\pi$.
\reviewedFirstRev{Policy $\pi$ is \emph{unichain} if and only if the matrix}
\begin{align*}
    \mA^{\pi}
        = \left[\begin{array}{ccccc}
                1 & - P_{12}^{\pi} & \dots & -P_{1n}^{\pi}\\
                1 & 1-P_{22}^{\pi} & \dots & -P_{2n}^{\pi}\\
                \vdots\\
                1 &  -P_{n2}^{\pi} & \dots &1-P_{nn}^{\pi}
    \end{array}\right]
\end{align*}
is invertible.
\end{lemma}
\begin{proof}
    $\mA^\pi$ is not invertible if there exists a column vector $\bff{u}\neq\bff{0}$ such that $\bff{u}^\top \mA^\pi=\bff{0}$.
    We prove that such $\bff{u}$ does not exist when \reviewedFirstRev{policy $\pi$} is unichain.
    Let $\bff{u}\in\real^n$ be an arbitrary vector such that $\bff{u}^\top \mA^\pi=\bff{0}$.
    Then, we have
    \begin{align*}
        \begin{cases}
            \sum_{i=1}^{n}u_i &= 0 \\
            u_i-\sum_{j=1}^{n}u_jP_{ji}^\pi &= 0, \text{ for } 2\le i\le n
        \end{cases}
    \end{align*}
    Combining the above equation with $\sum_j P^\pi_{ji}=1$, we get:
    \begin{align*}
        u_1 &= -\sum_{i=2}^n u_i
        =-\sum_{i=2}^n\sum_{j=1}^n u_j P^\pi_{ji}
        =-\sum_{j=1}^n u_j (1-P^\pi_{j1})
        = \sum_{j=1}^n u_j P^\pi_{j1},
    \end{align*}
    where we used that $\sum_{i=1}^{n}u_i=0$ to obtain the last equality.  
    This shows that 
    \begin{align*}
        \begin{cases}
            \sum_{i=1}^{n}u_i &= 0 \\
            \bff{u}^\top\mP^\pi &=\bff{u}^T
        \end{cases}
    \end{align*}
    \reviewedFirstRev{The set of vector $\vu$ such that $\bff{u}^\top\mP^\pi$ is a vector space. It is of dimension $1$ if and only if $\pi$ is unichain, in which case the vector} $\vu$ verifying $\bff{u}^\top\mP^\pi=\bff{u}^\top$ are multiples of a stationary distribution under policy $\pi$ \citep{putermanMarkovDecisionProcesses1994}. \reviewedFirstRev{Thus, if the policy $\pi$ induces a unichain Markov chain, then } $\sum_{i=1}^{n}u_i = 0$ implies $\vu = \vzero$. \reviewedFirstRev{If policy $\pi$ is not unichain, there exists $\vu\ne\vzero$ such that $\sum_{i=1}^{n}u_i = 0$.}
\end{proof}

\section{Implementations}
\label{apx:implementation}

\subsection{Arithmetic complexity of Subroutine~\ref{algo:update_X} and memory usage}

Recall that in Subroutine~\ref{algo:update_X}, we compute the values $X^{\ell}_{ij}$ by doing the update (for all iteration $k$, for all $\ell=1$ to $k$ and for all $i\in[n]$ or all $i\in\pi^{\ell+1}$ if we do not test indexability):
\begin{align}
    \label{eq:apx_update}
    X_{i\sigma^{k}}^{\ell+1} = X_{i\sigma^{k}}^{\ell} -\displaystyle\frac{X^{\ell}_{i\sigma^{\ell}}}{1+X^{\ell}_{\sigma^{\ell}\sigma^{\ell}}}X^{\ell}_{\sigma^{\ell}\sigma^{k}}
\end{align}

If we test indexability, there are $\sum_{k=1}^nk n = n^3/2+O(n^2)$ such updates. If we do not test indexability, there are $\sum_{k=1}^n \sum_{\ell=1}^k (n-\ell) = n^3/3 + O(n^2)$ such updates.  Below, we show each update of Equation~\eqref{eq:apx_update} can be done in two arithmetic operations (one addition and one multiplication), which leads to the complexity of $n^3+O(n^2)$ (or $(2/3)n^3+O(n^2)$) arithmetic operations for the computation of all the needed $X^k_{ij}$.  We also show how to reduce the memory size to $O(n^2)$.

Let $W_{i\ell} := X^\ell_{i \sigma^\ell}/(1+X^\ell_{\sigma^\ell\sigma^\ell})$ and $V_i :=X^\ell_{i\sigma^k}$. Using this, Equation~\eqref{eq:apx_update} can be rewritten as: 
\begin{align}
    \label{eq:apx_update_V}
    V_{i} = V_{i} - W_{i\ell} V_{\sigma^{\ell}}.
\end{align}

This results in the following loop at iteration $k$:
\begin{itemize}
    \item Initialize $V_{i}$ from $X^1_{i\sigma^k}$. 
    \item For all $\ell\in\{1,\dots, k-1\}$, and all $i\in[n]$ (or $i\in\pi^{\ell+1}$), apply \eqref{eq:apx_update_V}. 
    \item Compute $W_{ik}= V_{i}/(1+V_{\sigma^k})$
\end{itemize}
Note that the value of $\mV$ is not necessary for iteration $k$ (only the values of $W_{i\ell}$ are needed). This shows that the algorithm can be implemented with a memory $O(n^2)$.

\subsection{Speedup when not checking the indexability: First found go last} 
When the indexability is not tested, the update \eqref{eq:apx_update_V} is computed for all $i\in\pi^{\ell+1}$. This creates inefficiencies (due to  inefficient cache usage) because the elements $V_{i}$ are not accessed sequentially.

To speedup the memory accesses, our solution is to sort the items during the execution of the algorithm. At iteration $k$, the algorithm computes $\sigma^{k}$. When this is done, our implementation switches all quantities in positions $\sigma^{k}$ and $n-k+1$. These quantities are $\delta, y, z, \mW$ and $\mX$. For instance, once $\sigma^1$ is found, we know that the state at position $n$ is state $n$ and we do the following switches:
\begin{align*}
    \delta_{\sigma^1}, \delta_{n} &\rightarrow \delta_{n}, \delta_{\sigma^1} \\
    y^1_{\sigma^1}, y^1_{n} &\rightarrow y^1_{n}, y^1_{\sigma^1} \\
    z^1_{\sigma^1}, z^1_{n} &\rightarrow z^1_{n}, z^1_{\sigma^1} \\
    \mW_{\sigma^1 :}, \mW_{n :} &\rightarrow \mW_{n:}, \mW_{\sigma^1:} \\
    \text{ and } \mX_{\sigma^1 :}, \mX_{n :} &\rightarrow \mX_{n :}, \mX_{\sigma^1 :}.
\end{align*}

To do so, we need an array to store all states such that at iteration $k$, the first $n-k$ states of the array are the active states. We will need to track the position of each state in such array.

\section{Analysis of the experimental time to solve a linear system}
\label{apx:inversion}

In this section, we report in Figure~\ref{fig:benchmark_inverse} the time taken by the default implementation to solve a linear system of the form $\mA\mX=\mB$ where $\mA$ and $\mB$ are two square matrices. To obtain this figure, we generated random (full) matrices where each entry is between $0$ and $1$ and use the function  \texttt{scipy.linalg.solve} from the library \texttt{scipy}. The reported numbers suggest that the complexity of the solver is closer to $O(n^{2.8})$ than to $O(n^3)$, although we agree that the difference between the $O(n^{2.8})$ and the $O(n^3)$ curves is small. Note that this is in accordance with the papers \citep{huang2016strassen,huang2018practical} that claim that the fastest implementations of matrix multiplication and inversion are based on Strassen's algorithm and should therefore be in $O(n^{2.8})$. 

\begin{figure}[ht]
    \centering
    \begin{tabular}{cc}
        \begin{minipage}{0.5\linewidth}
            \includegraphics[width=\linewidth]{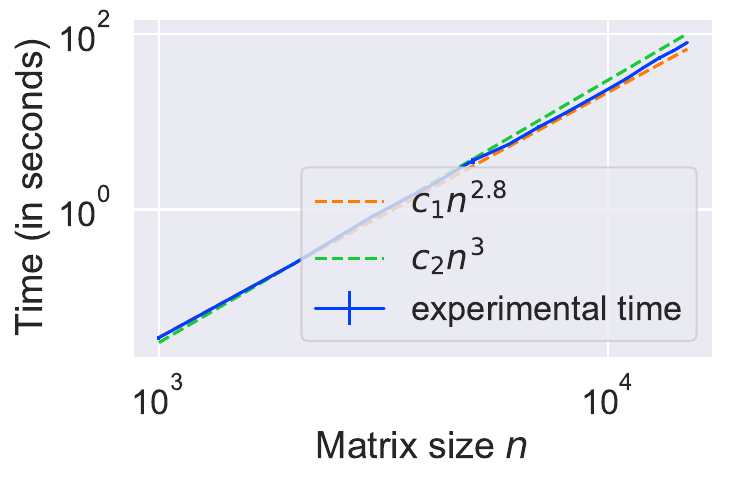}
        \end{minipage}
        &\begin{minipage}{0.45\linewidth}
            \begin{tabular}{|c|c|}
                \hline
                $n$ & Time (in second)\\\hline
                1000 & $ 0.03\pm0.02$ \\
                2000 & $ 0.24\pm0.01$ \\
                4000 & $ 1.83\pm0.05$ \\
                5000 & $ 3.6\pm0.3$ \\
                7000 & $ 8.7\pm0.5$ \\
                10000 & $23.7\pm0.6$ \\
                13000 & $54.0\pm2.5$ \\
                15000 & $80.8\pm1.8$ \\                
                \hline
            \end{tabular}
        \end{minipage}
    \end{tabular}
    \caption{Time taken of the default implementations \texttt{scipy.linalg.solve} of scipy to solve a linear system $\mA\mX=\mB$ where $\mA$ and $\mB$ are two square $n\times n$ matrices.}
    \label{fig:benchmark_inverse}
\end{figure}

\section{Detailed comparison with \texorpdfstring{\citep{akbarzadeh2020conditions} and \citep{nino2020fast}}{Akbarzadeh et al. 2020 and Nino Mora 2020.}}
\label{apx:comparison}

In this section, we compare our algorithm with two main related works for finite-state restless bandits problem.

\subsection{Comparison with \texorpdfstring{\citep{akbarzadeh2020conditions}}{Akbarzadeh2020 et al. 2022}}
The paper presents an  algorithm that computes Whittle indices  in $O(n^3)$ (no explicit constant before $n^3$ is given) for all indexable problems. Despite following a different approach, our algorithm for computing Whittle index can be viewed as a refinement of this work. Let us recall once again that our approach also allows one to check the indexability of general restless bandits.

In the following, we show how we can refine the work of \citep{akbarzadeh2020conditions} to obtain an algorithm that is exactly the same as ours.
Let $D^{\pi}$ and $N^{\pi}$ be two vectors defined as in \citep{akbarzadeh2020conditions} (we use the same notation,  $D^{\pi}$ and $N^{\pi}$, as the cited paper),
\begin{align*}
    D^{\pi} =(1-\beta)(\mI-\beta \mP^\pi)^{-1}\vr^{\pi}, \quad\text{and}\quad
    N^{\pi} =(1-\beta)(\mI-\beta \mP^\pi)^{-1}\vpi.
\end{align*}
Then, we have $D^{\pi} -\widx N^{\pi}=(1-\beta)\vu^{\pi}(\widx)$ where $\vu^{\pi}(\widx)$ is defined as in \eqref{eq:h is linear discounted}. In our proposition, at each iteration $k$, we compute $\mu_i^k$ by Line~\ref{algo2:muki}. Instead, it is defined in \citep{akbarzadeh2020conditions} by two steps:
\begin{enumerate}
\item for all state $j\in[n]$ such that $N_j^{\pi^k\setminus\{i\}}{\neq} N_j^{\pi^k}$, one needs to compute $\mu^k_{ij}{=}\displaystyle\frac{D_j^{\pi^k\setminus\{i\}} -D_j^{\pi^k}}{ N_j^{\pi^k\setminus\{i\}} -N_j^{\pi^k}}$ 
\item \label{it:mu^k_i*} compute $\mu^k_i=\displaystyle\argmin_{j\in[n]:N_j^{\pi^k\setminus\{i\}}\neq N_j^{\pi^k}} \mu^k_{ij}$.
\end{enumerate}
From \citep[Theorem 2]{akbarzadeh2020conditions}, in an indexable problem, for state $\sigma^k$, there exists a state $j\in[n]$ such that $N_j^{\pi^k\setminus\{\sigma^k\}}\neq N_j^{\pi^k}$. Now, suppose that for any active state $i\in\pi^k$, there exists $j\in[n]$ such that $N_j^{\pi^k\setminus\{i\}}\neq N_j^{\pi^k}$. Using the Sherman-Morrison formula, we have\footnote{the expression of $D^{\pi\setminus\{i\}}$ and $N^{\pi\setminus\{i\}}$ given by Equation~$(18)$ in \citep{akbarzadeh2020conditions} are erroneous.}
\begin{align*}
    &D^{\pi^k\setminus\{i\}} -D^{\pi^k} = -\frac{(1-\beta)\delta_i +\tilde{\mDelta}_iD^{\pi^k}}{1+\tilde{\mDelta}_i[(\mI-\beta \mP^{\pi^k})^{-1}]_{: i}}[(\mI-\beta \mP^{\pi^k})^{-1}]_{: i}, \\
    &\quad \text{and}\quad N^{\pi^k\setminus\{i\}} -N^{\pi^k} = -\frac{(1-\beta) +\tilde{\mDelta}_iN^{\pi^k}}{1+\tilde{\mDelta}_i[(\mI-\beta \mP^{\pi^k})^{-1}]_{: i}}[(\mI-\beta \mP^{\pi^k})^{-1}]_{: i}.
\end{align*}
Then, for any $j\in[n]$ such that $N_j^{\pi^k\setminus\{i\}}\neq N_j^{\pi^k}$, $\mu^k_{ij}=\displaystyle\frac{(1-\beta)\delta_i +\tilde{\mDelta}_iD^{\pi^k}}{(1-\beta) +\tilde{\mDelta}_iN^{\pi^k}}$ which does not depend on $j$.
Then, we simply have ${\mu^k_i=\displaystyle\frac{(1-\beta)\delta_i +\tilde{\mDelta}_iD^{\pi^k}}{(1-\beta) +\tilde{\mDelta}_iN^{\pi^k}}}$. 
Also, we have 
\begin{align*}
    \tilde{\mDelta}_iN^{\pi^k}&=(1-\beta)\tilde{\mDelta}_i(\mI-\beta \mP^{\pi^k})^{-1}\vpi^k=-(1-\beta)y^k_i \quad\text{and}\\
    \tilde{\mDelta}_iD^{\pi^k}&={(1-\beta)\tilde{\mDelta}_i\vu^{\pi^k}(\mu^{k-1}_{\min}) +\mu^{k-1}_{\min}\tilde{\mDelta}_i N^{\pi^k}} =(1-\beta)z^{k-1}_i -(1-\beta)\mu^{k-1}_{\min}y^k_i.
\end{align*}
So, replacing these terms in $\mu^k_i$, we get the formula in Equation~\eqref{eq:mu^k_i_from_y} of our work.

Note that the algorithm of \citep{akbarzadeh2020conditions} was only developed for the discounted case.  Our approach for the time-average reward case is different because we use the active advantage function defined in \eqref{eq:advantage} instead of working with the expected discounted total reward $D^{\pi^k}$ and total number of activations $N^{\pi^k}$ under policy $\pi^k$. 
\reviewedFirstRev{Note that the counterpart of $D^{\pi^k}$ in undiscounted MDP is the average reward $g^{\pi^k}$ and as we have seen in Appendix~\ref{apx:discussion_index}, utilizing average reward optimality is not rich enough for undiscounted MDPs with transient states.}
In addition, our code is also optimized to avoid unnecessary computation and to reduce memory usage. Finally, the way we do the update of our matrix $\mX$ makes it possible to obtain a subcubic algorithm whereas their approach does not (see also below).

\subsection{Comparison with the algorithm of \texorpdfstring{\citep{nino2020fast}}{Nino Mora 2020}}
The algorithm \citep{nino2020fast} has the best complexity up to date for discounted restless bandit. There is a square matrix $\bff{A}$ that plays a similar role as the square matrix $\mX$ in our proposed algorithm. The most costly operations in the algorithm of \citep{nino2020fast} is to update their matrix $\bff{A}$ at each iteration and it is done by Equation~\eqref{eq:update_X_naif} that we recall here (using the same notation $\mA$ as the cited paper):
\begin{align}
    \label{eq:15}
    \text{for } i,j\in\pi^k,\ \bff{A}^{k+1}_{ij}= \bff{A}^{k}_{ij} -\frac{\bff{A}^{k}_{i\sigma^k}}{\bff{A}^{k}_{\sigma^k\sigma^k}}\bff{A}^{k}_{\sigma^kj}.
\end{align}
This incurs a total complexity of $(2/3)n^3+O(n^2)$ arithmetic operations.
As mentioned in Section~\ref{ssec:two_third_algo}, if we updated $\mX^{k+1}$ as given by \eqref{eq:update_X_naif}, our algorithm would also have a $(2/3)n^3+O(n^2)$ complexity but this version of update cannot be optimized by using  fast matrix multiplication.

\subsection{Their approach cannot be directly transformed into a subcubic algorithm}
\label{apx:eq_15}

In addition to all the previously cited differences, one of the major contribution of our algorithm with respect to \citep{akbarzadeh2020conditions,nino2020fast} is that the most advanced version of our algorithm runs in a subcubic time. The approach\footnote{Equation~(18) of \citep{akbarzadeh2020conditions}, which is central to their algorithm is the same as the above equation \eqref{eq:15}.} used  in \citep{akbarzadeh2020conditions,nino2020fast} is to update the full matrix $\mX^{k+1}$ at iteration $k$, by using \eqref{eq:15}.  This idea is represented in Figure~\ref{fig:apx_fast_mm}(a): for a given $\ell$, their algorithm compute $\mX^\ell_{:\sigma^k}$ for all $k$ (i.e., the full vertical lines represented by arrows).  Our first Subroutine~\ref{algo:update_X} uses an horizontal approach based on \eqref{eq:X^ell+1-i}, which we recall here:
\begin{align*}
    X^{\ell+1}_{i\sigma^k} &= X^{\ell}_{i\sigma^k} - \frac{X^{\ell}_{i\sigma^\ell}}{1+X^{\ell}_{{\sigma^\ell\sigma^\ell}}} X^{\ell}_{{\sigma^\ell} \sigma^k}.
\end{align*}
At iteration $k+1$, we use $\mX^{1}_{: \sigma^k}$ to compute all values of $\mX^\ell_{:\sigma^k}$ up to $\ell=k+1$.  This is represented in Figure~\ref{fig:apx_fast_mm}(b).  Our approach can be used to obtain the subcubic algorithm illustrated in Figure~\ref{fig:apx_fast_mm}(c) by using subcubic algorithms for multiplication.

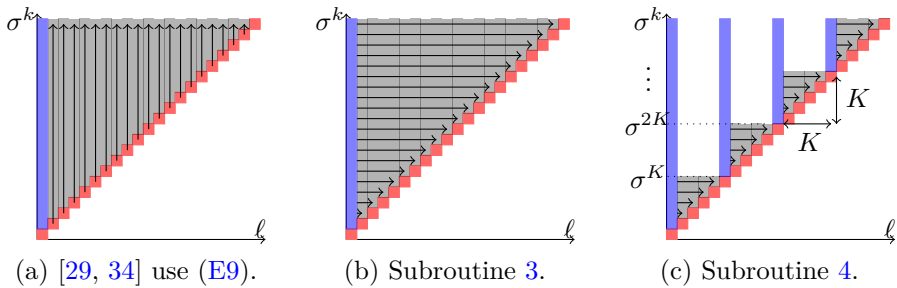
\begin{figure}[ht]
    \centering
    \begin{tabular}{ccc}
        \begin{tikzpicture}[scale=0.14, shorten >=2pt]
            \draw (0,0) edge[->] (22,0);
            \draw (0,0) edge[->] (0,22);
            \node at (21,1) {$\ell$};
            \node at (-1.5,21) {$\sigma^k$};
            \foreach \i in {0,...,20} \draw[fill,red!60] (\i,\i) rectangle (\i+1,\i+1);
            \foreach \i in {0}{
               \draw[fill,blue,opacity=0.5] (\i,\i+1) rectangle (\i+1,21);
               \foreach \j in {1,...,19}{
                   \fill[black,opacity=0.3] (\i+\j,\i+\j+1) rectangle (\i+\j+1,\i+21);
               }
            }
            \foreach \i in {1,...,19}{
                \draw (\i+0.5,\i+0.5) edge[->] (\i+0.5,21);
            }
        \end{tikzpicture}
        &
        \begin{tikzpicture}[scale=0.14, shorten >=2pt]
            \draw (0,0) edge[->] (22,0);
            \draw (0,0) edge[->] (0,22);
            \node at (21,1) {$\ell$};
            \node at (-1.5,21) {$\sigma^k$};
            \foreach \i in {0,...,20} \draw[fill,red!60] (\i,\i) rectangle (\i+1,\i+1);
            \foreach \i in {0}{
               \draw[fill,blue,opacity=0.5] (\i,\i+1) rectangle (\i+1,21);
               \foreach \j in {1,...,19}{
                   \fill[black,opacity=0.3] (\i+\j,\i+\j+1) rectangle (\i+\j+1,\i+21);
               }
            }
            \foreach \i in {2,...,20}{
                \draw (1,\i+0.5) edge[->] (\i,\i+0.5);
            }
        \end{tikzpicture}
        &
        \begin{tikzpicture}[scale=0.14, shorten >=2pt]
            \draw (0,0) edge[->] (22,0);
            \draw (0,0) edge[->] (0,22);
            \node at (21,1) {$\ell$};
            \node at (-1.7,21) {$\sigma^k$};
            \foreach \i in {0,...,20} \draw[fill,red!60] (\i,\i) rectangle (\i+1,\i+1);
            \foreach \i in {0,5,10,15}{
               \draw[fill,blue,opacity=0.5] (\i,\i+1) rectangle (\i+1,21);
               \foreach \j in {1,...,5}{
                   \fill[black,opacity=0.3] (\i+\j,\i+\j+1) rectangle (\i+\j+1,\i+6);
                }
                \foreach \j in {2,...,5}{
                    \draw (\i+1,\i+\j+0.5) edge[->] (\i+\j,\i+\j+0.5);
                }
            }
            \draw (11,11) edge[<->] node[below] {$K$} (16,11);
            \draw (16,11) edge[<->] node[right] {$K$} (16,16);
            \node at (-1.7,6) {$\sigma^K$}; \draw[dotted] (-.5,6) -- (5.5,6);
            \node at (-1.7,11) {$\sigma^{2K}$}; \draw[dotted] (-.5,11) -- (10,11);
            \node at (-1.7,16) {$\vdots$};
        \end{tikzpicture} \\
        (a) \citep{akbarzadeh2020conditions,nino2020fast} use \eqref{eq:15}.
        &(b) Subroutine~\ref{algo:update_X}.
        &(c) Subroutine~\ref{algo:FMM}.
    \end{tabular}
    \caption{Comparison of the computation load of \eqref{eq:15} used in \citep{akbarzadeh2020conditions,nino2020fast} with the one of Subroutine~\ref{algo:update_X} and Subroutine~\ref{algo:FMM}. }
    \label{fig:apx_fast_mm}
\end{figure}

This  leads to the next fundamental  question: {why should the computation of Whittle index be harder than matrix inversion (or multiplication)?} To us, the main difference is that when computing Whittle indices, the permutation $\sigma$ is not known a priori but discovered as the algorithm progresses: $\sigma^k$ is only known at iteration $k$. Hence, while all terms of the matrices $X^k_{ij}$ are not needed, it is difficult to know a priori which ones are needed and which ones are not. Hence, a simple divide and conquer algorithm cannot be used. This is why when recomputing $\mX^{k+1}$ in Subroutine~\ref{algo:FMM}, we recompute the whole matrix (the vertical blue line) and not just the part that  will be used to compute the gray zone: we do not know \emph{a priori} what part of $X^{k+1}_{ij}$ will be useful or not.

\end{appendices}


\bibliography{refWhittleIndex.bib}


\begin{thebibliography}{46}
\ifx \bisbn   \undefined \def \bisbn  #1{ISBN #1}\fi
\ifx \binits  \undefined \def \binits#1{#1}\fi
\ifx \bauthor  \undefined \def \bauthor#1{#1}\fi
\ifx \batitle  \undefined \def \batitle#1{#1}\fi
\ifx \bjtitle  \undefined \def \bjtitle#1{#1}\fi
\ifx \bvolume  \undefined \def \bvolume#1{\textbf{#1}}\fi
\ifx \byear  \undefined \def \byear#1{#1}\fi
\ifx \bissue  \undefined \def \bissue#1{#1}\fi
\ifx \bfpage  \undefined \def \bfpage#1{#1}\fi
\ifx \blpage  \undefined \def \blpage #1{#1}\fi
\ifx \burl  \undefined \def \burl#1{\textsf{#1}}\fi
\ifx \doiurl  \undefined \def \doiurl#1{\url{https://doi.org/#1}}\fi
\ifx \betal  \undefined \def \betal{\textit{et al.}}\fi
\ifx \binstitute  \undefined \def \binstitute#1{#1}\fi
\ifx \binstitutionaled  \undefined \def \binstitutionaled#1{#1}\fi
\ifx \bctitle  \undefined \def \bctitle#1{#1}\fi
\ifx \beditor  \undefined \def \beditor#1{#1}\fi
\ifx \bpublisher  \undefined \def \bpublisher#1{#1}\fi
\ifx \bbtitle  \undefined \def \bbtitle#1{#1}\fi
\ifx \bedition  \undefined \def \bedition#1{#1}\fi
\ifx \bseriesno  \undefined \def \bseriesno#1{#1}\fi
\ifx \blocation  \undefined \def \blocation#1{#1}\fi
\ifx \bsertitle  \undefined \def \bsertitle#1{#1}\fi
\ifx \bsnm \undefined \def \bsnm#1{#1}\fi
\ifx \bsuffix \undefined \def \bsuffix#1{#1}\fi
\ifx \bparticle \undefined \def \bparticle#1{#1}\fi
\ifx \barticle \undefined \def \barticle#1{#1}\fi
\bibcommenthead
\ifx \bconfdate \undefined \def \bconfdate #1{#1}\fi
\ifx \botherref \undefined \def \botherref #1{#1}\fi
\ifx \url \undefined \def \url#1{\textsf{#1}}\fi
\ifx \bchapter \undefined \def \bchapter#1{#1}\fi
\ifx \bbook \undefined \def \bbook#1{#1}\fi
\ifx \bcomment \undefined \def \bcomment#1{#1}\fi
\ifx \oauthor \undefined \def \oauthor#1{#1}\fi
\ifx \citeauthoryear \undefined \def \citeauthoryear#1{#1}\fi
\ifx \endbibitem  \undefined \def \endbibitem {}\fi
\ifx \bconflocation  \undefined \def \bconflocation#1{#1}\fi
\ifx \arxivurl  \undefined \def \arxivurl#1{\textsf{#1}}\fi
\csname PreBibitemsHook\endcsname

\bibitem{gittinsBanditProcessesDynamic1979a}
\begin{barticle}
\bauthor{\bsnm{Gittins}, \binits{J.C.}}:
\batitle{Bandit processes and dynamic allocation indices}.
\bjtitle{Journal of the Royal Statistical Society: Series B (Methodological)}
\bvolume{41}(\bissue{2}),
\bfpage{148}--\blpage{164}
(\byear{1979})
\end{barticle}
\endbibitem

\bibitem{papadimitriou1994complexity}
\begin{bchapter}
\bauthor{\bsnm{Papadimitriou}, \binits{C.H.}},
\bauthor{\bsnm{Tsitsiklis}, \binits{J.N.}}:
\bctitle{The complexity of optimal queueing network control}.
In: \bbtitle{Proceedings of IEEE 9th Annual Conference on Structure in
  Complexity Theory},
pp. \bfpage{318}--\blpage{322}
(\byear{1994}).
\bcomment{IEEE}
\end{bchapter}
\endbibitem

\bibitem{whittle1988restless}
\begin{barticle}
\bauthor{\bsnm{Whittle}, \binits{P.}}:
\batitle{Restless bandits: Activity allocation in a changing world}.
\bjtitle{Journal of applied probability}
\bvolume{25}(\bissue{A}),
\bfpage{287}--\blpage{298}
(\byear{1988})
\end{barticle}
\endbibitem

\bibitem{verloop2016asymptotically}
\begin{barticle}
\bauthor{\bsnm{Verloop}, \binits{I.M.}}:
\batitle{Asymptotically optimal priority policies for indexable and
  nonindexable restless bandits}.
\bjtitle{The Annals of Applied Probability}
\bvolume{26}(\bissue{4}),
\bfpage{1947}--\blpage{1995}
(\byear{2016})
\end{barticle}
\endbibitem

\bibitem{lott2000optimality}
\begin{barticle}
\bauthor{\bsnm{Lott}, \binits{C.}},
\bauthor{\bsnm{Teneketzis}, \binits{D.}}:
\batitle{On the optimality of an index rule in multichannel allocation for
  single-hop mobile networks with multiple service classes}.
\bjtitle{Probability in the Engineering and Informational Sciences}
\bvolume{14}(\bissue{3}),
\bfpage{259}--\blpage{297}
(\byear{2000})
\end{barticle}
\endbibitem

\bibitem{weber1990index}
\begin{botherref}
\oauthor{\bsnm{Weber}, \binits{R.R.}},
\oauthor{\bsnm{Weiss}, \binits{G.}}:
On an index policy for restless bandits.
Journal of applied probability,
637--648
(1990)
\end{botherref}
\endbibitem

\bibitem{glazebrook2006some}
\begin{barticle}
\bauthor{\bsnm{Glazebrook}, \binits{K.D.}},
\bauthor{\bsnm{Ruiz-Hernandez}, \binits{D.}},
\bauthor{\bsnm{Kirkbride}, \binits{C.}}:
\batitle{Some indexable families of restless bandit problems}.
\bjtitle{Advances in Applied Probability}
\bvolume{38}(\bissue{3}),
\bfpage{643}--\blpage{672}
(\byear{2006})
\end{barticle}
\endbibitem

\bibitem{ansell2003whittle}
\begin{barticle}
\bauthor{\bsnm{Ansell}, \binits{P.}},
\bauthor{\bsnm{Glazebrook}, \binits{K.D.}},
\bauthor{\bsnm{Nino-Mora}, \binits{J.}},
\bauthor{\bsnm{O'Keeffe}, \binits{M.}}:
\batitle{Whittle's index policy for a multi-class queueing system with convex
  holding costs}.
\bjtitle{Mathematical Methods of Operations Research}
\bvolume{57}(\bissue{1}),
\bfpage{21}--\blpage{39}
(\byear{2003})
\end{barticle}
\endbibitem

\bibitem{glazebrook2002index}
\begin{barticle}
\bauthor{\bsnm{Glazebrook}, \binits{K.}},
\bauthor{\bsnm{Mitchell}, \binits{H.}}:
\batitle{An index policy for a stochastic scheduling model with
  improving/deteriorating jobs}.
\bjtitle{Naval Research Logistics (NRL)}
\bvolume{49}(\bissue{7}),
\bfpage{706}--\blpage{721}
(\byear{2002})
\end{barticle}
\endbibitem

\bibitem{aalto2019whittle}
\begin{barticle}
\bauthor{\bsnm{Aalto}, \binits{S.}},
\bauthor{\bsnm{Lassila}, \binits{P.}},
\bauthor{\bsnm{Taboada}, \binits{I.}}:
\batitle{Whittle index approach to opportunistic scheduling with partial
  channel information}.
\bjtitle{Performance Evaluation}
\bvolume{136},
\bfpage{102052}
(\byear{2019})
\end{barticle}
\endbibitem

\bibitem{liu2010indexability}
\begin{barticle}
\bauthor{\bsnm{Liu}, \binits{K.}},
\bauthor{\bsnm{Zhao}, \binits{Q.}}:
\batitle{Indexability of restless bandit problems and optimality of {W}hittle
  index for dynamic multichannel access}.
\bjtitle{IEEE Transactions on Information Theory}
\bvolume{56}(\bissue{11}),
\bfpage{5547}--\blpage{5567}
(\byear{2010})
\end{barticle}
\endbibitem

\bibitem{avrachenkov2022whittle}
\begin{barticle}
\bauthor{\bsnm{Avrachenkov}, \binits{K.E.}},
\bauthor{\bsnm{Borkar}, \binits{V.S.}}:
\batitle{Whittle index based {Q}-learning for restless bandits with average
  reward}.
\bjtitle{Automatica}
\bvolume{139},
\bfpage{110186}
(\byear{2022})
\end{barticle}
\endbibitem

\bibitem{nino2014dynamic}
\begin{bchapter}
\bauthor{\bsnm{Ni{\~n}o-Mora}, \binits{J.}}:
\bctitle{A dynamic page-refresh index policy for web crawlers}.
In: \bbtitle{International Conference on Analytical and Stochastic Modeling
  Techniques and Applications},
pp. \bfpage{46}--\blpage{60}
(\byear{2014}).
\bcomment{Springer}
\end{bchapter}
\endbibitem

\bibitem{avrachenkov2013congestion}
\begin{barticle}
\bauthor{\bsnm{Avrachenkov}, \binits{K.}},
\bauthor{\bsnm{Ayesta}, \binits{U.}},
\bauthor{\bsnm{Doncel}, \binits{J.}},
\bauthor{\bsnm{Jacko}, \binits{P.}}:
\batitle{Congestion control of tcp flows in internet routers by means of index
  policy}.
\bjtitle{Computer Networks}
\bvolume{57}(\bissue{17}),
\bfpage{3463}--\blpage{3478}
(\byear{2013})
\end{barticle}
\endbibitem

\bibitem{avrachenkov2018impulsive}
\begin{bchapter}
\bauthor{\bsnm{Avrachenkov}, \binits{K.}},
\bauthor{\bsnm{Piunovskiy}, \binits{A.}},
\bauthor{\bsnm{Zhang}, \binits{Y.}}:
\bctitle{Impulsive control for {G-AIMD} dynamics with relaxed and hard
  constraints}.
In: \bbtitle{2018 IEEE Conference on Decision and Control (CDC)},
pp. \bfpage{880}--\blpage{887}
(\byear{2018}).
\bcomment{IEEE}
\end{bchapter}
\endbibitem

\bibitem{scully2018soap}
\begin{barticle}
\bauthor{\bsnm{Scully}, \binits{Z.}},
\bauthor{\bsnm{Harchol-Balter}, \binits{M.}},
\bauthor{\bsnm{Scheller-Wolf}, \binits{A.}}:
\batitle{{SOAP}: One clean analysis of all age-based scheduling policies}.
\bjtitle{Proceedings of the ACM on Measurement and Analysis of Computing
  Systems}
\bvolume{2}(\bissue{1}),
\bfpage{1}--\blpage{30}
(\byear{2018})
\end{barticle}
\endbibitem

\bibitem{aalto2011properties}
\begin{barticle}
\bauthor{\bsnm{Aalto}, \binits{S.}},
\bauthor{\bsnm{Ayesta}, \binits{U.}},
\bauthor{\bsnm{Righter}, \binits{R.}}:
\batitle{Properties of the {G}ittins index with application to optimal
  scheduling}.
\bjtitle{Probability in the Engineering and Informational Sciences}
\bvolume{25}(\bissue{3}),
\bfpage{269}--\blpage{288}
(\byear{2011})
\end{barticle}
\endbibitem

\bibitem{aalto2009gittins}
\begin{barticle}
\bauthor{\bsnm{Aalto}, \binits{S.}},
\bauthor{\bsnm{Ayesta}, \binits{U.}},
\bauthor{\bsnm{Righter}, \binits{R.}}:
\batitle{On the {G}ittins index in the {M/G/1} queue}.
\bjtitle{Queueing Systems}
\bvolume{63}(\bissue{1-4}),
\bfpage{437}
(\byear{2009})
\end{barticle}
\endbibitem

\bibitem{borkar2017whittle}
\begin{botherref}
\oauthor{\bsnm{Borkar}, \binits{V.S.}},
\oauthor{\bsnm{Pattathil}, \binits{S.}}:
Whittle indexability in egalitarian processor sharing systems.
Annals of Operations Research,
1--21
(2017)
\end{botherref}
\endbibitem

\bibitem{larranaga2015asymptotically}
\begin{barticle}
\bauthor{\bsnm{Larra{\~n}aga}, \binits{M.}},
\bauthor{\bsnm{Ayesta}, \binits{U.}},
\bauthor{\bsnm{Verloop}, \binits{I.M.}}:
\batitle{Asymptotically optimal index policies for an abandonment queue with
  convex holding cost}.
\bjtitle{Queueing systems}
\bvolume{81}(\bissue{2}),
\bfpage{99}--\blpage{169}
(\byear{2015})
\end{barticle}
\endbibitem

\bibitem{archibald2009indexability}
\begin{barticle}
\bauthor{\bsnm{Archibald}, \binits{T.W.}},
\bauthor{\bsnm{Black}, \binits{D.}},
\bauthor{\bsnm{Glazebrook}, \binits{K.D.}}:
\batitle{Indexability and index heuristics for a simple class of inventory
  routing problems}.
\bjtitle{Operations research}
\bvolume{57}(\bissue{2}),
\bfpage{314}--\blpage{326}
(\byear{2009})
\end{barticle}
\endbibitem

\bibitem{glazebrook2009index}
\begin{barticle}
\bauthor{\bsnm{Glazebrook}, \binits{K.D.}},
\bauthor{\bsnm{Kirkbride}, \binits{C.}},
\bauthor{\bsnm{Ouenniche}, \binits{J.}}:
\batitle{Index policies for the admission control and routing of impatient
  customers to heterogeneous service stations}.
\bjtitle{Operations Research}
\bvolume{57}(\bissue{4}),
\bfpage{975}--\blpage{989}
(\byear{2009})
\end{barticle}
\endbibitem

\bibitem{villar2015multi}
\begin{barticle}
\bauthor{\bsnm{Villar}, \binits{S.S.}},
\bauthor{\bsnm{Bowden}, \binits{J.}},
\bauthor{\bsnm{Wason}, \binits{J.}}:
\batitle{Multi-armed bandit models for the optimal design of clinical trials:
  benefits and challenges}.
\bjtitle{Statistical science: a review journal of the Institute of Mathematical
  Statistics}
\bvolume{30}(\bissue{2}),
\bfpage{199}
(\byear{2015})
\end{barticle}
\endbibitem

\bibitem{chen1986linear}
\begin{barticle}
\bauthor{\bsnm{Chen}, \binits{Y.R.}},
\bauthor{\bsnm{Katehakis}, \binits{M.N.}}:
\batitle{Linear programming for finite state multi-armed bandit problems}.
\bjtitle{Mathematics of Operations Research}
\bvolume{11}(\bissue{1}),
\bfpage{180}--\blpage{183}
(\byear{1986})
\end{barticle}
\endbibitem

\bibitem{katehakis1987multi}
\begin{barticle}
\bauthor{\bsnm{Katehakis}, \binits{M.N.}},
\bauthor{\bsnm{Veinott~Jr}, \binits{A.F.}}:
\batitle{The multi-armed bandit problem: decomposition and computation}.
\bjtitle{Mathematics of Operations Research}
\bvolume{12}(\bissue{2}),
\bfpage{262}--\blpage{268}
(\byear{1987})
\end{barticle}
\endbibitem

\bibitem{nino20072}
\begin{barticle}
\bauthor{\bsnm{Ni{\~n}o-Mora}, \binits{J.}}:
\batitle{A $(2/3)n^3$ fast-pivoting algorithm for the {G}ittins index and
  optimal stopping of a markov chain}.
\bjtitle{INFORMS Journal on Computing}
\bvolume{19}(\bissue{4}),
\bfpage{596}--\blpage{606}
(\byear{2007})
\end{barticle}
\endbibitem

\bibitem{sonin2008generalized}
\begin{barticle}
\bauthor{\bsnm{Sonin}, \binits{I.M.}}:
\batitle{A generalized {G}ittins index for a markov chain and its recursive
  calculation}.
\bjtitle{Statistics \& Probability Letters}
\bvolume{78}(\bissue{12}),
\bfpage{1526}--\blpage{1533}
(\byear{2008})
\end{barticle}
\endbibitem

\bibitem{chakravorty2014multi}
\begin{barticle}
\bauthor{\bsnm{Chakravorty}, \binits{J.}},
\bauthor{\bsnm{Mahajan}, \binits{A.}}:
\batitle{Multi-armed bandits, {G}ittins index, and its calculation}.
\bjtitle{Methods and applications of statistics in clinical trials: Planning,
  analysis, and inferential methods}
\bvolume{2}(\bissue{416-435}),
\bfpage{455}
(\byear{2014})
\end{barticle}
\endbibitem

\bibitem{nino2020fast}
\begin{barticle}
\bauthor{\bsnm{Ni{\~n}o-Mora}, \binits{J.}}:
\batitle{A fast-pivoting algorithm for {W}hittle’s restless bandit index}.
\bjtitle{Mathematics}
\bvolume{8}(\bissue{12}),
\bfpage{2226}
(\byear{2020})
\end{barticle}
\endbibitem

\bibitem{nino2010characterization}
\begin{bchapter}
\bauthor{\bsnm{Ni{\~n}o-Mora}, \binits{J.}}:
\bctitle{Characterization and computation of restless bandit marginal
  productivity indices}.
In: \bbtitle{1st International ICST Workshop on Tools for Solving Structured
  Markov Chains}
(\byear{2010})
\end{bchapter}
\endbibitem

\bibitem{akbarzadeh2019restless}
\begin{bchapter}
\bauthor{\bsnm{Akbarzadeh}, \binits{N.}},
\bauthor{\bsnm{Mahajan}, \binits{A.}}:
\bctitle{Restless bandits with controlled restarts: Indexability and
  computation of {W}hittle index}.
In: \bbtitle{2019 IEEE 58th Conference on Decision and Control (CDC)},
pp. \bfpage{7294}--\blpage{7300}
(\byear{2019}).
\bcomment{IEEE}
\end{bchapter}
\endbibitem

\bibitem{akbarzadeh2021maintenance}
\begin{botherref}
\oauthor{\bsnm{Akbarzadeh}, \binits{N.}},
\oauthor{\bsnm{Mahajan}, \binits{A.}}:
Maintenance of a collection of machines under partial observability:
  Indexability and computation of {W}hittle index.
arXiv preprint arXiv:2104.05151
(2021)
\end{botherref}
\endbibitem

\bibitem{nino2007dynamic}
\begin{barticle}
\bauthor{\bsnm{Ni{\~n}o-Mora}, \binits{J.}}:
\batitle{Dynamic priority allocation via restless bandit marginal productivity
  indices}.
\bjtitle{Top}
\bvolume{15}(\bissue{2}),
\bfpage{161}--\blpage{198}
(\byear{2007})
\end{barticle}
\endbibitem

\bibitem{akbarzadeh2020conditions}
\begin{botherref}
\oauthor{\bsnm{Akbarzadeh}, \binits{N.}},
\oauthor{\bsnm{Mahajan}, \binits{A.}}:
Conditions for indexability of restless bandits and an $O(K^3)$ algorithm to
  compute Whittle index.
arXiv
(2020)
\end{botherref}
\endbibitem

\bibitem{gibson2021novel}
\begin{bchapter}
\bauthor{\bsnm{Gibson}, \binits{L.J.}},
\bauthor{\bsnm{Jacko}, \binits{P.}},
\bauthor{\bsnm{Nazarathy}, \binits{Y.}}:
\bctitle{A novel implementation of {Q}-learning for the {W}hittle index}.
In: \bbtitle{EAI International Conference on Performance Evaluation
  Methodologies and Tools},
pp. \bfpage{154}--\blpage{170}
(\byear{2021}).
\bcomment{Springer}
\end{bchapter}
\endbibitem

\bibitem{fu2019towards}
\begin{bchapter}
\bauthor{\bsnm{Fu}, \binits{J.}},
\bauthor{\bsnm{Nazarathy}, \binits{Y.}},
\bauthor{\bsnm{Moka}, \binits{S.}},
\bauthor{\bsnm{Taylor}, \binits{P.G.}}:
\bctitle{Towards {Q}-learning the {W}hittle index for restless bandits}.
In: \bbtitle{2019 Australian \& New Zealand Control Conference (ANZCC)},
pp. \bfpage{249}--\blpage{254}
(\byear{2019}).
\bcomment{IEEE}
\end{bchapter}
\endbibitem

\bibitem{nakhleh2021neurwin}
\begin{botherref}
\oauthor{\bsnm{Nakhleh}, \binits{K.}},
\oauthor{\bsnm{Ganji}, \binits{S.}},
\oauthor{\bsnm{Hsieh}, \binits{P.-C.}},
\oauthor{\bsnm{Hou}, \binits{I.}},
\oauthor{\bsnm{Shakkottai}, \binits{S.}}, et al.:
{NeurWIN}: Neural {W}hittle index network for restless bandits via deep {RL}.
Advances in Neural Information Processing Systems
\textbf{34}
(2021)
\end{botherref}
\endbibitem

\bibitem{ayesta2021computation}
\begin{barticle}
\bauthor{\bsnm{Ayesta}, \binits{U.}},
\bauthor{\bsnm{Gupta}, \binits{M.K.}},
\bauthor{\bsnm{Verloop}, \binits{I.M.}}:
\batitle{On the computation of {W}hittle’s index for markovian restless
  bandits}.
\bjtitle{Mathematical Methods of Operations Research}
\bvolume{93}(\bissue{1}),
\bfpage{179}--\blpage{208}
(\byear{2021})
\end{barticle}
\endbibitem

\bibitem{strassen1969gaussian}
\begin{barticle}
\bauthor{\bsnm{Strassen}, \binits{V.}}:
\batitle{Gaussian elimination is not optimal}.
\bjtitle{Numerische mathematik}
\bvolume{13}(\bissue{4}),
\bfpage{354}--\blpage{356}
(\byear{1969})
\end{barticle}
\endbibitem

\bibitem{putermanMarkovDecisionProcesses1994}
\begin{bbook}
\bauthor{\bsnm{Puterman}, \binits{M.L.}}:
\bbtitle{Markov Decision Processes: Discrete Stochastic Dynamic Programming},
\bedition{1st} edn.
\bpublisher{John Wiley \& Sons, Inc.},
\blocation{USA}
(\byear{1994})
\end{bbook}
\endbibitem

\bibitem{schweitzer1978functional}
\begin{barticle}
\bauthor{\bsnm{Schweitzer}, \binits{P.J.}},
\bauthor{\bsnm{Federgruen}, \binits{A.}}:
\batitle{The functional equations of undiscounted {M}arkov renewal
  programming}.
\bjtitle{Mathematics of Operations Research}
\bvolume{3}(\bissue{4}),
\bfpage{308}--\blpage{321}
(\byear{1978})
\end{barticle}
\endbibitem

\bibitem{gall2018improved}
\begin{bchapter}
\bauthor{\bsnm{Gall}, \binits{F.L.}},
\bauthor{\bsnm{Urrutia}, \binits{F.}}:
\bctitle{Improved rectangular matrix multiplication using powers of the
  {C}oppersmith-{W}inograd tensor}.
In: \bbtitle{Proceedings of the Twenty-Ninth Annual ACM-SIAM Symposium on
  Discrete Algorithms},
pp. \bfpage{1029}--\blpage{1046}
(\byear{2018}).
\bcomment{SIAM}
\end{bchapter}
\endbibitem

\bibitem{woodbury1950inverting}
\begin{botherref}
\oauthor{\bsnm{Woodbury}, \binits{M.A.}}:
Inverting modified matrices.
Statistical Research Group
(1950)
\end{botherref}
\endbibitem

\bibitem{huang2016strassen}
\begin{bchapter}
\bauthor{\bsnm{Huang}, \binits{J.}},
\bauthor{\bsnm{Smith}, \binits{T.M.}},
\bauthor{\bsnm{Henry}, \binits{G.M.}},
\bauthor{\bsnm{Van De~Geijn}, \binits{R.A.}}:
\bctitle{Strassen's algorithm reloaded}.
In: \bbtitle{SC'16: Proceedings of the International Conference for High
  Performance Computing, Networking, Storage and Analysis},
pp. \bfpage{690}--\blpage{701}
(\byear{2016}).
\bcomment{IEEE}
\end{bchapter}
\endbibitem

\bibitem{huang2018practical}
\begin{botherref}
\oauthor{\bsnm{Huang}, \binits{J.}}, et al.:
Practical fast matrix multiplication algorithms.
PhD thesis
(2018)
\end{botherref}
\endbibitem

\bibitem{gast2020exponential}
\begin{botherref}
\oauthor{\bsnm{Gast}, \binits{N.}},
\oauthor{\bsnm{Gaujal}, \binits{B.}},
\oauthor{\bsnm{Yan}, \binits{C.}}:
Exponential convergence rate for the asymptotic optimality of {W}hittle index
  policy.
arXiv preprint arXiv:2012.09064
(2020)
\end{botherref}
\endbibitem

\end{thebibliography}



\end{document}